\documentclass[reqno,11pt,a4paper]{amsart}
\usepackage{amsmath}
\usepackage{amssymb}
\usepackage{amsfonts}
\usepackage[figuresright]{rotating}

\newtheorem{theorem}{Theorem}
\theoremstyle{plain}

\newtheorem{corollary}{Corollary}

\newtheorem{lemma}{Lemma}

\newtheorem{proposition}{Proposition}
\newtheorem{remark}{Remark}

\numberwithin{equation}{section}
\input{tcilatex}

\begin{document}
\title[Inference without smoothing for panels]{Inference without smoothing
for large panels with cross-sectional and temporal dependence}
\author{Javier Hidalgo}
\address{Economics Department\\
London School of Economics\\
Houghton Street \\
London WC2A 2AE\\
UK}
\email{f.j.hidalgo@lse.ac.uk}
\author{Marcia Schafgans}
\address{Economics Department\\
London School of Economics\\
Houghton Street \\
London WC2A 2AE\\
UK}
\email{m.schafgans@lse.ac.uk}

\begin{abstract}
This paper addresses inference in large panel data models in the presence of
both cross-sectional and temporal dependence of unknown form. We are
interested in making inferences that do not rely on the choice of any
smoothing parameter as is the case with the often employed \textquotedblleft 
$HAC$\textquotedblright\ estimator for the covariance matrix. To that end,
we propose a cluster estimator for the asymptotic covariance of the
estimators and valid bootstrap schemes that do not require the selection of
a bandwidth or smoothing parameter and accommodate the nonparametric nature
of both temporal and cross-sectional dependence. Our approach is based on
the observation that the spectral representation of the fixed effect panel
data model is such that the errors become approximately temporally
uncorrelated. Our proposed bootstrap schemes can be viewed as wild
bootstraps in the frequency domain. We present some Monte-Carlo simulations
to shed some light on the small sample performance of our inferential
procedure.


Transformation.\emph{\ }Nonparametric bootstrap algorithms.
\end{abstract}

\thanks{We would like to thank Professor Serena Ng, Silvia Gon\c{c}alves,
two anonymous referees, and participants at the Bristol ESG Meetings for
their helpful comments.}
\maketitle

\section{\textbf{INTRODUCTION}}

Nowadays we often encounter panel data sets where both the number of
individuals, $n$, and the time dimension, $T$, are large or increase without
limit. Phillips and Moon $\left( 1999\right) $\ and Pesaran and Yamagata $%
\left( 2008\right) $ provide some theoretical results for the parameter
estimators in large panel data models, that is where both $n$\ and $T$\ tend
to infinity. These works were done under the assumption of no dependence
among the cross-sectional units. Yet, it is well recognized that the latter
assumption is not very realistic, and there has been a surge of work on how
to provide valid inferences when this type of dependence is present. The
issues are closely related to Zellner's $\left( 1962\right) $ $SURE$
(Seemingly Unrelated Regression) model, be it that here both dimensions are
allowed to increase without limit.

Once one accepts the possibility that the errors of the model may exhibit
cross-sectional and/or temporal dependence, a key component to make valid
inferences is the consistent estimation of the asymptotic covariance matrix
of the estimators. For that purpose, we might proceed by explicitly assuming
some specific dependence structure on the error term. In our context, this
route appears to be quite cumbersome mainly for two reasons. First, it is
quite difficult to specify an appropriate model in the presence of
cross-sectional dependence as there are ample generic models capable to
justify such a dependence. Some examples are the Spatial Autoregressive ($%
SAR $) model of Cliff and Ord $\left( 1973\right) $, which has its origins
in Whittle $\left( 1954\right) $, Andrews' $\left( 2005\right) $\ proposal
to capture common shocks (e.g., macroeconomic, technological,
legal/institutional) across observations and Pesaran's $\left( 2006\right) $
factor model. Second, in many settings it may be quite unrealistic to assume
that the temporal dependence is the same for all individuals, so finding a
correct specification may be infeasible as $n$ increases with no limit.
Inferential properties based on parameter estimates that use a specific
(wrong)\ structure, moreover, may be worse than the least squares estimates\
($LSE$).\ The latter observation was first documented in Engle $\left(
1974\right) $\ and latter examined in Nicholls and Pagan $\left( 1977\right) 
$, who illustrated the adverse consequences of imposing incorrect temporal
dependence assumptions on inferences, say when the practitioner assumes an $%
AR\left( 1\right) $\ model instead of the true underlying $AR\left( 2\right) 
$\ specification.

As the task of finding an appropriate model for the dependence can be very
daunting, one of our main aims in this paper is then to provide inferences
in panel data not only when the error term (potentially) exhibits both
temporal and cross-sectional dependence, but more importantly doing so
without relying on any parametric functional form for such a dependence.
Under these circumstances, one standard methodology is based on the HAC
estimator, whose implementation requires the choice of one (or more)
bandwidth parameter(s).\footnote{%
In a time series regression model context several proposals, both in the
time and frequency domain, have been employed and bootstrap applications
commonly approximate the long term covariance by using a long AR polynomial
(sieve method). Other methods include the use of orthogonal polynomials,
see, e.g., Sun (2013) and Phillips (2005), instead of the use of Fourier
sequences. All of them have in common that they require the choice of a
bandwidth parameter and/or base function. Lazarus et al. (2018) provide an
interesting simulation study.} While this approach is often invoked and used
in the context of time series regression models, application of spatial HAC
estimators is less common. The use of HAC\ estimators in spatial
econometrics was advocated by Conley $\left( 1999\right) $\ and Kelejian and
Prucha $\left( 2007\right) $ studied its use in Cliff-Ord type spatial
models. Recently, a HAC\ estimator accounting for both temporal and spatial
correlation been considered by Kim and Sun $\left( 2013\right) $. The
implementation requires not only the selection of a bandwidth parameter but,
more importantly, an associated measure of distance between the
cross-sectional units. This explicitly assumes that there is some type of
ordering among the individuals or cross-sectional units which, in contrast
with\ the time dimension, is not unambiguous. Even if one accepts the
existence of such an ordering, it is likely that various economics and/or
geographical distance measures, each requiring their own bandwidth, may be
required to encapsulate the order. For instance, simply relying on the
geographic \textquotedblleft as the crow flies\textquotedblright\ distance
measure for ordering is questionable as one cannot expect that two
cross-sectional units located in the Rockies would behave the same as if
they were in the Midwest. Clearly, a distance measure which captures the
topography and other economic measures may be required.\textbf{\ }In
addition, if we recognize that the temporal dependence may not be the same
for all individuals, even the selection of a bandwidth parameter to account
for the temporal dependence may become infeasible. Any cross-validation
algorithm used to determine the bandwidth parameter for temporal dependence
may then need to be performed for each individual.

To deal with the potential caveats of HAC estimators, we shall propose a
cluster based estimator which is able to take into account both types of
dependence and permits the temporal dependence to vary across individuals,
see Condition $C1$ and its discussion, extending the work of Arellano $%
\left( 1987\right) $\ and Driscoll and Kraay $\left( 1998\right) $ in a
substantial way. While Driscoll and Kraay (1998) employed a cluster type of
estimator to account for the cross-sectional dependence, they relied on the
HAC methodology to accommodate the temporal dependence subjecting it to the
drawback mentioned before.\textbf{\ }We avoid the use of the HAC\
methodology altogether\textbf{. }In addition, we provide a new CLT that
accounts for an unknown and general temporal spatial dependence structure
that permits strong spatial dependence. Our approach allows for more general
dependence structure than permitted by Kim and Sun $\left( 2013\right) $ and
Driscoll and Kraay $\left( 1998\right) $. Our new results can therefore be
regarded as providing primitive conditions that guarantee Kim and Sun's and
Driscoll and Kraay's assumption of the existence of a suitable CLT.

Our approach is based on the observation that the spectral representation of
the fixed effect panel data model $\left( \text{\ref{model_1}}\right) $ is
such that the errors become approximately temporally uncorrelated whilst
heteroskedastic.\ It is this observation that enables us to conduct
inference without any smoothing. To provide finite sample improvements for
inference based on our cluster estimator, we present and examine bootstrap
schemes which also do not require the choice of any bandwidth parameter%
\textbf{,} contrary to the sieve or moving block bootstrap (henceforth
denoted MBB)\textbf{. }Two bootstrap algorithms are presented, one where we
assume homogeneous temporal dependence, which we shall denote as the na\"{\i}%
ve bootstrap, and a second one, denoted the wild bootstrap, where we allow
for heterogeneous temporal dependence. Our bootstrap schemes can be viewed
as wild bootstraps in the frequency domain which are shown to have good
finite sample properties.

We compare our proposal to other methods that also do not require any
ordering of the cross-sectional units. In particular, we consider Driscoll
and Kraay's HAC\ estimator and the fixed-b asymptotic framework advocated by
Vogelsang $\left( 2012\right) .$ We also consider the MBB bootstrap applied
to the vector containing all the individual observations at each point in
time as proposed by Gon\c{c}alves $\left( 2011\right) $.

While our estimator does permit more general spatio temporal dependence and
does not require any smoothing parameters, in line with Robinson $\left(
1989\right) $, the approach examined in Section 2 precludes the presence of
conditional heteroskedasticity.\ In Section 4, we examine how we can relax
this by introducing a multiplicative error structure, $v_{pt}=\sigma
_{1}(w_{p})\sigma _{2}\left( \varrho _{t}\right) u_{pt},$\ where $w_{p}$\
and $\varrho _{t}$\ can be functions of the fixed effects and/or variables
which are correlated with the included regressors. It is worth noting that
we do not need to observe these variables, as is the case when $w_{p},$ say,
is the fixed effect. That is, we can allow for \textquotedblleft
groupwise\textquotedblright\ heteroskedasticity and applications in
development economics are commonplace, see Deaton (1996) and Greene (2018).
Of particular interest, here, is the realisation that our cluster based
inference is robust to the presence of heteroskedasticity that is only
cross-sectional in nature (i.e., where $\sigma _{2}\left( \varrho
_{t}\right) $\ is constant). In the presence of a non-constant $\sigma
_{2}\left( \varrho _{t}\right) $, we propose a simple way to robustify our
cluster based inference. Whereas more general forms of heteroskedasticity,
where $v_{pt}=\sigma (x_{pt})u_{pt}$, can be permitted, its implementation
would require the use of nonparametric methods which would require the
selection of a bandwidth parameter to estimate the heteroskedasticity
function.\textbf{\ }We shall indicate how we should proceed if this were the
case.\ Finally, a benefit of our estimator is that it permits the temporal
dependence to vary across individuals, which is more realistic. It is
important to point out that MBB would not be valid in these settings as it
depends on some type of temporal homogeneity or even stationarity.

The remainder of the paper is organized as follows. In the next section we
discuss the regularity conditions for our model and describe the main
results. In Section 3 we introduce our bandwidth parameter free bootstrap
schemes and we demonstrate their validity. Section 4 discusses a
generalisation of our model that permits (conditional) heteroskedasticity.
Section 5 presents a Monte Carlo simulation experiment to shed some light on
the finite sample performance of our cluster estimator and its comparison to
others and we illustrate the finite sample benefits of our bootstrap
schemes. In Section 6 we summarize. The proofs of our main results are given
in Appendix A, which employs a series of lemmas given in Appendix B.

\section{\textbf{THE REGULARITY\ CONDITIONS AND MAIN RESULTS}}

We shall begin by considering the panel data model%
\begin{equation}
y_{pt}=\beta ^{\prime }x_{pt}+\eta _{p}+\alpha _{t}+u_{pt}\text{, \ \ }%
p=1,...,n,\text{ \ \ \ }t=1,...,T\text{,}  \label{model_1}
\end{equation}%
where $\beta $ is a $k\times 1$ vector of unknown parameters, $x_{pt}$ is a $%
k\times 1$ vector of covariates, $\alpha _{t}$ and $\eta _{p}$ represent
respectively the time and individual fixed effects and $\left\{
u_{pt}\right\} _{t\in \mathbb{Z}}$, $p\in \mathbb{N}^{+}$, are sequences of
zero mean errors with heterogeneous variance $E\left( u_{pt}^{2}\right)
=\sigma _{p}^{2}$, $p\in N^{+}$. We allow for general (unknown) temporal and
cross-sectional dependence structures of the sequence $\left\{
u_{pt}\right\} _{t\in \mathbb{Z}}$, $p\in N^{+}$, detailed in Condition $C1$
and $\left\{ x_{pt}\right\} _{t\in \mathbb{Z}}$, $p\in N^{+},$ detailed in
Condition $C2.$ Further details are provided in our discussion of these
conditions below. For simplicity, we shall assume that the sequences $%
\left\{ x_{pt}\right\} _{t\in \mathbb{Z}}$, $p\in N^{+}$, are mutually
independent of the error term $\left\{ u_{pt}\right\} _{t\in \mathbb{Z}}$, $%
p\in N^{+}$, whilst allowing for dependence of the covariates with the fixed
effects $\eta _{p}$\ and/or $\alpha _{t}$.\footnote{\label{FN1} In fact, all
that is needed is that the first conditonal moment of the error is zero and
the second conditonal moment is equal to the unconditonal one.} In Section
4, we shall relax this condition allowing for heteroskedasticity.\textbf{\ }%
A straightforward extension that allows for lagged endogenous variables $%
\left\{ y_{p,t-\ell }\right\} _{\ell =1}^{k_{1}},$\ as in Hidalgo and
Schafgans $\left( 2017\right) $, requires the use of the instrumental
variable estimator, where $\left\{ x_{p,t-\ell }\right\} _{\ell =1}^{k_{1}}$%
\ provide natural instruments for $\left\{ y_{p,t-\ell }\right\} _{\ell
=1}^{k_{1}}$. We have avoided this generalization as it would detract from
the main contribution of the paper and it will only add some extra
technicalities and/or considerations which are well known and understood
when $n=1$.

Our first aim in the paper is to perform inference on the slope parameters $%
\beta $ in the presence of a very general and unknown spatio-temporal
dependence structure. To that end, we first need to extend a Central Limit
Theorem\textbf{\ }provided in Phillips and Moon $\left( 1999\right) $, see
also Hahn and Kuersteiner $\left( 2002\right) $. The reason for this is that
in their work the sequences of random variables, say $\left\{ \psi
_{pt}\right\} _{t\in \mathbb{Z}}$, $p\in \mathbb{N}^{+}$, are assumed to be
independent, that is $\left\{ \psi _{pt}\right\} _{t\in \mathbb{Z}}$ and $%
\left\{ \psi _{qt}\right\} _{t\in \mathbb{Z}}$ are mutually independent for
any $p\neq q$, which is ruled out in our context as we permit
cross-sectional dependence. Moreover, as we shall allow for
\textquotedblleft \emph{strong-dependence}\textquotedblright\ in our error
and regressor sequences,\ we cannot use results and arguments based on any
type of \textquotedblleft \emph{strong}-\emph{mixing}\textquotedblright\
conditions, so that results in Jenish and Prucha $\left( 2009,2012\right) $
cannot be invoked in our framework either. Our theorem also extends the
results provided in Hidalgo and Schafgans $\left( 2017\right) $ by allowing
the errors $u_{pt}$ to exhibit temporal dependence. A second aim of the
paper is to extend the work of Driscoll and Kraay $\left( 1998\right) $ by
examining, in the presence of individual and temporal fixed effects, a
cluster estimator of the asymptotic covariance of the estimator of the slope
parameters that does not require the ordering of the observations (in the
cross-sectional dimension) or the selection of a bandwidth parameter.

The fixed effect model and the estimator for the slope parameters we
consider is well known. Denoting for any generic sequence $\left\{ \varsigma
_{pt}\right\} _{t=1}^{T}$,\textbf{\ }$p=1,...,n$\textbf{,} the transformation%
\begin{eqnarray}
\widetilde{\varsigma }_{pt} &=&\varsigma _{pt}-\overline{\varsigma }_{\cdot
t}-\overline{\varsigma }_{p\cdot }+\overline{\overline{\varsigma }}_{\cdot
\cdot }\text{;}  \label{not} \\
\overline{\varsigma }_{\cdot t} &=&\frac{1}{n}\sum_{p=1}^{n}\varsigma _{pt};%
\text{ }\overline{\varsigma }_{p\cdot }=\frac{1}{T}\sum_{t=1}^{T}\varsigma
_{pt};\text{ and }\overline{\overline{\varsigma }}_{\cdot \cdot }=\frac{1}{nT%
}\sum_{t=1}^{T}\sum_{p=1}^{n}\varsigma _{pt}\text{,}  \notag
\end{eqnarray}%
the estimator of $\beta $ is obtained by performing least squares on the
transformed model (where the individual and time effects are removed)%
\begin{equation}
\widetilde{y}_{pt}=\beta ^{\prime }\widetilde{x}_{pt}+\widetilde{u}_{pt}%
\text{, \ }p=1,...,n\text{ \ and }t=1,...,T,  \label{model_1R}
\end{equation}%
so that $\widehat{\beta }$ is defined as%
\begin{equation}
\widehat{\beta }=\left( \sum_{p=1}^{n}\sum_{t=1}^{T}\widetilde{x}_{pt}%
\widetilde{x}_{pt}^{\prime }\right) ^{-1}\left( \sum_{p=1}^{n}\sum_{t=1}^{T}%
\widetilde{x}_{pt}\widetilde{y}_{pt}\right) \text{.}  \label{beta_fe}
\end{equation}%
It is obvious that we can take $Ex_{pt}=0$\ as $\widetilde{x}_{pt}$\ is
invariant to additive constants, say $\mu _{t}$\ or $\nu _{p}$, to $x_{pt}$.

In this paper, we shall focus on an equivalent frequency domain formulation
of $\left( \text{\ref{model_1}}\right) $ and $\left( \text{\ref{model_1R}}%
\right) $. It is the application of the Discrete Fourier Transform (DFT) to
our model, as will become clear shortly, that plays an important role in
describing and motivating the cluster estimator of the asymptotic covariance
matrix of $\widehat{\beta },$ or equivalently $\tilde{\beta}$\ given in $%
\left( \text{\ref{beta_fef}}\right) $\ below, and the bootstrap schemes
described in Section 3.

For this purpose, we denote the DFT for generic sequences $\left\{ \varsigma
_{pt}\right\} _{t=1}^{T}$, $p\geq 1$, by%
\begin{equation}
\mathcal{J}_{\varsigma ,p}\left( \lambda _{j}\right) =\frac{1}{T^{1/2}}%
\sum_{t=1}^{T}\varsigma _{pt}e^{-it\lambda _{j}}\text{, \ \ }j=1,...,%
\widetilde{T}=\left[ T/2\right] \text{, \ }\lambda _{j}=\frac{2\pi j}{T}
\label{dft}
\end{equation}%
and $\mathcal{J}_{\varsigma ,p}\left( \lambda _{j}\right) =\mathcal{J}%
_{\varsigma ,p}\left( -\lambda _{T-j}\right) $, $j=\widetilde{T}+1,..,T.$ We
can then rewrite $\left( \text{\ref{model_1R}}\right) $ as%
\begin{equation}
\mathcal{J}_{\widetilde{y},p}\left( \lambda _{j}\right) =\beta ^{\prime }%
\mathcal{J}_{\widetilde{x},p}\left( \lambda _{j}\right) +\mathcal{J}_{%
\widetilde{u},p}\left( \lambda _{j}\right) \text{, \ }p=1,...,n\text{; \ }%
j=1,...,T-1.  \label{model_4freq}
\end{equation}%
Given that our sequences $\left\{ \widetilde{\varsigma }_{pt}\right\}
_{t=1}^{T},$\ $p\geq 1,$ are centered around their sample means, we can
leave out the frequency $\lambda _{j}$ for $j=T$ (and 0) as $J_{\widetilde{%
\varsigma },p}\left( 0\right) =\frac{1}{T^{1/2}}\sum_{t=1}^{T}\widetilde{%
\varsigma }_{pt}=0$. The interesting property of $\mathcal{J}_{\widetilde{u}%
,p}\left( \lambda _{j}\right) ,$ $j=1,...,T-1,$ that allows us to formulate
our new cluster estimator that accounts for both types of dependence, is
that it is serially uncorrelated over the Fourier frequencies for large $T,$%
\ whilst possibly heteroskedastic. Based on the frequency domain formulation
of our model $\left( \text{\ref{model_4freq}}\right) $, we can also compute
our estimator of $\beta $ as%
\begin{equation}
\widetilde{\beta }=\left( \sum_{p=1}^{n}\sum_{j=1}^{T-1}\mathcal{J}_{%
\widetilde{x},p}\left( \lambda _{j}\right) \mathcal{J}_{\widetilde{x}%
,p}^{\prime }\left( -\lambda _{j}\right) \right) ^{-1}\left(
\sum_{p=1}^{n}\sum_{j=1}^{T-1}\mathcal{J}_{\widetilde{x},p}\left( \lambda
_{j}\right) \mathcal{J}_{\widetilde{y},p}\left( -\lambda _{j}\right) \right)
.  \label{beta_fef}
\end{equation}

We introduce our regularity conditions next. To that end, and in what
follows, we denote for any generic sequence $\left\{ v_{pt}\right\} _{t\in 
\mathbb{Z}}$, $p\in N,$ 
\begin{equation*}
\varphi _{v}\left( p,q\right) =\text{\textsl{Cov}}\left(
v_{pt};v_{qt}\right) \text{, for any }p,q\geq 1\text{.}
\end{equation*}

\begin{description}
\item[Condition $\mathbf{C1}$] $\left\{ u_{pt}\right\} _{t\in \mathbb{Z}}$%
\emph{, }$p\in \mathbb{N}^{+}$\emph{, are zero mean sequences of random
variables such that}\newline
\noindent $\left( \mathbf{i}\right) \qquad \qquad \qquad 
\begin{array}[t]{c}
u_{pt}=\dsum\limits_{k=0}^{\infty }d_{k}\left( p\right) \xi _{p,t-k}\text{,
\ \ \ \ }\dsum\limits_{k=0}^{\infty }kd_{k}<\infty \text{,~\ \ \ }%
d_{k}=:\sup_{p}\left\vert d_{k}\left( p\right) \right\vert \text{,}%
\end{array}%
$\newline
\emph{where\ }$E\left( \xi _{pt}\mid \mathcal{V}_{p,t-1}\right) =0$\emph{; }$%
E\left( \xi _{pt}^{2}\mid \mathcal{V}_{p,t-1}\right) =\sigma _{\xi ,p}^{2}$%
\emph{\ and finite fourth moments, with }$\mathcal{V}_{p,t}$\emph{\ denoting
the }$\sigma -$\emph{algebra generated by }$\left\{ \xi _{ps}\text{, }s\leq
t\right\} $\emph{.\newline
}\noindent $\left( \mathbf{ii}\right) $\emph{\ For all }$t\in \mathbb{Z}$ 
\emph{and} $p\in \mathbb{N}^{+}$\emph{, }%
\begin{equation*}
\xi _{pt}=\sum_{\ell =1}^{\infty }a_{\ell }\left( p\right) \varepsilon
_{\ell t}\text{, \ \ \ \ }\sup_{p\in \mathbb{N}^{+}}\sum_{\ell =1}^{\infty
}\left\vert a_{\ell }\left( p\right) \right\vert ^{2}<\infty \text{, \ }%
\sup_{\ell \geq 1}\sum_{p=1}^{n}\left\vert a_{\ell }\left( p\right)
\right\vert ^{2}<\infty \text{,}
\end{equation*}%
\emph{where the sequences }$\left\{ \varepsilon _{\ell t}\right\} _{t\in 
\mathbb{Z}}$\emph{, }$\ell \in \mathbb{N}^{+}$\emph{, are zero mean
independent identically distributed }$\left( iid\right) $ \emph{random
variables. \newline
}\noindent $\left( \mathbf{iii}\right) $\ \emph{The fourth cumulant of }$%
\left\{ u_{pt}\right\} _{t\in \mathbb{Z}}$\emph{, }$p\in \mathbb{N}^{+}$, 
\emph{satisfies} 
\begin{equation*}
\lim_{T\nearrow \infty }\sup_{p\in \mathbb{N}^{+}}%
\sum_{t_{1},t_{2},t_{3}=1}^{T}\left\vert \text{\textsl{Cum}}\left(
u_{pt_{1}};u_{pt_{2}};u_{pt_{3}};u_{p0}\right) \right\vert <\infty \text{.}
\end{equation*}

\item[Condition $\mathbf{C2}$] $\left\{ x_{pt}\right\} _{t\in \mathbb{Z}}$%
\emph{, }$p\in \mathbb{N}^{+}$\emph{, are sequences of random variables such
that:}\newline
\noindent $\left( \mathbf{i}\right) \qquad \qquad \qquad 
\begin{array}[t]{c}
x_{pt}=\dsum\limits_{k=0}^{\infty }c_{k}\left( p\right) \chi _{p,t-k}\text{,
\ \ \ \ }\dsum\limits_{k=0}^{\infty }kc_{k}<\infty \text{,~\ \ }%
c_{k}=:\sup_{p}\left\Vert c_{k}\left( p\right) \right\Vert \text{,}%
\end{array}%
$\newline
\emph{where }$\left\Vert B\right\Vert $\emph{\ denotes the norm of the
matrix }$B$\emph{\ and }$E\left( \chi _{pt}\mid \Upsilon _{p,t-1}\right) =0$%
\emph{; }$Cov\left( \chi _{pt}\mid \Upsilon _{p,t-1}\right) =\Sigma _{\chi
,p}$\emph{\ and }$E\left\Vert \chi _{pt}\right\Vert ^{4}<\infty $\emph{,
with }$\Upsilon _{p,t}$\emph{\ denoting the }$\sigma -$\emph{algebra
generated by }$\left\{ \chi _{ps}\text{, }s\leq t\right\} $\emph{.}\newline
\noindent $\left( \mathbf{ii}\right) $\emph{\ The sequences of random
variables }$\left\{ \chi _{pt}\right\} _{t\in \mathbb{Z}}$\emph{, }$p\in 
\mathbb{N}^{+}$\emph{, are such that }%
\begin{equation*}
\chi _{pt}=\sum_{\ell =1}^{\infty }b_{\ell }\left( p\right) \eta _{\ell t}%
\text{, \ \ \ \ }\sup_{p\in \mathbb{N}^{+}}\sum_{\ell =1}^{\infty
}\left\vert b_{\ell }\left( p\right) \right\vert ^{2}<\infty \text{, \ \ }%
\sup_{\ell \geq 1}\sum_{p=1}^{n}\left\vert b_{\ell }\left( p\right)
\right\vert ^{2}<\infty \text{\emph{,}}
\end{equation*}%
\emph{where the sequences }$\left\{ \eta _{\ell t}\right\} _{t\in \mathbb{Z}%
} $\emph{, }$\ell \in \mathbb{N}^{+}$\emph{, are zero mean iid random
variables.}\newline
\noindent $\left( \mathbf{iii}\right) $\emph{\ Denoting} $\Sigma
_{x,p}=E\left( x_{pt}x_{pt}^{\prime }\right) $\emph{, we have }%
\begin{equation}
0<\Sigma _{x}=\lim_{n\rightarrow \infty }\frac{1}{n}\sum_{p=1}^{n}\Sigma
_{x,p}  \label{sigmax}
\end{equation}%
\emph{and the fourth cumulant of }$\left\{ x_{pt}\right\} _{t\in \mathbb{Z}}$%
\emph{, }$p\in \mathbb{N}^{+}$\emph{, satisfies} 
\begin{equation*}
\lim_{T\rightarrow \infty }\sup_{p\in \mathbb{N}^{+}}%
\sum_{t_{1},t_{2},t_{3}=1}^{T}\left\vert \text{\textsl{Cum}}\left(
x_{pt_{1},a};x_{pt_{2},b};x_{pt_{3},c};x_{p0,d}\right) \right\vert <\infty 
\text{, }a,b,c,d=1,...,k\text{,}
\end{equation*}%
\emph{where }$x_{pt,a}$\emph{\ denotes the }$a-th$\emph{\ element of }$%
x_{pt} $.

\item[Condition $\mathbf{C3}$] \emph{For all }$p\in \mathbb{N}^{+}$\emph{,
the sequences }$\left\{ u_{pt}\right\} _{t\in \mathbb{Z}}$\emph{\ and }$%
\left\{ x_{pt}\right\} _{t\in \mathbb{Z}}$\emph{\ are mutually independent
and}%
\begin{equation}
0<\max_{1\leq p\leq n}\sum_{q=1}^{n}\left\Vert \varphi \left( p,q\right)
\right\Vert <\infty \text{,}  \label{c3}
\end{equation}%
\emph{where }$\varphi \left( p,q\right) :=\varphi _{u}\left( p,q\right)
\varphi _{x}\left( p,q\right) $.
\end{description}

We now comment on our conditions. Conditions $C1$ and $C2$ indicate that $%
\left\{ u_{pt}\right\} _{t\in \mathbb{Z}}$\emph{\ }and\emph{\ }$\left\{
x_{pt}\right\} _{t\in \mathbb{Z}}$, $p\in \mathbb{N}^{+}$, are linear
processes that permit the usual $SAR$ (or more generally $SARMA$) model.
Indeed, by definition of the $SAR$ model, with $W$\ a spatial weight matrix
we have 
\begin{eqnarray*}
u &=&\left( I-\omega W\right) ^{-1}\varepsilon \\
&=&\left( I+\Xi \right) \varepsilon \text{, \ \ }\Xi =\left( \psi _{q}\left(
p\right) \right) _{p,q=1}^{n}\text{,}
\end{eqnarray*}%
so that $u_{p}=\sum_{q=0}^{n}\psi _{q}\left( p\right) \varepsilon _{q}$,
which implies that the $SAR$\ model satisfies Condition $C1$. Unlike the SAR
model, Condition $C1$ does permit the sequence $\sum_{p=1}^{n}\left\vert
a_{\ell }\left( p\right) \right\vert $ to grow with $n$. One can allow the
weights $a_{\ell }\left( p\right) $ to depend on the sample size
\textquotedblleft $n$\textquotedblright\ as is often done in $SAR$ models
with weight matrices $W$ row-normalized, but it does not add anything
significant. Our conditions, therefore, appear to be weaker than those
typically assumed when cross-sectional dependence is allowed while being
similar to those of Lee and Robinson\textbf{\ }$\left( 2013\right) $\textbf{.%
} As the sequences may exhibit long memory spatial dependence, the condition
of strong mixing for the spatial dependence in Jenish and Prucha $\left(
2012\right) $ is ruled out. This appears to be the case as Ibragimov and
Rozanov $\left( 1978\right) $ showed; if the sequence $\left\{ \gamma
_{u,pq}\left( j\right) \right\} _{j\in \mathbb{Z}}$ is not summable, the
process $\left\{ u_{pt}\right\} _{t\in \mathbb{Z}},p\in \mathbb{N}^{+}$,
cannot be \emph{strong}-\emph{mixing}. The long memory dependence also rules
out that the process is Near Epoch Dependent with size greater than $1/2$,
which appears to be a necessary condition for standard asymptotic CLT
results.

Conditions $C1$ and $C2$ do rule out long\textbf{\ }memory temporal
dependence on the sequences $\left\{ x_{pt}\right\} _{t\in \mathbb{Z}}$ and $%
\left\{ u_{pt}\right\} _{t\in \mathbb{Z}}$ for each $p$. Even though there
are several results available allowing their temporal dependence to exhibit
long memory, see Robinson and Hidalgo $\left( 1997\right) $ or Hidalgo $%
\left( 2003\right) $, we have decided to assume the temporal dependence of
the regressors and errors to be weakly dependent to simplify the arguments.
It is worth pointing out that our Conditions $C1\left( \mathbf{i}\right) $
and $C2\left( \mathbf{i}\right) $ can be relaxed to some extend to allow
some type of mixing condition such as $L^{4}-$Near Epoch dependence with
size greater than or equal to $2$. The latter condition is often invoked
when we allow the errors to have a nonlinear type of dependence structure or
if $\left( \text{\ref{model_1}}\right) $ were replaced by a nonlinear panel
data model 
\begin{equation*}
y_{pt}=g\left( x_{pt};\beta \right) +\eta _{p}+\alpha _{t}+u_{pt}\text{, \ \ 
}p=1,...,n,\text{ \ \ \ }t=1,...,T\text{.}
\end{equation*}%
In fact, we expect the conclusions of our results to hold under such a
mixing condition\textbf{\ }as it has been shown in numerous situations.
Conditions $C1$ and $C2$\textbf{\ }do permit heterogeneity in its second
moments as $E\left( \xi _{pt}^{2}\mid \mathcal{V}_{p,t-1}\right) =\sigma
_{\xi ,p}^{2}$\emph{\ }and\emph{\ }$Cov\left( \chi _{pt}\mid \Upsilon
_{p,t-1}\right) =\Sigma _{\chi ,p}$. This follows from our conditions
because $E\left( \xi _{pt}^{2}\mid \mathcal{V}_{p,t-1}\right) =\sum_{\ell
=1}^{\infty }\left\vert a_{\ell }\left( p\right) \right\vert ^{2}$ clearly
depends on $p$. Furthermore, we allow for some trending behaviour of the
sequences $\left\{ x_{pt}\right\} _{t\in \mathbb{Z}}$, $p\in \mathbb{N}^{+}$%
, as we allow the mean of $x_{pt}$ to depend on time.

An important consequence of Conditions $C1$ and $C2$ is that they guarantee
that the covariance structure of the sequences $\left\{ u_{pt}\right\}
_{t\in \mathbb{Z}}$ and $\left\{ x_{pt}\right\} _{t\in \mathbb{Z}}$, $p\in 
\mathbb{N}^{+}$, is multiplicative. For instance, Condition $C1$ implies
that, for all $p,q\in \mathbb{N}^{+}$, 
\begin{eqnarray}
E\left( u_{pt}u_{qs}\right) &=&E\left( \dsum\limits_{k=0}^{\infty
}d_{k}\left( p\right) \xi _{p,t-k}\dsum\limits_{\ell =0}^{\infty }d_{\ell
}\left( q\right) \xi _{q,s-\ell }\right)  \notag \\
&=&E\left( \xi _{p1}\xi _{q1}\right) \left\{ 
\begin{array}{c}
\dsum\limits_{\ell =0}^{\infty }d_{t-s+\ell }\left( p\right) d_{\ell }\left(
q\right) \text{ \ \ }t>s \\ 
\dsum\limits_{\ell =0}^{\infty }d_{\ell }\left( p\right) d_{s-t+\ell }\left(
q\right) \text{ \ \ }t\leq s%
\end{array}%
\right.  \label{cov_u} \\
&=&\varphi _{u}\left( p,q\right) \gamma _{u;pq}\left( t-s\right) \text{.} 
\notag
\end{eqnarray}%
Following the spatio-temporal literature, see Cressie and Huang $\left(
1999\right) $, we can denote this covariance structure as \emph{separable}.
Of course, there are \emph{nonseparable} covariance structures, see Gneiting 
$\left( 2002\right) $ and tests for separability are available, see Fuentes $%
\left( 2006\right) $ or Matsuda and Yajima $\left( 2004\right) $.\ Notice
that in the absence of cross-sectional dependence, $E\left( \xi _{p1}\xi
_{q1}\right) =\sigma _{\xi p}^{2}\mathbf{1}\left( p=q\right) $ and $E\left(
u_{pt}u_{qs}\right) =\sigma _{\xi p}^{2}\gamma _{u;pp}\left( t-s\right) 
\mathbf{1}\left( p=q\right) $. Here, and in what follows, $\mathbf{1}\left(
A\right) $ denotes the indicator function.

\begin{remark}
The condition $\sup_{p\in \mathbb{N}^{+}}\sum_{\ell =0}^{\infty }\left\vert
a_{\ell }\left( p\right) \right\vert ^{2}<\infty $ guarantees that for any
reordering of the sequence $\left\{ \left\vert a_{\ell }\left( p\right)
\right\vert ^{2}\right\} _{\ell \in \mathbb{N}^{+}}$, say $\left\{
\left\vert a_{\ell \left( \tau \right) }\left( p\right) \right\vert
^{2}\right\} _{\ell \left( \tau \right) \in \mathbb{N}^{+}}$, we have that $%
a_{\ell \left( \tau \right) }\left( p\right) =O\left( \ell \left( \tau
\right) ^{-\zeta }\right) $ for some $\zeta >1/2$. Similarly the requirement 
$\sup_{\ell \geq 1}\sum_{p=1}^{n}\left\vert a_{\ell }\left( p\right)
\right\vert ^{2}<\infty $ will mean that $a_{\ell }\left( p\right) =O\left(
p^{-\zeta }\right) $ for some $\zeta >1/2$ uniformly in $\ell \geq 1$.
Similar arguments follow for $\left\{ \left\vert b_{\ell }\left( p\right)
\right\vert ^{2}\right\} _{\ell \in \mathbb{N}^{+}}$, $p\geq 1$.
\end{remark}

Condition $C3$ assumes that the sequences $\left\{ x_{pt}\right\} _{t\in 
\mathbb{Z}}$ and $\left\{ u_{pt}\right\} _{t\in \mathbb{Z}}$\emph{, }$p\in 
\mathbb{N}^{+}$ are independent, although we envisage that it can be relaxed
to require only conditional independence in first and second moments. To
simplify the arguments somewhat, we have preferred to keep the condition as
it stands. Even though we allow long memory spatial dependence of the
individual sequences, the absolute summability requirement in $\left( \text{%
\ref{c3}}\right) $\ limits the combined cross-sectional dependence, that is
the dependence of the sequence $\left\{ z_{pt}=u_{pt}x_{pt}\right\} _{t\in 
\mathbb{Z}}$, $p\in N^{+}$, is \textquotedblleft weakly spatially
dependent\textquotedblright , see also Hidalgo and Schafgans $\left(
2017\right) $.\textbf{\ }We have adopted the convention that $\gamma
_{u;pp}\left( t-s\right) =E\left( u_{pt}u_{ps}\right) /\varphi _{u}\left(
p,p\right) .$ Importantly, as we assume that the errors and regressors are
uncorrelated, the spectral density matrix of the sequences $\left\{
z_{pt}=:u_{pt}x_{pt}\right\} _{t\in \mathbb{Z}}$, $p\in \mathbb{N}^{+}$ is
given by the convolution of the spectral density matrix of $\left\{
x_{pt}\right\} _{t\in \mathbb{Z}}$ and spectral density function of $\left\{
u_{pt}\right\} _{t\in \mathbb{Z}}$, that is%
\begin{equation}
f_{p}\left( \lambda \right) =:\int_{-\pi }^{\pi }f_{u,p}\left( \upsilon
\right) f_{x,p}\left( \lambda -\upsilon \right) d\upsilon \text{, \ }p\in 
\mathbb{N}^{+}\text{,}  \label{fz}
\end{equation}%
where Conditions $C1$ and $C2$ imply that $f_{p}\left( \lambda \right) $ is
twice continuous differentiable. By Fuller's $\left( 1996\right) $ Theorem
3.4.1, or Corollary 3.4.1.2, the Fourier coefficients of $f_{p}\left(
\lambda \right) $ are given by $\gamma _{p}\left( j\right) =\gamma
_{x,p}\left( j\right) \gamma _{u,p}\left( j\right) $, \ $p\in \mathbb{N}^{+}$%
, so that 
\begin{equation*}
\sup_{p,q=1,..,n}\sum_{\ell =-\infty }^{\infty }\left\Vert \gamma
_{pq}\left( \ell \right) \right\Vert <\infty ;\text{ \ \textsl{Cov}}\left(
z_{pt};z_{qs}\right) =\gamma _{pq}\left( t-s\right) \varphi \left(
p,q\right) \text{.}
\end{equation*}%
With the convention that $\gamma _{u,pq}\left( 0\right) =\gamma
_{x,pq}\left( 0\right) =1$, $Cov\left( z_{pt},z_{qt}\right) =\varphi \left(
p,q\right) =:\varphi _{u}\left( p,q\right) \varphi _{x}\left( p,q\right) $
as defined in Condition $C3$.

\begin{remark}
\label{Remark_1}It is worth noticing that $\left( \text{\ref{c3}}\right) $
ensures that $\varphi \left( p,q\right) =O\left( q^{-1-\delta }\right) $ or $%
\varphi \left( p,q\right) =O\left( p^{-1-\delta }\right) $ for some $\delta
>0$, so that 
\begin{equation*}
\lim_{n\rightarrow \infty }\frac{1}{n}\sum_{p,q=1}^{n}\varphi \left(
p,q\right) <\infty \text{.}
\end{equation*}%
The latter displayed expression can be regarded as a type of weak dependence
in the cross-sectional dimension, see also Robinson $\left( 2011\right) $ or
Lee and Robinson $\left( 2013\right) $. In addition, the ergodicity in
second mean, that is%
\begin{equation*}
\frac{1}{n^{2}}\sum_{p,q=1}^{n}\left( \varphi _{u}\left( p,q\right) +\varphi
_{x}\left( p,q\right) \right) =o(1)\text{,}
\end{equation*}%
implies that $\varphi _{u}\left( p,q\right) =O\left( q^{-\varsigma
_{u}}\right) $ and $\varphi _{x}\left( p,q\right) =O\left( q^{-\varsigma
_{x}}\right) $ such that $\varsigma _{u}+\varsigma _{x}=1+\delta >0$.
\end{remark}

Conditions $C1-C3,$ therefore, imply that the \textquotedblleft \emph{average%
}\textquotedblright\ long-run variance of the sequences $\left\{
z_{pt}=:u_{pt}x_{pt}\right\} _{t\in \mathbb{Z}}$, $p\in \mathbb{N}^{+},$ is
given by\textbf{\ }%
\begin{eqnarray}
&&\Phi =:2\pi \lim_{n\rightarrow \infty }\frac{1}{n}\sum_{p,q=1}^{n}f_{pq}%
\left( 0\right) \varphi \left( p,q\right) <\infty  \label{sep_1} \\
&&2\pi f_{pq}\left( 0\right) =\sum_{\ell =-\infty }^{\infty }\gamma
_{pq}\left( \ell \right) .  \notag
\end{eqnarray}%
Observe that standard algebra yields that%
\begin{eqnarray}
\Phi &=&:\lim_{n\rightarrow \infty }\lim_{T\rightarrow \infty }\frac{1}{nT}%
E\left\{ \left( \sum_{p=1}^{n}\sum_{t=1}^{T}x_{pt}u_{pt}\right) \left(
\sum_{p=1}^{n}\sum_{t=1}^{T}x_{pt}^{\prime }u_{pt}\right) \right\}  \notag \\
&=&\lim_{n\rightarrow \infty }\lim_{T\rightarrow \infty }\frac{1}{nT}%
\sum_{p,q=1}^{n}\sum_{t,s=1}^{T}E\left( x_{pt}x_{qs}^{\prime }\right)
E\left( u_{pt}u_{qs}\right) ,  \label{v_1}
\end{eqnarray}%
or, using its spectral domain formulation,%
\begin{eqnarray}
\Phi &=&\lim_{n\rightarrow \infty }\lim_{T\rightarrow \infty }\frac{1}{nT}%
E\left\{ \left( \sum_{j=1}^{T-1}\sum_{p=1}^{n}\mathcal{J}_{x,p}\left(
\lambda _{j}\right) \mathcal{J}_{u,p}\left( -\lambda _{j}\right) \right)
\left( \sum_{j=1}^{T-1}\sum_{p=1}^{n}\mathcal{J}_{x,p}^{\prime }\left(
-\lambda _{j}\right) \mathcal{J}_{u,p}\left( \lambda _{j}\right) \right)
\right\}  \notag \\
&=&\lim_{n\rightarrow \infty }\lim_{T\rightarrow \infty }\frac{1}{nT}%
\sum_{j=1}^{T-1}\sum_{p,q=1}^{n}E\left( \mathcal{J}_{x,p}\left( \lambda
_{j}\right) \mathcal{J}_{x,q}^{\prime }\left( -\lambda _{j}\right) \right)
E\left( \mathcal{J}_{u,p}\left( -\lambda _{j}\right) \mathcal{J}_{u,q}\left(
\lambda _{j}\right) \right) \text{.}  \label{v_1freq}
\end{eqnarray}%
Finally, we denote%
\begin{equation}
\text{\textsl{V}}=\Sigma _{x}^{-1}\Phi \Sigma _{x}^{-1}\text{,}  \label{v_2}
\end{equation}%
where $\Sigma _{x}>0$ was defined in Condition $C2$.

The following theorem presents our result establishing the CLT for our slope
parameter estimates in the presence of general temporal and cross-sectional
dependence.

\begin{theorem}
\label{ThmEst}Under Conditions $C1-C3$, we have that as $n,T\rightarrow
\infty $,%
\begin{equation*}
\text{\ }\left( nT\right) ^{1/2}\left( \widetilde{\beta }-\beta \right) 
\overset{d}{\rightarrow }\mathcal{N}\left( 0,\text{\textsl{V}}\right) \text{.%
}
\end{equation*}
\end{theorem}

\begin{proof}
The proof of this result, based on either the time or frequency domain
formulation, will be given in Appendix A. All other proofs are relegated to
this appendix as well.
\end{proof}

\begin{remark}
While the result could be shown to hold with finite $n$,\ a setting
considered by Robinson (1998), the presence of the time fixed effect would
require special attention since the dependence structure of $u_{pt}$ and $%
n^{-1}\sum_{p=1}^{n}u_{pt}$ are not quite the same when $n$\ is finite.
\end{remark}

With \textsl{V} defined in $\left( \text{\ref{v_2}}\right) $, Theorem \ref%
{ThmEst} indicates that to make inferences on $\beta $, we need to provide a
consistent estimator of $\Phi $. A first glance at $\left( \text{\ref{v_1}}%
\right) $ or $\left( \text{\ref{v_1freq}}\right) $ suggests that this might
be complicated or computationally burdensome due to the general
spatio-temporal dependence structure of the data. As we pointed out in the
introduction, the standard approach to deal with such dependence, that is a
HAC type of estimator, has various potential drawbacks in the presence of
cross-sectional dependence. While choosing a bandwidth parameter associated
with the cross-sectional dependence requires or induces an artificial and/or
nontrivial ordering, the presence of individual heterogeneous temporal
dependence (as assumed in Conditions $C1$\ and $C2)$\ would even render any
cross validation method used to choose the temporal bandwidth parameter
intractable.

While Kim and Sun's $\left( 2013\right) $\ approach is subject to both these
criticisms, Driscoll and Kraay $\left( 1989\right) $ avoid the need to
specify an ordering of individuals by introducing an HAC estimator of
cross-sectional averages, so that one can consider their estimator as a
hybrid between an HAC and a cluster one: they employ the HAC methodology to
deal with the temporal dependence whereas they employ a cluster type of
estimator to account for the cross-sectional dependence. We advocate to use
an approach that does not require any ordering and/or selection of a
bandwidth parameter and also permits a more general spatio-temporal
dependence than allowed by either Driscoll and Kraay $\left( 1989\right) $
or Kim and Sun $\left( 2013\right) $ and permits the cross-sectional
dependence to be "long-memory" which latter work ruled out. Moreover our
approach permits the temporal dependence to be heterogeneous across
individuals, which is more realistic.

Our approach can be regarded as a natural extension of the earlier work by
Robinson $\left( 1998\right) $\ on inference without smoothing in a time
series regression model context. In his case, abstracting from
cross-sectional dependence%
\begin{equation*}
\Phi =:\lim_{n\rightarrow \infty }\frac{2\pi }{n}\sum_{p=1}^{n}f_{pp}\left(
0\right) \text{.}
\end{equation*}%
Applying his estimator to our model, would yield the estimator%
\begin{equation}
\frac{2\pi }{n}\sum_{p=1}^{n}\frac{1}{T}\sum_{j=1}^{T}\mathcal{I}%
_{u,p}\left( \lambda _{j}\right) \mathcal{I}_{x,p}\left( -\lambda
_{j}\right) =\frac{1}{n}\sum_{p=1}^{n}\sum_{\ell =-T+1}^{T-1}\widehat{\gamma 
}_{x,p}\left( \ell \right) \widehat{\gamma }_{u,p}\left( \ell \right) \text{,%
}  \label{rob_1}
\end{equation}%
where $\widehat{\gamma }_{x,p}\left( j\right) $ and $\widehat{\gamma }%
_{u,p}\left( j\right) $ are respectively the standard sample moment
estimators of $\gamma _{x,p}\left( j\right) $ and $\gamma _{u,p}\left(
j\right) $ and $I_{u,p}\left( \lambda \right) =T^{-1}\left(
\sum_{t=1}^{T}u_{pt}e^{it\lambda }\right) \left(
\sum_{t=1}^{T}u_{pt}e^{-it\lambda }\right) ^{\prime }$ with $I_{x,p}\left(
\lambda \right) $ similarly defined. When cross-sectional dependence is
allowed, the latter arguments suggest that $\left( \text{\ref{rob_1}}\right) 
$ is not a consistent (cluster) estimator of $\Phi $. The reason for this
(see also the proof of Proposition \ref{PropV1} below) is that 
\begin{equation*}
\frac{1}{n}\sum_{p=1}^{n}\sum_{\ell =-T+1}^{T-1}\gamma _{x,p}\left( \ell
\right) \gamma _{u,p}\left( \ell \right) \nrightarrow \Phi
\end{equation*}%
as expected since the first moment of $\left( \text{\ref{rob_1}}\right) $
does not capture the cross-sectional dependence. The purpose of the next
section is therefore to provide a consistent \textquotedblleft
cluster\textquotedblright\ estimator of $\Phi $ that accounts for the
presence of cross-sectional dependence.

\subsection{\textbf{Cluster estimator of }$\Phi $}

$\left. {}\right. $

We shall present a simple cluster estimator of $\Phi $ using the
\textquotedblleft frequency\textquotedblright\ domain methodology.
Obviously, there is a time domain analogue, which we shall briefly describe
at the end of the section. Our cluster estimator appears to be the first one
which permits both time and cross-sectional dependence and gives a formal
justification of its statistical properties. Our estimator therefore becomes
an extension of previous cluster estimators in the literature such as that
in Arellano $\left( 1987\right) $\ (where only temporal dependence is
present) or Bester, Conley and Hansen $\left( 2011\right) $ (where only
cross-sectional dependence is present).

Our main motivation to propose a cluster estimator using the frequency
domain methodology comes from the well known observation that for all $j\neq
k$, $J_{u,p}\left( \lambda _{j}\right) $ and $J_{u,q}\left( \lambda
_{k}\right) $ can be considered as being uncorrelated although possibly
heteroskedastic. This observation was employed in the landmark paper by
Hannan $\left( 1963\right) $ on adaptive estimation in a time series
regression model. The fact that we may therefore consider $\mathcal{J}_{%
\widetilde{x},p}\left( \lambda _{j}\right) \mathcal{J}_{u,p}\left( -\lambda
_{j}\right) $ as a sequence of uncorrelated and heteroskedastic random
variables in $j$, although not in $p$, suggests that, in a spirit similar to
White's $\left( 1980\right) $ estimator, we may estimate $\Phi $ by 
\begin{equation}
\breve{\Phi}=\frac{1}{T}\sum_{j=1}^{T-1}\left\{ \left( \frac{1}{n^{1/2}}%
\sum_{p=1}^{n}\mathcal{J}_{\widetilde{x},p}\left( \lambda _{j}\right) 
\mathcal{J}_{\widehat{u},p}\left( -\lambda _{j}\right) \right) \left( \frac{1%
}{n^{1/2}}\sum_{p=1}^{n}\mathcal{J}_{\widetilde{x},p}^{\prime }\left(
-\lambda _{j}\right) \mathcal{J}_{\widehat{u},p}\left( \lambda _{j}\right)
\right) \right\} \text{.}  \label{v_estimate}
\end{equation}%
Based on the DFT formulation, we denote the estimator of $\Sigma _{x}$ by 
\begin{equation*}
\widetilde{\Sigma }_{x}=\frac{1}{nT}\sum_{j=1}^{T-1}\sum_{p=1}^{n}\mathcal{J}%
_{\widetilde{x},p}\left( \lambda _{j}\right) \mathcal{J}_{\widetilde{x}%
,p}^{\prime }\left( -\lambda _{j}\right) \text{.}
\end{equation*}

The following proposition establishes the consistency of our cluster
estimator for the \textquotedblleft average\textquotedblright\ long-run
variance of the sequences $\left\{ z_{pt}=:u_{pt}x_{pt}\right\} _{t\in 
\mathbb{Z}}$, $p\in N^{+}.$

\begin{proposition}
\label{PropV1}Under the conditions of Theorem \ref{ThmEst}, we have that%
\begin{eqnarray*}
&&\left( \mathbf{a}\right) \ \ \ \ \ \breve{\Phi}-\Phi =o_{p}\left( 1\right)
\\
&&\left( \mathbf{b}\right) \ \ \ \ \ \widetilde{\Sigma }_{x}-\Sigma
_{x}=o_{p}\left( 1\right) \text{.}
\end{eqnarray*}
\end{proposition}

Denoting \textsl{\^{V}}$=:\widetilde{\Sigma }_{x}^{-1}\breve{\Phi}\widetilde{%
\Sigma }_{x}^{-1}$, we now obtain the following corollary.

\begin{corollary}
\label{Pro_est}Under the conditions of Theorem \ref{ThmEst}, we have that%
\begin{equation*}
\left( nT\right) ^{1/2}\text{\textsl{\^{V}}}^{-1/2}\left( \widetilde{\beta }%
-\beta \right) \overset{d}{\rightarrow }\mathcal{N}\left( 0,I\right) \text{.}
\end{equation*}
\end{corollary}

\begin{proof}
The proof is standard from Theorem \ref{ThmEst} and Proposition \ref{PropV1}%
, and is therefore omitted.
\end{proof}

We now describe the time domain analogue estimator of $\Phi $. For that
purpose, using $\sum_{t=1}^{T}e^{it\lambda _{\ell }}=0$ if $1\leq \ell \leq
T-1$, we have after standard algebra that 
\begin{equation*}
\breve{\Phi}=\frac{1}{n}\sum_{p,q=1}^{n}\sum_{\left\vert \ell \right\vert
=0}^{T-1}\widehat{\gamma }_{x,pq}\left( \ell \right) \widehat{\gamma }%
_{u,pq}\left( \ell \right) \text{,}
\end{equation*}%
where due to $\left( \ref{cov_u}\right) $ 
\begin{eqnarray*}
\widehat{\gamma }_{x,pq}\left( \ell \right) &=&\frac{1}{T}%
\sum_{t=1}^{T-\left\vert \ell \right\vert }\widetilde{x}_{pt}\widetilde{x}%
_{q,t+\ell }^{\prime }; \\
\widehat{\gamma }_{u,pq}\left( \ell \right) &=&\frac{1}{T}%
\sum_{t=1}^{T-\left\vert \ell \right\vert }\widehat{u}_{pt}\widehat{u}%
_{q,t+\ell }\mathbf{1}\left( \ell >0\right) +\frac{1}{T}\sum_{t=1}^{T-\left%
\vert \ell \right\vert }\widehat{u}_{qt}\widehat{u}_{p,t+\ell }\mathbf{1}%
\left( \ell <0\right) \text{,}
\end{eqnarray*}%
and $\widehat{u}_{pt}=\widetilde{y}_{pt}-\widetilde{\beta }^{\prime }%
\widetilde{x}_{pt}$, \ $p=1,...,n;$ $t=1,...,T$.

\section{\textbf{BOOTSTRAP SCHEMES}}

Our motivation to introduce bootstrap schemes emanates from findings in our
Monte-Carlo experiment, which suggest that the asymptotic distribution of $%
\left( nT\right) ^{1/2}$\textsl{\^{V}}$^{-1/2}\left( \widetilde{\beta }%
-\beta \right) $ does not appear to provide a good approximation of its
finite sample distribution. In such situations, the use of the bootstrap has
been advocated as it has been shown to improve the finite sample
performance. The general spatio-temporal dependence inherent in our model
suggests that a valid bootstrap mechanism may not to be easy to implement
since one of the basic requirements for its validity is that it has to
preserve the covariance structure of the data/model. Drawing analogies from
the time series literature, one might be tempted to use the block bootstrap\
($BB$) principle as it is not clear how the sieve bootstrap can be
implemented under cross-sectional dependence in the absence of a clear
ordering of the data. Applying a $BB$\ in both dimensions, however, would
also be sensitive to the particular ordering chosen by the practitioner and
be subject to the absence of weak stationarity, where the dependence
structure of say $\left( x_{p_{1},t},...,x_{p_{1}+m,t}\right) ^{\prime }$\
and $\left( x_{p_{1}+1,t},...,x_{p_{1}+m+1,t}\right) ^{\prime }$\ need not
be identical.

Avoiding the need to establish a particular ordering of the cross sectional
units, Gon\c{c}alves $\left( 2011\right) $\ proposes to apply a moving block
bootstrap ($MBB$) to the vector containing all individual observations for
each\textbf{\ }$t$, that is it only applies a $BB$ in the time dimension.\
The $MBB$, however, does require the choice of the block size and is known
to be sensitive to its choice in finite samples. In the absence of temporal
dependence, the block size equals one, and the approach is similar to
Hidalgo and Schafgans (2017).\ 

Here we propose valid bootstrap schemes with the interesting feature that
they are computationally simple (there is no need to estimate, either by
parametric or nonparametric methods, the time and/or cross-sectional
dependence of the error term) and do not require the choice of any bandwidth
parameter for its implementation, thereby avoiding any level of
arbitrariness.

Both bootstrap schemes considered are in the frequency domain. We recall the
DFT for generic sequences $\left\{ \varsigma _{pt}\right\} _{t=1}^{T}$, $%
p\geq 1$, as $\mathcal{J}_{\varsigma ,p}\left( \lambda _{j}\right) =\frac{1}{%
T^{1/2}}\sum_{t=1}^{T}\varsigma _{pt}e^{-it\lambda _{j}}$, \ \ $j=1,...,%
\widetilde{T}=\left[ T/2\right] $, \ $\lambda _{j}=\frac{2\pi j}{T}$, and
define its periodogram 
\begin{equation*}
\mathcal{I}_{\varsigma ,p}\left( \lambda _{j}\right) =\left\vert \mathcal{J}%
_{\varsigma ,p}\left( \lambda _{j}\right) \right\vert ^{2}\text{ \ \ }%
j=1,...,\widetilde{T}=[T/2]\text{, \ \ }p=1,...,n\text{.}
\end{equation*}

The first scheme, labelled the na\"{\i}ve bootstrap, imposes Condition $C4$\
under which the time dependence is assumed to be homogeneous across
individuals. We relax this assumption in the second scheme, labelled the
wild bootstrap, in line with Conditions $C1$ and $C2$.

For homogenous temporal dependence, therefore, we impose

\begin{description}
\item[Condition $\mathbf{C4}$] Homogeneous time dependence: $d_{k}(p)$\ and $%
c_{k}(p)$\ defined in Conditions $C1$ and $C2$ do not vary over $p.$
\end{description}

Denoting $\sigma _{u}^{2}\left( p\right) =Eu_{pt}^{2}$ and $f_{u,p}\left(
\lambda \right) $ the spectral density function of the sequence $\left\{
u_{pt}\right\} _{t\in \mathbb{Z}}$, for any $p=1,2,...$, Condition $C4$
ensures that%
\begin{equation}
\frac{f_{u,p}\left( \lambda \right) }{\sigma _{u}^{2}\left( p\right) }%
=:g_{u}\left( \lambda \right) \text{, \ \ }p=1,2,3,...\text{.}  \label{HTT}
\end{equation}%
That is, the spectral density function normalized by the variance does not
depend on $p$. This enables us to use the average periodogram of
standardized residuals in constructing the valid bootstrap model. What
really matters here is that the "correlation" structure is the same.

The na\"{\i}ve bootstrap, involves resampling from the residuals of the model%
\footnote{%
See also Remark \ref{DFTbootstrap} below.} and involves the following simple 
\emph{3 STEPS:}

\begin{description}
\item[\emph{STEP 1}] Obtain the residuals%
\begin{equation*}
\widehat{u}_{pt}=\widetilde{y}_{pt}-\widetilde{\beta }^{\prime }\widetilde{x}%
_{pt}\text{, \ }p=1,...,n;\text{ }t=1,...,T,
\end{equation*}
compute $\widetilde{\sigma }_{\widehat{u}}^{2}\left( p\right)
=T^{-1}\sum_{t=1}^{T}\widehat{u}_{pt}^{2}$, and obtain the standardized
residuals 
\begin{equation*}
\check{u}_{pt}=\widehat{u}_{pt}/\widetilde{\sigma }_{\widehat{u}}\left(
p\right) \text{.}
\end{equation*}

\item[\emph{STEP 2}] Denoting $\hat{U}_{t}=\left\{ \hat{u}_{pt}\right\}
_{p=1}^{n}$, do standard random sampling from the empirical distribution of
the residuals $\{\hat{U}_{t}\}_{t=1}^{T}$. That is, we assign probability $%
T^{-1}$ to each $n\times 1$ vector $\hat{U}_{t}$. Denote the bootstrap
sample by $\left\{ U_{t}^{\ast }\right\} _{t=1}^{T},$ where $U_{t}^{\ast
}=\left\{ u_{pt}^{\ast }\right\} _{p=1}^{n}.$ Compute the bootstrap analogue
of $\left( \text{\ref{model_1R}}\right) $ as%
\begin{equation*}
\mathcal{J}_{y\ast ,p}\left( \lambda _{j}\right) =\widetilde{\beta }^{\prime
}\mathcal{J}_{\widetilde{x},p}\left( \lambda _{j}\right) +\left( \frac{1}{n}%
\sum_{q=1}^{n}\mathcal{I}_{\check{u},q}\left( \lambda _{j}\right) \right)
^{1/2}\mathcal{J}_{u^{\ast },p}\left( \lambda _{j}\right)
\end{equation*}%
\ for $p=1,...,n$ \ and $j=1,...,T-1.$

\item[\emph{STEP 3}] Compute the corresponding bootstrap analogue of $\left( 
\text{\ref{beta_fef}}\right) $ as 
\begin{equation}
\widetilde{\beta }^{\ast }=\left( \sum_{p=1}^{n}\sum_{j=1}^{T-1}\mathcal{J}_{%
\widetilde{x},p}\left( \lambda _{j}\right) \mathcal{J}_{\widetilde{x}%
,p}^{\prime }\left( -\lambda _{j}\right) \right) ^{-1}\left(
\sum_{p=1}^{n}\sum_{j=1}^{T-1}\mathcal{J}_{\widetilde{x},p}\left( \lambda
_{j}\right) \mathcal{J}_{\widetilde{y}\ast ,p}\left( -\lambda _{j}\right)
\right) ,  \label{beta_fe**}
\end{equation}%
with $\mathcal{J}_{\widetilde{y}\ast ,p}\left( \lambda _{j}\right) =\mathcal{%
J}_{y\ast ,p}\left( \lambda _{j}\right) -\frac{1}{n}\sum_{q=1}^{n}\mathcal{J}%
_{y\ast ,q}\left( \lambda _{j}\right) $.
\end{description}

\begin{remark}
Since $\overline{\widehat{u}}_{p}=0$ there is no need the recenter in Step
1. The standardization of the residuals (the variance is not the same for
all individuals) is used in Step 2 to impose the appropriate dependence
structure on our bootstrap regression. As the bootstrap is done on the
vector containing all individual observations for each $t$\ there is no need
for standardization otherwise.
\end{remark}

\begin{remark}
We use the average periodogram of the standardized residuals to impose the
appropriate dependence structure on our bootstrap regression in Step 2. When
the time dependence is homogeneous among the cross-sectional units $\frac{1}{%
n}\sum_{q=1}^{n}\mathcal{I}_{\check{u},q}\left( \lambda _{j}\right) =$ $%
\sigma _{u}^{-2}(p)f_{u,p}\left( \lambda _{j}\right)
(1+o_{p}(1))=:g_{u}\left( \lambda _{j}\right) \left( 1+o_{p}\left( 1\right)
\right) ,$ see also (\ref{HTT}). That is, if the temporal dependence were
given by an AR(1) model, the right side becomes the spectral density
function of an AR(1) sequence, where the innovation sequence has variance
equal to 1. In addition, as we bootstrap from $\hat{u}_{pt}$, which are the
residuals, we ensure that the variance of $u_{pt}^{\ast }$ is that of $%
u_{pt} $.
\end{remark}

\begin{remark}
\label{DFTbootstrap}Alternatively, we could have used random sampling from
the normalized DFT of the residuals as considered by Hidalgo (2003). In that
case, denoting $T_{\widehat{u}}\left( \lambda _{j}\right) =\left\{ \mathcal{J%
}_{\widehat{u},p}\left( \lambda _{j}\right) /|\mathcal{J}_{\widehat{u}%
,p}\left( \lambda _{j}\right) |\right\} _{p=1}^{n}$,\ $T_{u^{\ast },p}\left(
\lambda _{j}\right) $\ form independent draws from the empirical
distribution of $\tilde{T}_{\widehat{u}}\left( \lambda _{j}\right) =(T_{%
\widehat{u}}\left( \lambda _{j}\right) -\bar{T}_{\widehat{u}})/\hat{\sigma}_{%
\mathcal{T}}$\ where $\bar{T}_{\widehat{u}}=\left[ T/2\right]
^{-1}\sum_{j=1}^{T/2}T_{\widehat{u}}\left( \lambda _{j}\right) $\ and $\hat{%
\sigma}_{\mathcal{T}}^{2}=\left[ T/2\right] ^{-1}\sum_{j=1}^{T/2}\left( 
\mathcal{T}_{\widehat{u}}\left( \lambda _{j}\right) -\mathcal{\bar{T}}%
\right) ^{2}.$\ The bootstrap analogue of $\left( \text{\ref{model_1R}}%
\right) $ would then be obtained using $\mathcal{J}_{y\ast ,p}\left( \lambda
_{j}\right) =\widetilde{\beta }^{\prime }\mathcal{J}_{\widetilde{x},p}\left(
\lambda _{j}\right) +\left( \frac{1}{n}\sum_{q=1}^{n}\mathcal{I}_{\check{u}%
,q}\left( \lambda _{j}\right) \right) ^{1/2}\widetilde{\sigma }_{\widehat{u}%
}\left( p\right) \mathcal{T}_{u^{\ast },p}\left( \lambda _{j}\right) $. Our
scheme uses Step 2, which has better finite sample properties as observed in
Hidalgo (2003).
\end{remark}

The key feature of this na\"{\i}ve bootstrap, is that there is no need to
choose any bandwidth parameter for its implementation. Under Condition $C4$%
,\ uniformly in $j=1,...,T-1$, we have that 
\begin{eqnarray*}
\mathcal{I}_{\check{u},p}\left( \lambda _{j}\right) &=&\widetilde{\sigma }_{%
\widehat{u}}^{-2}\left( p\right) \left\{ \mathcal{I}_{u,p}\left( \lambda
_{j}\right) +\left( \widetilde{\beta }-\beta \right) ^{2}\mathcal{I}%
_{x,p}\left( \lambda _{j}\right) +\left( \widetilde{\beta }-\beta \right) 
\mathcal{J}_{x,p}\left( \lambda _{j}\right) \mathcal{J}_{u,p}\left( -\lambda
_{j}\right) \right\} \\
&=&\sigma _{u}^{-2}\left( p\right) \mathcal{I}_{u,p}\left( \lambda
_{j}\right) \left( 1+o_{p}\left( 1\right) \right) ,
\end{eqnarray*}%
and 
\begin{eqnarray*}
&&E\mathcal{I}_{u,p}\left( \lambda _{j}\right) =f_{u,p}\left( \lambda
_{j}\right) \left( 1+o\left( 1\right) \right) \\
&&E^{\ast }\left( \mathcal{J}_{u^{\ast },p}\left( \lambda _{j}\right) 
\mathcal{J}_{u^{\ast },p}\left( -\lambda _{\ell }\right) \right) =0,\ \ 
\text{if }j\neq \ell ,\text{ }\sigma _{u}^{2}(p)\text{ otherwise} \\
&&\widetilde{\sigma }_{\widehat{u}}^{2}\left( p\right) =\sigma
_{u}^{2}\left( p\right) \left( 1+o_{p}\left( 1\right) \right) \text{.}
\end{eqnarray*}%
The last displayed expressions suggest that, under Condition $C4$, we can
consider \linebreak $\left( \frac{1}{n}\sum_{q=1}^{n}\mathcal{I}_{\check{u}%
,q}\left( \lambda _{j}\right) \right) ^{1/2}\mathcal{J}_{u^{\ast },p}\left(
\lambda _{j}\right) $\textbf{\ }as some type of wild bootstrap in the
frequency domain because under homogeneous time dependence%
\begin{eqnarray*}
E^{\ast }\left\vert \left( \frac{1}{n}\sum_{q=1}^{n}\mathcal{I}_{\check{u}%
,q}\left( \lambda _{j}\right) \right) ^{1/2}\mathcal{J}_{u^{\ast },p}\left(
\lambda _{j}\right) \right\vert ^{2} &=&\left( \sigma
_{u}^{-2}(p)f_{u,p}\left( \lambda _{j}\right) \right) \cdot \sigma
_{u}^{2}\left( p\right) .\left( 1+o_{p}\left( 1\right) \right) \\
&=&f_{u,p}\left( \lambda _{j}\right) \left( 1+o_{p}\left( 1\right) \right) 
\text{.}
\end{eqnarray*}

The following theorem is used to establish the validity of our na\"{\i}ve
bootstrap scheme.

\begin{theorem}
\label{ThmBoot}(Na\"{\i}ve Bootstrap) Under Conditions $C1-C4$, we have that
in probability, 
\begin{equation*}
\left( nT\right) ^{1/2}\left( \widetilde{\beta }^{\ast }-\widetilde{\beta }%
\right) \overset{d^{\ast }}{\rightarrow }\mathcal{N}\left( 0,\text{\textsl{V}%
}\right) \text{.}
\end{equation*}
\end{theorem}

With the bootstrap cluster estimator of the asymptotic covariance, given by, 
\begin{equation}
\breve{\Phi}^{\ast }=\frac{1}{T}\sum_{j=1}^{T-1}\left\{ \left( \frac{1}{%
n^{1/2}}\sum_{p=1}^{n}\mathcal{J}_{\widetilde{x},p}\left( \lambda
_{j}\right) \mathcal{J}_{\hat{u}^{\ast },p}\left( -\lambda _{j}\right)
\right) \left( \frac{1}{n^{1/2}}\sum_{p=1}^{n}\mathcal{J}_{\widetilde{x}%
,p}^{\prime }\left( -\lambda _{j}\right) \mathcal{J}_{\hat{u}^{\ast
},p}\left( \lambda _{j}\right) \right) \right\}  \label{clus_2boot}
\end{equation}%
the next proposition establishes the consistency of the bootstrap cluster
estimator.

\begin{proposition}
\label{PropV1Boot}(Na\"{\i}ve Bootstrap) Under the assumptions of Theorem %
\ref{ThmBoot}, we have 
\begin{equation*}
\breve{\Phi}^{\ast }-\breve{\Phi}=o_{p^{\ast }}\left( 1\right) \text{.}
\end{equation*}
\end{proposition}

The previous results can be extended to incorporate the more realistic
situation where the temporal dynamics might differ by individual, as allowed
by Conditions $C1$ and $C2.$ This bootstrap, labelled the wild bootstrap,
merges ideas from Hidalgo $\left( 2003\right) $ and Chan and Ogden $\left(
2009\right) $. As the DFT residuals are heterogeneous whilst independent
over the Fourier frequencies, it applies the wild-type bootstrap approach to
the increasing dimensional vector $\left\{ \mathcal{J}_{\hat{u},p}\left(
\lambda _{j}\right) \right\} _{p=1}^{n}$. It requires a modification of the
above bootstrap, which primarily involves replacing \emph{STEP 2. }For
completeness we provide all steps:

\begin{description}
\item[\emph{STEP 1}$^{\prime }$] Obtain the residuals%
\begin{equation*}
\widehat{u}_{pt}=\widetilde{y}_{pt}-\widetilde{\beta }^{\prime }\widetilde{x}%
_{pt}\text{, \ }p=1,...,n;\text{ }t=1,...,T.
\end{equation*}

\item[\emph{STEP 2}$^{\prime }$] Denote $\left\{ \eta _{j}\right\} _{j=1}^{%
\widetilde{T}}$ a sequence of independent identically distributed random
variables with mean zero and unit variance. We then compute the bootstrap
analogue of $\left( \text{\ref{model_1R}}\right) $ as%
\begin{equation*}
\mathcal{J}_{y\ast ,p}\left( \lambda _{j}\right) =\widetilde{\beta }^{\prime
}\mathcal{J}_{\widetilde{x},p}\left( \lambda _{j}\right) +\mathcal{J}_{\hat{u%
},p}\left( \lambda _{j}\right) \eta _{j}\text{, \ \ \ }\left\{ 
\begin{array}{l}
p=1,...,n \\ 
j=1,...,T-1\text{,}%
\end{array}%
\right.
\end{equation*}%
where $\mathcal{J}_{y\ast ,p}\left( \lambda _{j}\right) =\overline{\mathcal{J%
}_{y\ast ,p}\left( \lambda _{T-j}\right) }$ and $\eta _{j}=\eta _{T-j}$, for 
$j=\widetilde{T}+1,...,T-1$.

\item[\emph{STEP 3}$^{\prime }$] Compute the corresponding bootstrap
analogue of $\left( \text{\ref{beta_fef}}\right) $ as 
\begin{equation*}
\widetilde{\beta }^{\ast }=\left( \sum_{p=1}^{n}\sum_{j=1}^{T-1}\mathcal{J}_{%
\widetilde{x},p}\left( \lambda _{j}\right) \mathcal{J}_{\widetilde{x}%
,p}^{\prime }\left( -\lambda _{j}\right) \right) ^{-1}\left(
\sum_{p=1}^{n}\sum_{j=1}^{T-1}\mathcal{J}_{\widetilde{x},p}\left( \lambda
_{j}\right) \mathcal{J}_{\widetilde{y}\ast ,p}\left( -\lambda _{j}\right)
\right) ,
\end{equation*}%
with $\mathcal{J}_{\widetilde{y}\ast ,p}\left( \lambda _{j}\right) =\mathcal{%
J}_{y\ast ,p}\left( \lambda _{j}\right) -\frac{1}{n}\sum_{q=1}^{n}\mathcal{J}%
_{y\ast ,q}\left( \lambda _{j}\right) $.
\end{description}

\begin{remark}
For a discussion regarding the requirement that $\eta _{j}=\eta _{T-j}$ for $%
j=\widetilde{T}+1,...,T-1$., we refer to Hidalgo $\left( 2003\right) $%
.\medskip
\end{remark}

The validity of the wild bootstrap scheme follows from the following
proposition.

\begin{proposition}
\label{Pro_estboot}(Wild Bootstrap) Under Conditions $C1-C3,$ in
probability, 
\begin{equation*}
\left( nT\right) ^{1/2}\left( \widetilde{\beta }^{\ast }-\widetilde{\beta }%
\right) \overset{d^{\ast }}{\rightarrow }\mathcal{N}\left( 0,\text{\textsl{V}%
}\right)
\end{equation*}%
and $\breve{\Phi}^{\ast }-\breve{\Phi}=o_{p^{\ast }}\left( 1\right) .$
\end{proposition}

We conclude with the stating the validity of the standardised bootstrap
statistic

\begin{corollary}
\label{Coll_estboot}Under Conditions $C1-C3$, we have that in probability, 
\begin{equation*}
\left( nT\right) ^{1/2}\text{\textsl{\^{V}}}^{\ast -1/2}\left( \widetilde{%
\beta }^{\ast }-\widetilde{\beta }\right) \overset{d^{\ast }}{\rightarrow }%
\mathcal{N}\left( 0,I\right) \text{,}
\end{equation*}%
where \textsl{\^{V}}$^{\ast }=\widetilde{\Sigma }_{x}^{-1}\breve{\Phi}^{\ast
}\widetilde{\Sigma }_{x}^{-1}$\textit{.}
\end{corollary}

\begin{proof}
The proof is standard after Theorem \ref{ThmBoot} and Propositions \ref%
{PropV1}, \ref{PropV1Boot} and \ref{Pro_estboot}.
\end{proof}

\section{(CONDITIONAL)\ HETEROSKEDASTICITY}

In this section, we extend our model to permit general forms of
heteroskedasticity. Specifically, we begin by considering%
\begin{equation}
y_{pt}=\beta ^{\prime }\acute{x}_{pt}+\eta _{p}+\alpha _{t}+v_{pt}\text{, \
\ }p=1,...,n,\text{ \ \ \ }t=1,...,T\text{,}  \label{model 2}
\end{equation}%
where 
\begin{equation}
v_{pt}=:\sigma _{1}\left( w_{p}\right) \sigma _{2}\left( \varrho _{t}\right)
u_{pt}\text{.}  \label{vpt}
\end{equation}%
The error $\left\{ u_{pt}\right\} _{t\in \mathbb{Z}}$, $p\in \mathbb{N}^{+}$%
satisfies the same regularity conditions given in Condition $C1$, exhibiting
general spatial and temporal dependence. The sequences $\left\{
w_{p}\right\} _{p\in \mathbb{N}}$\ and $\left\{ \varrho _{t}\right\} _{t\in 
\mathbb{Z}}$, which can even be functions of the fixed effects, are not
required to be mutually independent of the regressors\textbf{\ }$\left\{ 
\acute{x}_{pt}\right\} _{t\in \mathbb{Z}}$\textbf{, }$p\in N^{+}$\textbf{.}
Without loss of generality we will normalize $\sigma _{u,p}^{2}=1$ in
Condition $C1$; $\sigma _{u,p}^{2}$ is not separately identified from $%
\sigma _{1}^{2}(w_{p})$ and $\sigma _{2}^{2}\left( \varrho _{t}\right) $.
The error $\left\{ u_{pt}\right\} _{t\in \mathbb{Z}}$, $p\in \mathbb{N}^{+}$
is assumed to be independent of the regressors $\left\{ \acute{x}%
_{pt}\right\} _{t\in \mathbb{Z}}$, $p\in N^{+}$, $\left\{ w_{p}\right\}
_{p\in \mathbb{N}^{+}}$ and $\left\{ \varrho _{t}\right\} _{t\in \mathbb{Z}}$%
, see also footnote \ref{FN1}. Here the errors $v_{pt}$\ permit conditional
heteroskedasticity. This is an extension of the so-called \emph{groupwise
heteroskedasticity}, where observations belonging to different groups have
distinct variances, see for instance Greene $\left( 2018\right) $. This type
of heteroskedasticity is not uncommon in applications such as in development
economics, where it has been suggested that observations within a villages
or strata would have the same (conditional) variance while differences over
villages or strata exist (Deaton, $1996$), that is the variance depends on
some specific village variable(s).

Before we modify our Condition $C3$ to ensure we can permit this
generalisation, it is useful to introduce some notation. We shall denote $%
\left\{ \ddot{x}_{pt}\right\} _{t\in \mathbb{Z}},p\in \mathbb{N}^{+}$, the
sequence that applies the usual transformation to remove the fixed effects,
see (\ref{not}), to the sequence $\left\{ \acute{x}_{pt}\right\} _{t\in 
\mathbb{Z}},p\in \mathbb{N}^{+},$ such that 
\begin{equation*}
\ddot{x}_{pt}=\acute{x}_{pt}-\frac{1}{n}\sum_{q=1}^{n}\acute{x}_{qt}-\frac{1%
}{T}\sum_{s=1}^{T}\acute{x}_{ps}+\frac{1}{nT}\sum_{q=1}^{n}\sum_{s=1}^{T}%
\acute{x}_{qs}.
\end{equation*}%
Observe that as it happens with $\tilde{x}_{pt},$ we can take $E\left( \ddot{%
x}_{pt}\right) =0.$ Our new Condition $C3^{\prime }$ is given next.

\begin{description}
\item[Condition $\mathbf{C3}^{\prime }$] \emph{For all }$p\in \mathbb{N}^{+}$%
\emph{, the sequence }$\left\{ u_{pt}\right\} _{t\in \mathbb{Z}}$\emph{\ is
independent of }$\left\{ \acute{x}_{pt}\right\} _{t\in \mathbb{Z}},$\emph{\ }%
$\left\{ w_{p}\right\} _{p\in \mathbb{N}}$ \emph{and }$\left\{ \varrho
_{t}\right\} _{t\in \mathbb{Z}}$ \emph{and}%
\begin{equation}
0<\max_{1\leq p\leq n}\sum_{q=1}^{n}\left\Vert \varphi \left( p,q\right)
\right\Vert <\infty \text{\emph{,}}  \label{c3prime}
\end{equation}%
\emph{where }$\varphi \left( p,q\right) :=\varphi _{u}\left( p,q\right)
\varphi _{\ddot{x}}\left( p,q\right) $ \emph{and} 
\begin{equation*}
\varphi _{\ddot{x}}\left( p,q\right) =\text{\textsl{Cov}}\left( \sigma
_{1}\left( w_{p}\right) \ddot{x}_{pt};\sigma _{1}\left( w_{q}\right) \ddot{x}%
_{qt}^{\prime }\right) \text{\emph{, for any} }p,q\geq 1\text{.}
\end{equation*}
\end{description}

The requirement given in (\ref{c3prime}) limits the combined cross-sectional
dependence in $v_{pt}$ and $\ddot{x}_{pt}$ ($\acute{x}_{pt}$) needed to
ensure the existence of a consistent estimator of the \textquotedblleft 
\emph{average}\textquotedblright\ long-run variance of the sequences $%
\left\{ z_{pt}=:v_{pt}\ddot{x}_{pt}\right\} _{t\in \mathbb{Z}}$, $p\in 
\mathbb{N}^{+}$ in this framework. This is an obvious extension of our
previous Condition $C3$, since our expression for $\Phi $ under our
generalisation becomes%
\begin{eqnarray}
\Phi &=&:\lim_{n\rightarrow \infty }\lim_{T\rightarrow \infty }\frac{1}{nT}%
E\left\{ \left( \sum_{p=1}^{n}\sum_{t=1}^{T}\ddot{x}_{pt}v_{pt}\right)
\left( \sum_{p=1}^{n}\sum_{t=1}^{T}\ddot{x}_{pt}^{\prime }v_{pt}\right)
\right\}  \notag \\
&=&\lim_{n\rightarrow \infty }\lim_{T\rightarrow \infty }\frac{1}{nT}%
\sum_{p,q=1}^{n}\sum_{t,s=1}^{T}E\left( \left\{ \sigma _{2}\left( \varrho
_{t}\right) \sigma _{2}\left( \varrho _{s}\right) \right\} \left\{ \sigma
_{1}\left( w_{p}\right) \ddot{x}_{pt}\right\} \left\{ \sigma _{1}\left(
w_{q}\right) \ddot{x}_{qs}^{\prime }\right\} \right) E\left(
u_{pt}u_{qs}\right) .  \label{v_1mod}
\end{eqnarray}

We shall now give some examples. We can allow 
\begin{equation}
\left( \mathbf{i}\right) \text{ \ \ }\acute{x}_{pt}=:x_{pt}+h_{1}\left(
w_{p}\right) +h_{2}\left( \varrho _{t}\right) ~\text{and }\left( \mathbf{ii}%
\right) \text{ \ }\acute{x}_{pt}=:h\left( w_{p};\varrho _{t}\right) x_{pt}%
\text{,}  \label{xprime}
\end{equation}%
for given $h_{1},$ $h_{2}$\ and $h$, where $x_{pt}$ satisfies Condition $C2$%
. It is clear that under the additive structure in $\mathbf{(i)}$, the
transformed variables that account for the fixed effects, recall (\ref{not}%
), satisfy 
\begin{equation*}
\widetilde{\acute{x}}_{pt}=:\ddot{x}_{pt}\equiv \widetilde{x}_{pt}\text{, }%
p=1,...,n;~t=1,...,T\text{,}
\end{equation*}%
which renders this, potentially, the most straightforward setting. In this
case we have%
\begin{eqnarray*}
\Phi &=&:\lim_{n\rightarrow \infty }\lim_{T\rightarrow \infty }\frac{1}{nT}%
E\left\{ \left( \sum_{p=1}^{n}\sum_{t=1}^{T}\ddot{x}_{pt}v_{pt}\right)
\left( \sum_{p=1}^{n}\sum_{t=1}^{T}\ddot{x}_{pt}^{\prime }v_{pt}\right)
\right\} \\
&=&\lim_{n\rightarrow \infty }\lim_{T\rightarrow \infty }\frac{1}{nT}%
E\left\{ \left( \sum_{p=1}^{n}\sum_{t=1}^{T}x_{pt}\sigma _{1}\left(
w_{p}\right) \sigma _{2}\left( \varrho _{t}\right) u_{pt}\right) \left(
\sum_{p=1}^{n}\sum_{t=1}^{T}x_{pt}^{\prime }\sigma _{1}\left( w_{p}\right)
\sigma _{2}\left( \varrho _{t}\right) u_{pt}\right) \right\} \\
&=&\lim_{n\rightarrow \infty }\lim_{T\rightarrow \infty }\frac{1}{nT}%
\sum_{p,q=1}^{n}\sum_{t,s=1}^{T}E\left( \text{\H{x}}_{pt}\text{\H{x}}%
_{qs}^{\prime }\right) E\left( u_{pt}u_{qs}\right) ,
\end{eqnarray*}%
where \H{x}$_{pt}=x_{pt}\sigma _{1}\left( w_{p}\right) \sigma _{2}\left(
\varrho _{t}\right) $. The behaviour of the second moments of \H{x}$_{pt}$
are essentially those of $x_{pt}$ because 
\begin{equation*}
Cov\left( \text{\H{x}}_{pt},\text{\H{x}}_{qs}\right) =E\left( \sigma
_{1}\left( w_{p}\right) \sigma _{1}\left( w_{q}\right) \right) E\left(
\sigma _{2}\left( \varrho _{t}\right) \sigma _{2}\left( \varrho _{s}\right)
\right) Cov\left( x_{pt},x_{qs}\right) \text{.}
\end{equation*}

\begin{remark}
In the last displayed assumption we have only assumed that 
\begin{equation*}
E\left( x_{pt}\mid w_{p},\varrho _{t}\right) =0;\text{ \ \ }E\left(
x_{pt}x_{qs}\mid w_{p},w;\varrho _{t},\varrho _{s}\right) =E\left(
x_{pt}x_{qs}\right) =:Cov\left( x_{pt},x_{qs}\right)
\end{equation*}%
so some type of dependence between $x_{pt}$ and $\left( w_{p},\varrho
_{t}\right) $ is still allowed.
\end{remark}

With the multiplicative structure in $\left( \mathbf{ii}\right) $, it is
basically the same since 
\begin{eqnarray*}
\Phi &=&:\lim_{n\rightarrow \infty }\lim_{T\rightarrow \infty }\frac{1}{nT}%
E\left\{ \left( \sum_{p=1}^{n}\sum_{t=1}^{T}\ddot{x}_{pt}v_{pt}\right)
\left( \sum_{p=1}^{n}\sum_{t=1}^{T}\ddot{x}_{pt}^{\prime }v_{pt}\right)
\right\} \\
&=&\lim_{n\rightarrow \infty }\lim_{T\rightarrow \infty }\frac{1}{nT}%
E\left\{ \left( \sum_{p=1}^{n}\sum_{t=1}^{T}x_{pt}h\left( w_{p};\varrho
_{t}\right) \sigma _{1}\left( w_{p}\right) \sigma _{2}\left( \varrho
_{t}\right) u_{pt}\right) \left( \sum_{p=1}^{n}\sum_{t=1}^{T}x_{pt}^{\prime
}h\left( w_{p};\varrho _{t}\right) \sigma _{1}\left( w_{p}\right) \sigma
_{2}\left( \varrho _{t}\right) u_{pt}\right) \right\} \\
&=&\lim_{n\rightarrow \infty }\lim_{T\rightarrow \infty }\frac{1}{nT}%
\sum_{p,q=1}^{n}\sum_{t,s=1}^{T}E\left( \text{\H{x}}_{pt}\text{\H{x}}%
_{qs}^{\prime }\right) E\left( u_{pt}u_{qs}\right) ,
\end{eqnarray*}%
where now \H{x}$_{pt}=h\left( w_{p};\varrho _{t}\right) \sigma _{1}\left(
w_{p}\right) \sigma _{2}\left( \varrho _{t}\right) x_{pt}$,and $\left\vert
Cov\left( \text{\H{x}}_{pt},\text{\H{x}}_{qs}\right) \right\vert \leq
K\left\vert Cov\left( x_{pt},x_{qs}\right) \right\vert $ using Markov
inequality.\ The same caveats mentioned in the last remark apply in this
case.

We now turn to the consistent estimator of the \textquotedblleft \emph{%
average}\textquotedblright\ long-run variance of the sequences $\left\{
z_{pt}=:v_{pt}\ddot{x}_{pt}\right\} _{t\in \mathbb{Z}}$, $p\in \mathbb{N}%
^{+} $ in this framework (recognizing that we have established the necessary
regularity conditions for its existence). Following a rescaling of our
regressors, 
\begin{equation*}
\dot{x}_{pt}=:\ddot{x}_{pt}\sigma _{1}\left( w_{p}\right) \sigma _{2}\left(
\varrho _{t}\right) \text{,}
\end{equation*}%
for given $\sigma _{1}\left( w_{p}\right) \sigma _{2}\left( \varrho
_{t}\right) ,$ our estimator for $\Phi ,$ see also $\left( \ref{v_estimate}%
\right) ,$ becomes 
\begin{equation}
\breve{\Phi}=\frac{1}{T}\sum_{j=1}^{T-1}\left\{ \left( \frac{1}{n^{1/2}}%
\sum_{p=1}^{n}\mathcal{J}_{\widetilde{\dot{x}},p}\left( \lambda _{j}\right) 
\mathcal{J}_{\hat{u},p}\left( -\lambda _{j}\right) \right) \left( \frac{1}{%
n^{1/2}}\sum_{p=1}^{n}\mathcal{J}_{\widetilde{\dot{x}},p}^{\prime }\left(
-\lambda _{j}\right) \mathcal{J}_{\hat{u},p}\left( \lambda _{j}\right)
\right) \right\} \text{,}  \label{v_estimate_robust}
\end{equation}%
where $\hat{u}_{pt}:=\hat{v}_{pt}/\left( \hat{\sigma}_{1}\left( w_{p}\right) 
\hat{\sigma}_{2}(\varrho _{t}\right) )$ and $\widetilde{\dot{x}}_{pt}=\hat{%
\sigma}_{1}\left( w_{p}\right) \hat{\sigma}_{2}(\varrho _{t})\widetilde{%
\acute{x}}_{pt}.$ Implementation of this estimator only requires a
consistent estimator of $\sigma _{2}^{2}\left( \varrho _{t}\right) $ (up to
unknown scale of proportionality), and a natural estimator we can use is 
\begin{equation*}
\widehat{\sigma }_{2}^{2}\left( \varrho _{t}\right) =\frac{1}{n}%
\sum_{p=1}^{n}\hat{v}_{pt}^{2}\text{.}
\end{equation*}%
The estimator for $\sigma _{1}^{2}\left( w_{p}\right) $, $\hat{\sigma}%
_{1}^{2}\left( w_{p}\right) $, indeed cancels out when considering the
product $\mathcal{J}_{\widetilde{\dot{x}},p}\left( \lambda _{j}\right) 
\mathcal{J}_{\hat{u},p}\left( -\lambda _{j}\right) $, as 
\begin{eqnarray*}
\mathcal{J}_{\widetilde{\dot{x}},p}\left( \lambda _{j}\right) \mathcal{J}_{%
\widehat{u},p}\left( -\lambda _{j}\right) &=&\frac{1}{T^{1/2}}\sum_{t=1}^{T}%
\widetilde{\acute{x}}_{pt}\hat{\sigma}_{1}\left( w_{p}\right) \hat{\sigma}%
_{2}\left( \varrho _{t}\right) e^{-it\lambda _{j}}\sum_{t=1}^{T}\frac{%
\widehat{v}_{pt}}{\widehat{\sigma }_{2}\left( \varrho _{t}\right) \hat{\sigma%
}_{1}\left( w_{p}\right) }e^{-it\lambda _{j}} \\
&=&\frac{1}{T^{1/2}}\sum_{t=1}^{T}\widetilde{\acute{x}}_{pt}\widehat{\sigma }%
_{2}\left( \varrho _{t}\right) e^{-it\lambda _{j}}\sum_{t=1}^{T}\frac{%
\widehat{v}_{pt}}{\widehat{\sigma }_{2}\left( \varrho _{t}\right) }%
e^{-it\lambda _{j}}
\end{eqnarray*}%
Moreover, this result shows that when $\sigma _{2}\left( \varrho _{t}\right) 
$\ is a constant, our results in Section 2 and 3 continue to hold true. That
is, our estimators in previous section are robust to groupwise
heteroskedasticity in the cross-sectional unit, a result supported by our
Monte-Carlo simulations in Table 4 in the next section.

The intuition of the validity of this estimator comes form the standard
observation that 
\begin{equation*}
\frac{\widehat{\sigma }_{2}\left( \varrho _{t}\right) }{\sigma
_{2}^{2}\left( \varrho _{t}\right) }\overset{P}{\rightarrow }1,
\end{equation*}%
so that 
\begin{equation*}
\widehat{u}_{pt}=\frac{\widehat{v}_{pt}}{\widehat{\sigma }_{2}\left( \varrho
_{t}\right) }\simeq \frac{v_{pt}}{\widehat{\sigma }_{2}\left( \varrho
_{t}\right) }=\frac{v_{pt}}{\sigma _{2}\left( \varrho _{t}\right) }\left(
1+o_{p}\left( 1\right) \right) =:\sigma _{1}\left( w_{p}\right) u_{pt}\left(
1+o_{p}\left( 1\right) \right) \text{,}
\end{equation*}%
and 
\begin{equation*}
\frac{\widehat{v}_{pt}}{\widehat{\sigma }_{2}\left( \varrho _{t}\right)
\sigma _{1}\left( w_{p}\right) }=:u_{pt}\left( 1+o_{p}\left( 1\right) \right)
\end{equation*}%
from the above arguments. Of course the details can be lengthy, but
otherwise has been done in other contexts many times.

Our bootstrap algorithms also require some obvious and minimal change. The
only adjustment to the wild bootstrap algorithm relates to the use of the
robust estimator of $\Phi $ provided in (\ref{v_estimate_robust}). For the na%
\"{\i}ve bootstrap, a straightforward modification involves the following
steps

\begin{description}
\item[\emph{STEP 1}$^{\prime \prime }$] Obtain the residuals%
\begin{equation*}
\widehat{v}_{pt}=\widetilde{y}_{pt}-\widetilde{\beta }^{\prime }\widetilde{%
\acute{x}}_{pt}\text{, \ }p=1,...,n;\text{ }t=1,...,T,
\end{equation*}%
compute $\widehat{\sigma _{1}^{2}(w_{p})\sigma _{2}^{2}(\varrho _{t})}%
=T^{-1}\sum_{t=1}^{T}\widehat{v}_{pt}^{2}\cdot n^{-1}\sum_{p=1}^{n}\widehat{v%
}_{pt}^{2}$, and obtain the standardized residuals%
\begin{equation*}
\widehat{u}_{pt}=\widehat{v}_{pt}/\widehat{\sigma _{1}(w_{p})\sigma
_{2}(\varrho _{t})}
\end{equation*}

\item[\emph{STEP 2}$^{\prime \prime }$] Denoting $\hat{U}_{t}=\left\{ \hat{u}%
_{pt}\right\} _{p=1}^{n}$, do standard random sampling from the empirical
distribution of the residuals $\{\hat{U}_{t}\}_{t=1}^{T}$. That is, we
assign probability $T^{-1}$ to each $n\times 1$ vector $\hat{U}_{t}$. Denote
the bootstrap sample by $\left\{ U_{t}^{\ast }\right\} _{t=1}^{T},$ where $%
U_{t}^{\ast }=\left\{ u_{pt}^{\ast }\right\} _{p=1}^{n}.$ Let $V_{t}^{\ast
}=\left\{ \widehat{\sigma _{1}(w_{p})\sigma _{2}(\varrho _{t})}u_{pt}^{\ast
}\right\} _{p=1}^{n}$Compute the bootstrap analogue of $\left( \text{\ref%
{model_1R}}\right) $ as%
\begin{equation*}
\mathcal{J}_{y\ast ,p}\left( \lambda _{j}\right) =\widetilde{\beta }^{\prime
}\mathcal{J}_{\widetilde{\acute{x}},p}\left( \lambda _{j}\right) +\left( 
\frac{1}{n}\sum_{q=1}^{n}\mathcal{I}_{\hat{u},q}\left( \lambda _{j}\right)
\right) ^{1/2}\mathcal{J}_{v^{\ast },p}\left( \lambda _{j}\right)
\end{equation*}%
\ for $p=1,...,n$ \ and $j=1,...,T-1.$

\item[\emph{STEP 3}$^{\prime \prime }$] Compute the corresponding bootstrap
analogue of $\left( \text{\ref{beta_fef}}\right) $ as 
\begin{equation*}
\widetilde{\beta }^{\ast }=\left( \sum_{p=1}^{n}\sum_{j=1}^{T-1}\mathcal{J}_{%
\widetilde{\acute{x}},p}\left( \lambda _{j}\right) \mathcal{J}_{\widetilde{%
\acute{x}},p}^{\prime }\left( -\lambda _{j}\right) \right) ^{-1}\left(
\sum_{p=1}^{n}\sum_{j=1}^{T-1}\mathcal{J}_{\widetilde{\acute{x}},p}\left(
\lambda _{j}\right) \mathcal{J}_{\widetilde{y}\ast ,p}\left( -\lambda
_{j}\right) \right) ,
\end{equation*}%
with $\mathcal{J}_{\widetilde{y}\ast ,p}\left( \lambda _{j}\right) =\mathcal{%
J}_{y\ast ,p}\left( \lambda _{j}\right) -\frac{1}{n}\sum_{q=1}^{n}\mathcal{J}%
_{y\ast ,q}\left( \lambda _{j}\right) $
\end{description}

\begin{remark}
Step 2$^{\prime \prime }$ assumes that the temporal dependence of the $%
n\times 1$ vector $\{v_{pt}/(\sigma _{1}(w_{p})\sigma _{2}(\varrho
_{t}))\}_{p=1}^{n}$ is homogeneous so we can use the average periodogram to
impose proper dependence structure on $u_{pt}^{\ast }$ (drawings from the
empirical distribution of $\{\hat{v}_{pt}/(\widehat{\sigma _{1}(w_{p})\sigma
_{2}(\varrho _{t})})\}_{p=1}^{n}$).\medskip
\end{remark}

We now discuss the scenario where the conditional moment of the error term
depends on the regressors $\acute{x}_{pt}=:x_{pt}$\ themselves, i.e.,%
\begin{equation*}
y_{pt}=\beta ^{\prime }\acute{x}_{pt}+\eta _{p}+\alpha _{t}+v_{pt},\text{
with }v_{pt}=:\sigma (\acute{x}_{pt})u_{pt}
\end{equation*}%
As mentioned in the introduction, this would require us to estimate the
conditional expectation, $\sigma ^{2}\left( \acute{x}_{pt}\right) ,$\
nonparametrically. Several methods are available such as the Kernel
regression method or sieve estimation. As this approach would require the
selection of a bandwidth parameter which we set out to avoid in this paper,
we do not consider this in detail although we outline how to proceed.
Regardless of the approach used, we anticipate that the estimator would be
pretty accurate as the number of observations in large panel data will
normally be huge. For instance in a typical data set, with $T=20$\ and $%
n=1000,$\ we can use $20,000$\ observations to estimate the nonparametric
function. The estimator for $\Phi ,$\ see also $\left( \ref%
{v_estimate_robust}\right) ,$\ becomes 
\begin{equation*}
\breve{\Phi}=\frac{1}{T}\sum_{j=1}^{T-1}\left\{ \left( \frac{1}{n^{1/2}}%
\sum_{p=1}^{n}\mathcal{J}_{\widetilde{\acute{x}},p}\left( \lambda
_{j}\right) \mathcal{J}_{\hat{u},p}\left( -\lambda _{j}\right) \right)
\left( \frac{1}{n^{1/2}}\sum_{p=1}^{n}\mathcal{J}_{\acute{x},p}^{\prime
}\left( -\lambda _{j}\right) \mathcal{J}_{\hat{u},p}\left( \lambda
_{j}\right) \right) \right\} \text{,}
\end{equation*}%
where $\hat{u}_{pt}:=\hat{v}_{pt}/\hat{\sigma}\left( \acute{x}_{pt}\right) $%
\ and $\widetilde{\dot{x}}_{pt}=\hat{\sigma}\left( \acute{x}_{pt}\right) 
\widetilde{x}_{pt}.$\ For the associated na\"{\i}ve bootstrap procedure, we
can proceed as above where $\widehat{\sigma _{1}(w_{p})\sigma _{2}(\varrho
_{t})}$\ is replaced by $\hat{\sigma}\left( \acute{x}_{pt}\right) .$

\section{\textbf{FINITE SAMPLE BEHAVIOUR}}

In this section, we discuss the finite sample performance of our
cluster-based inference procedure in the presence of cross-sectional and
temporal dependence of unknown form. We contrast this performance with the
HAC-based inference procedure proposed by Driscoll and Kraay (1989), which
unlike ours, requires the choice of smoothing parameters that may be
arbitrary and erroneous. We also provide evidence of the potential finite
sample improvements of our frequency domain bootstrap schemes and implement
the MBB time domain bootstrap to the vector containing all individual
observations for each $t$. Our frequency domain approaches have the benefit
that they do not rely on the choice of any smoothing parameter or require an
ordering of cross-sectional units, which, as we argued before, may be
arbitrary and erroneous. Another benefit of our estimator we address in our
simulations is the fact that our estimator permits heterogeneity in the
temporal dependence. In our simulations, we also consider a multiplicative
error structure that permits groupwise heteroskedasticity and we reveal the
robustness of our estimator to this setting.

In our Monte-Carlo experiments, we first consider the following data
generating process%
\begin{equation*}
y_{pt}=\alpha _{t}+\eta _{p}+\beta x_{pt}+u_{pt}\text{ for }p=1,...,n\text{
and }t=1,..,T.
\end{equation*}%
The time fixed effects $\alpha _{t}$ and individual fixed effects $\eta _{p}$
are drawn independently $(\alpha _{t}\sim IIDN(1,1)$ and $\eta _{p}\sim
IIDN(1,1)$) and are held fixed across replications and without loss of
generality $\beta $ is set equal to zero. The independently drawn errors and
regressors are postulated to exhibit a variety of scenarios for the temporal
and cross-sectional dependence that are assumed to be the same for
simplicity.

To evaluate the performance of our proposed cluster estimator, we analyze
the empirical size and power for testing the significance of our
parameter,~\ $H_{0}:\beta =0$ against $H_{A}:\beta \neq 0$, at the nominal
5\% level for various pairs of $n$ and $T$ using 5,000 simulations. In
addition to presenting the rejection rates based on the asymptotic
distribution of the Wald statistic $nT\hat{\beta}_{FE}^{\prime }\hat{V}^{-1}%
\hat{\beta}_{FE},$\ with $\widehat{V}=:\widetilde{\Sigma }_{x}^{-1}\breve{%
\Phi}\widetilde{\Sigma }_{x}^{-1}$\ where $\breve{\Phi}$\ is defined in (\ref%
{v_estimate}) (or equivalently the asymptotic t-test as $\beta $ is scalar),
we present rejection rates based on the empirical distribution of the
bootstrapped test statistic%
\begin{equation*}
nT\left( \hat{\beta}_{FE}^{\ast }-\hat{\beta}_{FE}\right) ^{\prime }\left[ 
\QTR{sl}{\hat{V}}^{\ast }\right] ^{-1}\left( \hat{\beta}_{FE}^{\ast }-\hat{%
\beta}_{FE}\right) ,
\end{equation*}%
where $\hat{\beta}_{FE}^{\ast }$ and $\hat{V}^{\ast }$\ are the bootstrapped
estimators of $\beta $ and $V$ defined in Section 3\textbf{. }As inference
based on the asymptotic distribution might not provide a good approximation
to the finite sample one, this allows us to assess the finite sample
improvements our bootstrap schemes may yield.

We compare the finite sample performance of our cluster-based inference
procedure to the HAC based inference procedure and select the bandwidth
parameter, denoted $m_{T}$, using the parametric AR(1) plug-in method
suggested in Andrews $\left( 1991\right) .$\footnote{$m_{T}$ is chosen to be
upward rounded integers.}\textbf{\ }This lag window is designed to minimize
(approximately) the mean square of the standard error.\footnote{%
Kiefer and Vogelsang (2002) discuss the use of HAC estimators with bandwidth
equal to the sample size $(b=1)$. This bandwidth free approach does come at
the cost of power relative to Andrews' popular data driven optimal bandwidth
selection, see also\ Vogelsang (2012).} For the HAC based inference we
provide rejection rates of the Wald statistic $nT\hat{\beta}_{FE}^{\prime }%
\hat{V}_{m_{T}}^{-1}\hat{\beta}_{FE}$ based on the asymptotic critical
values (asy) and the critical values based on the fixed-b asymptotics (fixb)
of Kiefer and Vogelsang (2005) as this is shown to lead to more reliable
inference, see also Vogelsang $(2012)$. With $\widehat{V}_{m_{T}}=:%
\widetilde{\Sigma }_{x}^{-1}\hat{\Phi}_{m_{T}}\widetilde{\Sigma }_{x}^{-1},$ 
$\hat{\Phi}_{m_{T}}$ is defined as%
\begin{equation*}
\hat{\Phi}_{m_{T}}=\frac{1}{nT}\sum_{p=1}^{n}\sum_{q=1}^{n}\sum_{t=1}^{T}%
\sum_{s=1}^{T}K\left( \frac{\left\vert t-s\right\vert }{m_{T}}\right) \hat{z}%
_{pt}\hat{z}_{qs}^{\prime }\text{,}
\end{equation*}%
where $\hat{z}_{pt}=\widetilde{x}_{pt}\widehat{u}_{pt}$ and $K(h)=\left(
1-\left\vert h\right\vert \right) \mathbf{1}\left( \left\vert h\right\vert
\leq 1\right) $ is the Bartlett kernel. The fixed-b asymptotic distribution
is non-standard and our critical values are obtained by simulation.\footnote{%
Let $W_{q}(r)$ denote a $q$ dimensional vector of independent standard
Wiener processes and define $\tilde{W}_{q}(r)=W_{q}\left( r\right)
-rW_{q}\left( 1\right) .$ The limiting distribution of the t-test is $%
W_{1}(1)/\sqrt{C_{1}}$ with $C_{q}=\frac{2}{b}\int_{0}^{1}\tilde{W}_{q}(r)%
\tilde{W}_{q}(r)^{\prime }dr$ $-\frac{1}{b}\int_{0}^{1-b}\left[ \tilde{W}%
_{q}(r+b)\tilde{W}_{q}(r)^{\prime }+\tilde{W}_{q}(r)\tilde{W}%
_{q}(r+b)^{\prime }\right] dr$ given the use of the Bartlet kernel, where $%
b\in (0,1]$ with $m_{T}=bT$ (see Theorem 4, Vogelsang, 2012); the limiting
distribution of the Wald test is $W_{q}(1)^{\prime }C_{q}^{-1}W_{q}(1).$ We
obtain the critical values using 500,000 simulations.}

We also provide critical values for HAC based inference that rely on the
pairs moving block bootstrap proposed by Gon\c{c}alves\ $\left( 2011\right)
. $ She obtained bootstrapped samples $z_{it}^{\ast }=(y_{it}^{\ast
},x_{it}^{\ast \prime })^{\prime }$ by arranging $k$ resampled blocks of $%
\ell $ observations from the set of $T-\ell +1$\ overlapping blocks $\left\{
B_{1,\ell },..,B_{T-\ell +1,\ell }\right\} $\ with $B_{t,\ell }=\left\{
z_{t,n},z_{t+1,n},..,z_{t+\ell -1,n}\right\} $\ and $z_{t,n}=\left(
z_{1t},..,z_{nt}\right) ^{\prime }$ in sequence (for notational simplicity $%
T=k\ell $). When $\ell =1$\ this corresponds to the standard iid bootstrap
on $\left\{ z_{t,n}\right\} _{t=1}^{T}$. The MMB based critical value are
based on the standardized test statistic $T\left( \hat{\beta}_{FE}^{\ast }-%
\hat{\beta}_{FE}\right) ^{\prime }\allowbreak \left[ \QTR{sl}{\hat{V}}_{\ell
}^{\ast }\right] ^{-1}\allowbreak \left( \hat{\beta}_{FE}^{\ast }-\hat{\beta}%
_{FE}\right) .$ Here $\QTR{sl}{\hat{V}}_{\ell }^{\ast }=\left( \widetilde{%
\Sigma }_{x}^{\ast }\right) ^{-1}\breve{\Phi}_{\ell }^{\ast }\left( 
\widetilde{\Sigma }_{x}^{\ast }\right) ^{-1}$, $\widetilde{\Sigma }%
_{x}^{\ast }=\frac{1}{nT}\tsum\nolimits_{p=1}^{n}\tsum\nolimits_{t=1}^{T}%
\widetilde{x}_{pt}^{\ast }\widetilde{x}_{pt}^{\ast }$ and 
\begin{equation*}
\breve{\Phi}_{\ell }^{\ast }=\frac{1}{k}\tsum\nolimits_{j=1}^{k}\left( \ell
^{-1/2}\tsum\nolimits_{t=1}^{\ell }n^{-1}\hat{s}_{n,\left( j-1\right) \ell
+t}^{\ast }\right) \left( \ell ^{-1/2}\tsum\nolimits_{t=1}^{\ell }n^{-1}\hat{%
s}_{n,\left( j-1\right) \ell +t}^{\ast }\right) ^{\prime },
\end{equation*}%
where $\hat{s}_{nt}^{\ast }=\tsum\nolimits_{p=1}^{n}\widetilde{x}_{pt}^{\ast
}\left( \tilde{y}_{pt}^{\ast }-\tilde{x}_{pt}^{\ast \prime }\hat{\beta}%
_{FE}^{\ast }\right) $ with $\widetilde{y}_{pt}^{\ast }=y_{pt}^{\ast }-%
\overline{y}_{p\cdot }^{\ast }-\overline{y}_{\cdot t}^{\ast }+\overline{%
\overline{y}}_{\cdot \cdot }^{\ast }$\ and $\widetilde{x}_{pt}^{\ast
}=x_{pt}^{\ast }-\overline{x}_{p\cdot }^{\ast }-\overline{x}_{\cdot t}^{\ast
}+\overline{\overline{x}}_{\cdot \cdot }^{\ast }$ (see also G\"{o}tze and K%
\"{u}nsch, 1996). The block size used is given by the integer part of the
automatic bandwidth chosen by the Andrews (1991) as proposed by Gon\c{c}%
alves\ $\left( 2011\right) .$

\subsection{Simulations with Homogeneous Time Dependence}

In the first set of simulations, we assume that the time dependence is
homogenous among individuals $p=1,..,n$. In particular, we assume that the
error and regressors are mutually independent, homoscedastic, first order
auto regressive random variables with $\rho =0.7$\ or $\rho =0.9.$ The error
term, therefore, takes the form%
\begin{equation*}
u_{pt}=\rho u_{p,t-1}+\sqrt{1-\rho ^{2}}\eta _{pt}\text{, with }\rho =0.7,0.9
\end{equation*}%
where $\eta _{pt}$ characterizes the spatial dependence inherent in the
error.\footnote{%
We generated the spatial data with $49+T$\ periods and take the last $T$\
periods as our sample using $0$\ as the starting value.} We consider both a
weak and strong cross-sectional dependence scenarios for $u_{pt}$ $(\eta
_{pt})$. To describe the cross-sectional dependence, we follow Lee and
Robinson $\left( 2013\right) $ and draw random locations for individual
units along a line, denoted $s=\left( s_{1},...s_{n}\right) ^{\prime }$ with 
$s_{p}\sim IIDU[0,n]$ for $p=1,..,n$. Using the linear time dependence
representation, $\eta _{pt}=\sigma _{p}\left( \sum_{\ell =1}^{\infty
}c_{\ell }\left( p\right) e_{\ell t}\right) $ with $e_{\ell t}\sim
IIDN(0,1), $ we set $c_{\ell }(p)=(1+|s_{\ell }-s_{p}|_{+})^{-10}$ to permit
weak dependence; $\sigma _{p}$ is such that $Var(\eta _{pt})=1$. For the
strong spatial dependence setting, we use $c_{\ell }(p)=(1+|s_{\ell
}-s_{p}|_{+})^{-0.7}$ instead, see also Hidalgo and Schafgans (2017). The
same discussion holds for the independently drawn, strictly exogenous
regressor $x_{pt}$, where, to allow for some time heterogeneity, we may,
without loss of generality add $\mu _{t},$ which is independently drawn $%
\left( \mu _{t}\sim IIDN(1,1)\right) .$

In Table \ref{MC1}, we report the empirical size for testing the
significance of $\beta $ at the 5\% level of significance based on our
cluster estimator of the variance of $\tilde{\beta}$ in the columns labelled 
$HS$ (Cluster). In addition to presenting the rejection rates based on the
asymptotic critical values (asy), we report the empirical size based on the
na\"{\i}ve bootstrap (nb), and the wild bootstrap (nb). The empirical size
based on the HAC based inference procedure proposed by Driscoll and Kraay
are reported in the columns labelled $DK$ (HAC). For the HAC based
inference, we provide rejection rates based on the asymptotic critical
values (asy), the critical values based on the fixed-b asymptotics (fixb) of
Kiefer and Vogelsang (2005), and Gon\c{c}alves'\textbf{\ }$\left(
2011\right) $ MBB (mbb). We used the parametric AR(1) plug-in method
suggested by Andrews $\left( 1991\right) $ to determine the window lag $%
m_{T} $ and the block length $\ell $.\renewcommand{\thetable}{\arabic{table}}%
\begin{table}[tbp] \centering%
\caption{Monte Carlo Simulations with Homogeneous Time Dependence \newline Empirical size of test for significance of $\beta$}%
\label{MC1}%
\begin{tabular}{|c|cccccc|cccccc|}
\hline
Spatial & \multicolumn{6}{|c}{Weak Spatial Dependence} & 
\multicolumn{6}{|c|}{Strong Spatial Dependence} \\ 
Dependence & \multicolumn{6}{|c}{} & \multicolumn{6}{|c|}{} \\ \hline
\multicolumn{1}{|l|}{\ $\quad $Estimator} & \multicolumn{3}{|c}{$\,HS$
(Cluster)} & \multicolumn{3}{|c}{$DK$ (HAC)} & \multicolumn{3}{|c}{$\,HS$
(Cluster)} & \multicolumn{3}{|c|}{$DK$ (HAC)} \\ \hline
\multicolumn{1}{|l|}{} & asy & nb & wb & \multicolumn{1}{|c}{asy} & fixb & 
mbb & asy & nb & wb & \multicolumn{1}{|c}{asy} & fixb & mbb \\ \hline
\multicolumn{1}{|l|}{} & \multicolumn{11}{|c}{} &  \\ 
\multicolumn{1}{|l|}{$(n,T)$} & \multicolumn{12}{|c|}{Time Dependence:
AR(1), $\rho =0.7$} \\ 
\multicolumn{1}{|l|}{$(50,16)$} & .180 & .074 & .134 & .253 & .163 & .028 & 
.177 & .068 & .133 & .261 & .176 & .030 \\ 
\multicolumn{1}{|l|}{$(50,32)$} & .126 & .067 & .091 & .192 & .131 & .042 & 
.129 & .056 & .091 & .210 & .148 & .043 \\ 
\multicolumn{1}{|l|}{$(50,64)$} & .080 & .054 & .068 & .128 & .092 & .049 & 
.091 & .050 & .076 & .158 & .119 & .056 \\ 
\multicolumn{1}{|l|}{$(50,128)$} & .067 & .049 & .062 & .108 & .084 & .056 & 
.068 & .046 & .057 & .116 & .089 & .060 \\ 
\multicolumn{1}{|l|}{$(50,256)$} & .055 & .048 & .055 & .087 & .073 & .058 & 
.060 & .051 & .050 & .096 & .083 & .065 \\ 
\multicolumn{1}{|l|}{} &  &  &  &  &  &  &  &  &  &  &  &  \\ 
\multicolumn{1}{|l|}{$(100,16)$} & .172 & .070 & .120 & .249 & .153 & .033 & 
.183 & .073 & .134 & .261 & .174 & .031 \\ 
\multicolumn{1}{|l|}{$(100,32)$} & .122 & .057 & .094 & .185 & .126 & .050 & 
.121 & .053 & .088 & .200 & .143 & .037 \\ 
\multicolumn{1}{|l|}{$(100,64)$} & .082 & .056 & .070 & .132 & .098 & .064 & 
.096 & .055 & .084 & .153 & .110 & .054 \\ 
\multicolumn{1}{|l|}{$(100,128)$} & .065 & .047 & .056 & .108 & .082 & .066
& .072 & .052 & .060 & .114 & .091 & .062 \\ 
\multicolumn{1}{|l|}{$(100,256)$} & .058 & .050 & .063 & .088 & .074 & .065
& .062 & .054 & .058 & .089 & .078 & .059 \\ 
\multicolumn{1}{|l|}{} &  &  &  &  &  &  &  &  &  &  &  &  \\ 
\multicolumn{1}{|l|}{$(n,T)$} & \multicolumn{12}{|c|}{Time Dependence:\
AR(1), $\rho =0.9$} \\ 
\multicolumn{1}{|l|}{$(50,16)$} & .320 & .131 & .276 & .410 & .258 & .009 & 
.312 & .106 & .257 & .415 & .279 & .013 \\ 
\multicolumn{1}{|l|}{$(50,32)$} & .242 & .097 & .189 & .327 & .209 & .013 & 
.260 & .093 & .201 & .368 & .246 & .022 \\ 
\multicolumn{1}{|l|}{$(50,64)$} & .168 & .058 & .107 & .261 & .169 & .026 & 
.174 & .068 & .124 & .281 & .195 & .037 \\ 
\multicolumn{1}{|l|}{$(50,128)$} & .111 & .057 & .084 & .192 & .132 & .046 & 
.115 & .059 & .089 & .199 & .147 & .050 \\ 
\multicolumn{1}{|l|}{$(50,256)$} & .081 & .055 & .067 & .142 & .107 & .062 & 
.085 & .055 & .069 & .149 & .114 & .061 \\ 
\multicolumn{1}{|l|}{} &  &  &  &  &  &  &  &  &  &  &  &  \\ 
\multicolumn{1}{|l|}{$(100,16)$} & .316 & .125 & .254 & .414 & .255 & .007 & 
.302 & .130 & .253 & .400 & .268 & .011 \\ 
\multicolumn{1}{|l|}{$(100,32)$} & .252 & .084 & .204 & .350 & .224 & .017 & 
.242 & .086 & .173 & .344 & .229 & .015 \\ 
\multicolumn{1}{|l|}{$(100,64)$} & .174 & .067 & .118 & .249 & .167 & .026 & 
.174 & .069 & .139 & .269 & .188 & .033 \\ 
\multicolumn{1}{|l|}{$(100,128)$} & .112 & .054 & .091 & .181 & .131 & .052
& .118 & .060 & .083 & .198 & .140 & .056 \\ 
\multicolumn{1}{|l|}{$(100,256)$} & .075 & .049 & .068 & .132 & .096 & .057
& .088 & .047 & .071 & .146 & .115 & .061 \\ \hline
\end{tabular}%
\end{table}%

The results from Table \ref{MC1} reveal that our cluster based inference
performs remarkably well even in the presence of strong cross sectional
dependence for moderately large panels. As before, the rejection rates based
on the asymptotic critical values tend to be closer to the nominal rejection
rates as $n$ and $T$ increase. The finite sample performance using these
asymptotic critical values does suffer, in particular, from $T$ being small,
more so when the temporal dependence is stronger.\textbf{\ }This suggests
that the cluster variance's finite sample performance, in particular,
appears to require larger $T$, in order for us to be able to rely on the
asymptotic critical value. Nevertheless, finite sample improvements in
inference can be made using either frequency domain bootstrap schemes as
rejection rates based on them are typically closer to the nominal rejection
rates, with the differences typically smaller as sample sizes increase.
Given that we assume the temporal dynamics to be the same for all
individuals in this simulation, both bootstrap schemes are valid. The na%
\"{\i}ve bootstrap approach tends to perform better in the sense of
providing a size closer to the nominal rejection rate.

Our cluster based inference, using the na\"{\i}ve bootstrap for small
panels, suggests large improvements in size relative to HAC based inference.
While the use of fixed-b asymptotic critical values for HAC based inference
does indeed improve its performance, in accordance with Vogelsang, $\left(
2012\right) $, the gains in improvement in size achieved by our cluster
based estimator remain significant and are larger when the temporal or
spatial dependence is stronger. Our cluster based inference, however, does
not necessarily perform superior to the HAC based inference that use the
critical values based on Gon\c{c}alves' pairwise MBB. Her approach indeed
performs very well in this setting where the temporal dependence is
homogenous across individuals. As we will see in Table \ref{MC3} relaxing
this assumption, which is more realistic, does reveal a marked improvement
of our cluster based performance over HAC based inference using the MBB. But
even in the homogenous setting, it should be noted that the MBB approach is
sensitive to the chosen blocksize, and its selection here was appropriate
given the imposed AR(1) temporal dependence (which is unknown in practical
applications). Contrary to the MBB we do not need to choose a block size.

In Table \ref{MC2}, we present the empirical power of our test for the
significance of the slope when $\beta =0.1$ for a selection of $(n,T)$ pairs
and compare the performance of our cluster-based inference procedure to the
HAC\ based inference procedure proposed by Driscoll and Kraay as before.%
\begin{table}[t] \centering%
\caption{Monte Carlo Simulations with Homogeneous Time Dependence \newline Empirical Power of test for signifcance of $\beta$ when $\beta=0.1$}%
\label{MC2}%
\begin{tabular}{|c|cccccc|cccccc|}
\hline
Spatial & \multicolumn{6}{|c}{Weak Spatial Dependence} & 
\multicolumn{6}{|c|}{Strong Spatial Dependence} \\ 
Dependence & \multicolumn{6}{|c}{} & \multicolumn{6}{|c|}{} \\ \hline
\multicolumn{1}{|l|}{\ $\quad $Estimator} & \multicolumn{3}{|c}{$\,HS$
(Cluster)} & \multicolumn{3}{|c}{$DC$ (HAC)} & \multicolumn{3}{|c}{$\,HS$
(Cluster)} & \multicolumn{3}{|c|}{$DK$ (HAC)} \\ \hline
\multicolumn{1}{|l|}{} & asy & nb & wb & \multicolumn{1}{|c}{asy} & fixb & 
mbb & asy & nb & wb & \multicolumn{1}{|c}{asy} & fixb & mbb \\ \hline
\multicolumn{1}{|l|}{} & \multicolumn{11}{|c}{} &  \\ 
\multicolumn{1}{|l|}{$(n,T)$} & \multicolumn{12}{|c|}{Time Dependence:
AR(1), $\rho =0.7$} \\ 
\multicolumn{1}{|l|}{$(50,64)$} & .852 & .794 & .830 & .919 & .882 & .802 & 
.243 & .158 & .212 & .337 & .277 & .162 \\ 
\multicolumn{1}{|l|}{$(50,128)$} & .980 & .971 & .979 & .992 & .986 & .978 & 
.386 & .322 & .357 & .465 & .428 & .356 \\ 
\multicolumn{1}{|l|}{$(50,256)$} & 1.00 & 1.00 & 1.00 & 1.00 & 1.00 & 1.00 & 
.549 & .527 & .525 & .619 & .588 & .558 \\ 
\multicolumn{1}{|l|}{} &  &  &  &  &  &  &  &  &  &  &  &  \\ 
\multicolumn{1}{|l|}{$(100,64)$} & .972 & .959 & .967 & .991 & .987 & .971 & 
.318 & .239 & .297 & .418 & .354 & .241 \\ 
\multicolumn{1}{|l|}{$(100,128)$} & 1.00 & 1.00 & 1.00 & 1.00 & 1.00 & 1.00
& .451 & .395 & .426 & .539 & .495 & .429 \\ 
\multicolumn{1}{|l|}{$(100,256)$} & 1.00 & 1.00 & 1.00 & 1.00 & 1.00 & 1.00
& .690 & .662 & .675 & .747 & .721 & .679 \\ 
\multicolumn{1}{|l|}{} &  &  &  &  &  &  &  &  &  &  &  &  \\ 
\multicolumn{1}{|l|}{$(n,T)$} & \multicolumn{12}{|c|}{Time Dependence:
AR(1), $\rho =0.9$} \\ 
\multicolumn{1}{|l|}{$(50,64)$} & .605 & .401 & .513 & .707 & .605 & .232 & 
.248 & .116 & .195 & .359 & .268 & .065 \\ 
\multicolumn{1}{|l|}{$(50,128)$} & .716 & .575 & .659 & .812 & .745 & .539 & 
.253 & .164 & .214 & .361 & .287 & .133 \\ 
\multicolumn{1}{|l|}{$(50,256)$} & .884 & .849 & .867 & .942 & .917 & .856 & 
.290 & .227 & .253 & .381 & .333 & .231 \\ 
\multicolumn{1}{|l|}{} &  &  &  &  &  &  &  &  &  &  &  &  \\ 
\multicolumn{1}{|l|}{$(100,64)$} & .784 & .602 & .710 & .872 & .799 & .408 & 
.287 & .142 & .238 & .402 & .305 & .082 \\ 
\multicolumn{1}{|l|}{$(100,128)$} & .926 & .857 & .905 & .969 & .949 & .867
& .288 & .184 & .230 & .401 & .321 & .152 \\ 
\multicolumn{1}{|l|}{$(100,256)$} & .995 & .989 & .994 & .999 & .998 & .995
& .348 & .256 & .319 & .453 & .393 & .272 \\ \hline
\end{tabular}%
\end{table}%

The results from Table \ref{MC2} show that our cluster based inference has
good power to reject $H_{0}:\beta =0$ when $\beta =0.1$ in both time
dependence scenarios, even for small panels, in particular when the spatial
dependence is not strong. For the reported sample sizes, the cluster based
inference using the na\"{\i}ve bootstrap only showed limited size
distortions. As expected its power approaches one as the sample size, and
therefore the precision of our estimator, increases. This improved power
performance comes about faster when the cross sectional and/or temporal
dependence is lower and improved power performance appears stronger with
increases in $T$ relative to $n.$ The power for our cluster based inference,
using the na\"{\i}ve bootstrap, compares well with that of power of HAC
based inference. Where the size-distortions for HAC\ based inference are
smallest, any apparent power loss of cluster based inference disappears.
Both cluster based inference and HAC\ based inference have a comparable loss
of power when both spatial and temporal dependence is large.

\subsection{Simulations with Heterogeneous Time Dependence}

In our second set of simulations, we allow individual heterogeneity in the
time dependence of the error and the strictly exogenous regressor. The error
term $u_{pt}$ is generated using various heterogeneous ARMA processes 
\begin{equation*}
(1-\rho _{1,p}L)(1+\rho _{2}L+\rho _{3}L^{2})u_{pt}=(1+\theta _{1,p}L+\theta
_{2}L^{2}+\theta _{3}L^{3})\eta _{pt},
\end{equation*}%
with $L\ $denoting the lag operator, such that, e.g., $Lu_{pt}=u_{p,t-1},$ $%
\rho _{1,p}$ and $\theta _{1,p}$ are individual specific AR and MA
coefficients, and $(\rho _{2},\rho _{3})$ and $\left( \theta _{2},\theta
_{3}\right) $ are additional non-varying higher order AR and MA
coefficients. As before $\eta _{pt}$ characterizes the spatial dependence. A
similar description holds for the independently drawn, strictly exogenous
regressor, $x_{it},$ which is assumed to have the same spatial temporal
dependence as the error for simplicity. We allow the variance of $u_{pt}$ to
vary across individuals $p=1,...,n$.

We consider four heterogeneous specifications: Mixed AR(1), Mixed
AR(1)/MA(1), Mixed AR(3), and Mixed AR(3)/MA(3). The individual specific
parameters $\rho _{1,p}$ and $\theta _{1,p}$, where non-zero, reflect
equidistant points on $[0.5,0.9].$ The full details of these heterogeneous
specifications are provided at the bottom of Table \ref{MC3}.%
\begin{table}[t] \centering%
\caption{Monte Carlo Simulations with Heterogeneous Time Dependence  \newline Empirical size of test for significance of $\beta$}%
\label{MC3}%
\begin{tabular}{|c|cccccc|cccccc|}
\hline
Spatial & \multicolumn{6}{|c|}{Weak Spatial Dependence} & 
\multicolumn{6}{|c|}{Strong Spatial Dependence} \\ 
Dependence & \multicolumn{6}{|c}{} & \multicolumn{6}{|c|}{} \\ \hline
\multicolumn{1}{|l|}{\ $\quad $Estimator} & \multicolumn{3}{|c}{$\,HS$
(Cluster)} & \multicolumn{3}{|c}{$DC$ (HAC)} & \multicolumn{3}{|c}{$\,HS$
(Cluster)} & \multicolumn{3}{|c|}{$DK$ (HAC)} \\ \hline
\multicolumn{1}{|l|}{} & asy & nb & wb & \multicolumn{1}{|c}{asy} & fixb & 
mbb & asy & nb & wb & \multicolumn{1}{|c}{asy} & fixb & mbb \\ \hline
\multicolumn{1}{|l|}{} & \multicolumn{12}{|c|}{} \\ 
\multicolumn{1}{|l|}{$(n,T)$} & \multicolumn{12}{|c|}{Time Dependence: Mixed
AR(1)} \\ 
\multicolumn{1}{|l|}{$(100,64)$} & .101 & .055 & .079 & .189 & .132 & .062 & 
.114 & .065 & .094 & .207 & .148 & .052 \\ 
\multicolumn{1}{|l|}{$(100,128)$} & .082 & .055 & .071 & .144 & .114 & .078
& .079 & .055 & .067 & .146 & .116 & .065 \\ 
\multicolumn{1}{|l|}{$(100,256)$} & .064 & .046 & .052 & .121 & .098 & .073
& .069 & .053 & .066 & .117 & .098 & .073 \\ 
\multicolumn{1}{|l|}{} &  &  &  &  &  &  &  &  &  &  &  &  \\ 
\multicolumn{1}{|l|}{$(n,T)$} & \multicolumn{12}{|c|}{Time Dependence: Mixed
AR(1)/MA(1)} \\ 
\multicolumn{1}{|l|}{$(100,64)$} & .080 & .049 & .066 & .176 & .132 & .083 & 
.092 & .061 & .087 & .185 & .136 & .064 \\ 
\multicolumn{1}{|l|}{$(100,128)$} & .068 & .051 & .059 & .149 & .120 & .097
& .069 & .053 & .063 & .140 & .110 & .074 \\ 
\multicolumn{1}{|l|}{$(100,256)$} & .058 & .049 & .052 & .112 & .095 & .084
& .067 & .053 & .063 & .111 & .095 & .081 \\ 
\multicolumn{1}{|l|}{} &  &  &  &  &  &  &  &  &  &  &  &  \\ 
\multicolumn{1}{|l|}{$(n,T)$} & \multicolumn{12}{|c|}{Time Dependence: Mixed
AR(3)} \\ 
\multicolumn{1}{|l|}{$(100,64)$} & .064 & .051 & .062 & .147 & .135 & .120 & 
.074 & .055 & .062 & .150 & .134 & .125 \\ 
\multicolumn{1}{|l|}{$(100,128)$} & .057 & .048 & .056 & .146 & .137 & .122
& .062 & .052 & .056 & .143 & .134 & .121 \\ 
\multicolumn{1}{|l|}{$(100,256)$} & .053 & .048 & .050 & .137 & .132 & .127
& .063 & .061 & .063 & .138 & .133 & .135 \\ 
\multicolumn{1}{|l|}{} &  &  &  &  &  &  &  &  &  &  &  &  \\ 
\multicolumn{1}{|l|}{$(n,T)$} & \multicolumn{12}{|c|}{Time Dependence: Mixed
AR(3)/MA(3)} \\ 
\multicolumn{1}{|l|}{$(100,64)$} & .068 & .048 & .062 & .134 & .113 & .094 & 
.072 & .058 & .068 & .148 & .129 & .106 \\ 
\multicolumn{1}{|l|}{$(100,128)$} & .057 & .049 & .049 & .109 & .094 & .089
& .063 & .052 & .057 & .120 & .107 & .091 \\ 
\multicolumn{1}{|l|}{$(100,256)$} & .053 & .046 & .050 & .094 & .087 & .080
& .061 & .056 & .061 & .106 & .098 & .094 \\ \hline
\multicolumn{13}{|l|}{Note: With $(1-\rho _{1,p}L)(1+\rho _{2}L+\rho
_{3}L^{2})u_{pt}=(1+\theta _{1,p}L+\theta _{2}L^{2}+\theta _{3}L^{3})\eta
_{pt},$ the follo-} \\ 
\multicolumn{13}{|l|}{wing parameterisations are used: Denoting $\rho
_{p}=\left( \rho _{1,p},\rho _{2},\rho _{3}\right) ^{\prime }$ and $\theta
_{p}=\left( \theta _{1,p},\theta _{2},\theta _{3}\right) ^{\prime }$} \\ 
\multicolumn{13}{|c|}{} \\ 
\multicolumn{13}{|l|}{Mixed AR(1): $\left\{ \rho _{p}=\left( 0.5+0.4\tfrac{%
p-1}{n-1},0,0\right) ^{\prime },\theta _{p}=0\right\} _{p=1}^{n}$} \\ 
\multicolumn{13}{|l|}{$\text{Mixed AR(1)/MA(1): }%
\begin{array}[t]{l}
\left\{ \rho _{p}=\left( 0.5+0.4\tfrac{p-1}{n/2-1},0,0\right) ^{\prime
},\theta _{p}=0\right\} _{p=1}^{n/2} \\ 
\left\{ \rho _{p}=0,\theta _{p}=\left( 0.5+0.4\tfrac{p-n/2-1}{n/2-1}%
,0,0\right) \right\} _{p=n/2+1}^{n}%
\end{array}%
$} \\ 
\multicolumn{13}{|l|}{Mixed AR(3): $\left\{ \rho _{p}=\left( 0.5+0.4\tfrac{%
p-1}{n-1},0.3,0.6\right) ^{\prime },\theta _{p}=0\right\} _{p=1}^{n}$} \\ 
\multicolumn{13}{|l|}{$\text{Mixed AR(3)/MA(3): }%
\begin{array}[t]{l}
\left\{ \rho _{p}=\left( 0.5+0.4\tfrac{p-1}{n/2-1},0.3,0.6\right) ^{\prime
},\theta _{p}=0\right\} _{p=1}^{n/2} \\ 
\left\{ \rho _{p}=0\text{, }\theta _{p}=\left( 0.5+0.4\tfrac{p-n/2-1}{n/2-1}%
,0.3,0.6\right) \right\} _{p=1}^{n/2}%
\end{array}%
$} \\ \hline
\end{tabular}%
\end{table}%

In Table \ref{MC3}, we report the empirical size for testing the
significance of $\beta $ in the presence of heterogenous time dependence for
panels where $n=100$ and $T=$ $64,$ $128,$ and $256.$ As before, we consider
both weak and strong spatial dependence scenarios. For HAC based inference
we used the parametric AR(1) plug-in method suggested by Andrews $\left(
1991\right) $ again, to determine the window lag $m_{T}$ and the block
length $\ell .$ A common approach, which does not recognize the temporal
heterogeneity nor the higher order (autoregressive) nature of the temporal
dependence under consideration.

The results in Table \ref{MC3} show that our cluster estimator of the
variance is robust to the presence of individual specific time dependence.
The rejection rates based on the asymptotic critical values in the
heterogeneous AR(1) time dependence setting, with $\left\{ \rho
_{1,p}\right\} _{p=1}^{n}$ in the range $[0.5,0.9],$ are comparable to the
rejection rates in the homogenous AR(1) setting with $\rho =0.7.$\footnote{%
Associated simulations considering the power to reject $H_{0}:\beta =0$ when 
$\beta =0.1$ in the presence of heterogenous temporal dependence show
comparable results as in the homogenous time dependence setting, see also
Hidalgo and Schafgans (2018).} As in the homogeneous time dependence
setting, the rejection rates based on the asymptotic critical values
approach the nominal rejection rate of 5\% as the sample size increases. The
rejection rates based on both frequency-based bootstrap schemes show that
finite sample improvements in inference can be made. The improvements
achieved when applying the wild bootstrap, proven to be valid in the
heterogeneous time dependence scenario, are more modest than those suggested
by the na\"{\i}ve bootstrap, which assumes homogeneous time dependence. Our
cluster based inference reveals a similar pattern when we permit higher
order heterogeneous autoregressive/moving average temporal dependence with
the na\"{\i}ve bootstrap performing remarkably well again, suggesting that
the na\"{\i}ve bootstrap may be robust to violations of the homogeneous time
dependence such as those considered in these simulations. Whereas the wild
bootstrap does perform less well then expected, in particular in the
presence of strong spatial dependence, the discrepancy between the rejection
rates based on the two bootstrap schemes does appear to be smaller than in
the homogeneous time dependence scenario.

Importantly, our cluster based inference suggests large improvements over
HAC\ based inference in these heterogeneous time dependence settings,
whether we use the asymptotic, the fixed-b asymptotic critical values or
base its rejection rates on the MBB. The inferior HAC based inference may be
explained by the inappropriate use of a single smoothing parameter in these
heterogeneous settings, as is common practice, in addition to the fact that
the parametric AR(1)\ plug-in method does not account for other, and
possible higher order (autoregressive) processes, than AR(1). Our cluster
based inference benefits from not requiring the choice of any smoothing
parameter, and is therefore not subject to this deterioration in size. Aside
from the ease of implementation, the robustness of our approach to the
presence of individual specific time dependence is a particularly attractive
feature of our cluster robust inference.

\subsection{Simulations with (Conditional) Heteroskedasticity}

Finally, we consider simulations that make use of our "modified" cluster
based inference that permits general forms of heteroskedasticity. Here the
data generating process is given by 
\begin{equation*}
y_{pt}=\alpha _{t}+\eta _{p}+\beta \acute{x}_{pt}+\sigma _{1}\left(
w_{p})\sigma _{2}(\varrho _{t}\right) u_{pt}\text{ for }p=1,...,n\text{ and }%
t=1,..,T.
\end{equation*}%
We consider both an additive and multiplicative specification for $\acute{x}%
_{pt}$, in particular%
\begin{equation*}
\acute{x}_{pt}=x_{pt}+w_{p}+\varrho _{t}\text{ and }\acute{x}%
_{pt}=x_{pt}\left( w_{p}\varrho _{t}\right) ^{2}.
\end{equation*}%
Here $u_{pt}$ and $x_{pt}$ are drawn independently with weak temporal and
weak cross-sectional dependence; $w_{p}$ and $\varrho _{t}$ are additional
regressors where $w_{p}$ exhibits strong spatial dependence and $\varrho
_{t} $ follows an AR(1) with coefficient equal to $0.7$. Without loss of
generality $\beta =0$ again. Due to the presence of the multiplicative error 
$\sigma _{1}\left( w_{p}\right) \sigma _{2}\left( \varrho _{t}\right)
u_{pt}:=v_{t}$, this setting does permit (conditional) heteroskedasticity
with $Var\left( v_{pt}|x_{pt},w_{p}\right) =\sigma _{1}^{2}\left(
w_{p}\right) \sigma _{2}^{2}\left( \varrho _{t}\right) $ after a
normalization of the variance of $u_{pt}$ to one, for simplicity. In
particular, we consider 
\begin{equation*}
\sigma _{1}\left( w_{p}\right) \sigma _{2}\left( \varrho _{t}\right) =\sigma
\cdot \left[ \exp (\delta _{1}w_{p})+1\right] \left[ \exp (\delta
_{2}\varrho _{t})+1\right] \text{ with }\delta _{1}=0.5,\text{ }2.0\text{
and }\delta _{2}=0,0.2,0.5.
\end{equation*}%
The severity of heteroskedasticity, which we can measure using the
coefficient of variation of $\sigma _{1}^{2}\left( w_{p}\right) \sigma
_{2}^{2}\left( \varrho _{t}\right) ,$ increases with the values of $\delta
_{1}$ and $\delta _{2}.$ The coefficient of variation is defined as the
ratio of the standard deviation of $\sigma _{1}^{2}\left( w_{p}\right)
\sigma _{2}^{2}\left( \varrho _{t}\right) $ to its mean. The average
coefficient of variation of $\sigma _{1}^{2}\left( w_{p}\right) \sigma
_{2}^{2}\left( \varrho _{t}\right) $ over our simulations with $\delta
_{2}=0.5$ ranges from 42\% ($\delta _{1}=0.5)$ to 260\% ($\delta _{1}=2.0).$
The constant $\sigma $ in chosen in such a way that the expected variability
of $\sigma _{1}^{2}\left( w_{p}\right) \sigma _{2}^{2}\left( \varrho
_{t}\right) ,$ equals one for comparability across simulations.

In Table \ref{MC4}, we report the empirical size for testing the
significance of $\beta $ in the presence of (conditional)
heteroskedasticity. The average coefficient of variation for each
specification across the simulations is given in the first column. We
provide two sets of simulations for our cluster based inference: first we
apply the original cluster based inference, which is robust to the presence
of heteroskedasticity that is only cross-sectional in nature, followed by
the heteroskedasticity robust cluster based inference. In the top panel, we
report the results based on the additive specification of the regressor. The
multiplicative specification of the regressors is in the bottom panel. As
before, we will compare the empirical size of our (robust) cluster based
inference with the HAC based inference, in particular those using the MMB
based critical values. As we impose an AR(1) temporal dependence, we have
ensured that the use of the parametric AR(1) plug-in method suggested by
Andrews $\left( 1991\right) $ to determine the window lag $m_{T}$ and the
block length $\ell $ required for this approach is suitable.

\begin{table}[t] \centering%
\caption{Monte Carlo Simulations with Conditional Heteroskedasticity  \newline Empirical size of test for significance of $\beta$}%
\label{MC4}%
\begin{tabular}{|c|c|ccccccccc|}
\hline
Spatial Dep. &  & \multicolumn{9}{|c|}{Weak Spatial Dependence} \\ 
&  & \multicolumn{9}{|c|}{} \\ \hline
\multicolumn{1}{|l|}{\ $\quad $Estimator} &  & \multicolumn{3}{|c}{$\,HS$
(Cluster)$^{\text{orig}}$} & \multicolumn{3}{c}{$HS$ (Cluster)$^{robust}$} & 
\multicolumn{3}{|c|}{$DC$ (HAC)} \\ \hline
\multicolumn{1}{|l|}{} &  & asy & nb & wb & \multicolumn{1}{|c}{asy} & nb & 
wb & \multicolumn{1}{|c}{asy} & fixb & mbb \\ \hline
& CV & \multicolumn{9}{|c|}{$\acute{x}$ - additive} \\ 
\multicolumn{1}{|l|}{$(n,T)$} &  & \multicolumn{9}{|c|}{$\sigma \left(
w_{p},\varrho _{t}\right) =\sigma \cdot \left[ \exp (0.5w_{p})+1\right] %
\left[ \exp (0.2\varrho _{t})+1\right] $} \\ 
\multicolumn{1}{|l|}{$(100,64)$} & .428 & .085 & .054 & .068 & .089 & .054 & 
.068 & .138 & .100 & .052 \\ 
\multicolumn{1}{|l|}{$(100,128)$} & .424 & .069 & .053 & .058 & .071 & .052
& .060 & .110 & .089 & .065 \\ 
\multicolumn{1}{|l|}{$(100,256)$} & .432 & .060 & .052 & .059 & .062 & .052
& .060 & .090 & .076 & .061 \\ 
\multicolumn{1}{|l|}{$(n,T)$} &  & \multicolumn{9}{|c|}{$\sigma \left(
w_{p},\varrho _{t}\right) =\sigma \cdot \left[ \exp (0.5w_{p})+1\right] %
\left[ \exp (0.5\varrho _{t})+1\right] $} \\ 
\multicolumn{1}{|l|}{$(100,64)$} & .733 & .086 & .054 & .068 & .089 & .052 & 
.069 & .132 & .100 & .055 \\ 
\multicolumn{1}{|l|}{$(100,128)$} & .760 & .069 & .053 & .060 & .073 & .056
& .061 & .109 & .087 & .061 \\ 
\multicolumn{1}{|l|}{$(100,256)$} & .788 & .059 & .050 & .060 & .060 & .050
& .060 & .086 & .074 & .060 \\ 
\multicolumn{1}{|l|}{$(n,T)$} &  & \multicolumn{9}{|c|}{$\sigma \left(
w_{p},\varrho _{t}\right) =\sigma \cdot \left[ \exp (2w_{p})+1\right] \left[
\exp (0.5\varrho _{t})+1\right] $} \\ 
\multicolumn{1}{|l|}{$(100,64)$} & 2.560 & .090 & .057 & .070 & .099 & .055
& .073 & .133 & .098 & .051 \\ 
\multicolumn{1}{|l|}{$(100,128)$} & 2.626 & .068 & .054 & .058 & .075 & .053
& .060 & .106 & .086 & .063 \\ 
\multicolumn{1}{|l|}{$(100,256)$} & 2.588 & .055 & .047 & .056 & .062 & .045
& .059 & .086 & .074 & .062 \\ 
\multicolumn{1}{|l|}{$(n,T)$} &  & \multicolumn{9}{|c|}{$\sigma \left(
w_{p},\varrho _{t}\right) =\sigma \cdot \left[ \exp (2w_{p})+1\right] $} \\ 
\multicolumn{1}{|l|}{$(100,64)$} & 2.117 & .088 & .052 & .070 & .093 & .050
& .070 & .134 & .101 & .056 \\ 
\multicolumn{1}{|l|}{$(100,128)$} & 2.135 & .068 & .052 & .056 & .074 & .050
& .057 & .106 & .085 & .057 \\ 
\multicolumn{1}{|l|}{$(100,256)$} & 2.074 & .059 & .052 & .059 & .066 & .054
& .063 & .088 & .075 & .064 \\ 
&  &  &  &  &  &  &  &  &  &  \\ \hline
\multicolumn{1}{|l|}{} & CV & \multicolumn{9}{|c|}{$\acute{x}$ -
multiplicative} \\ 
\multicolumn{1}{|l|}{$(n,T)$} &  & \multicolumn{9}{|c|}{$\sigma \left(
w_{p},\varrho _{t}\right) =\sigma \cdot \left[ \exp (0.5w_{p})+1\right] %
\left[ \exp (0.2\varrho _{t})+1\right] $} \\ 
\multicolumn{1}{|l|}{$(100,64)$} & .428 & .077 & .062 & .077 & .070 & .057 & 
.069 & .160 & .142 & .066 \\ 
\multicolumn{1}{|l|}{$(100,128)$} & .424 & .076 & .065 & .081 & .065 & .054
& .073 & .143 & .129 & .079 \\ 
\multicolumn{1}{|l|}{$(100,256)$} & .432 & .067 & .062 & .069 & .057 & .057
& .056 & .105 & .098 & .068 \\ 
\multicolumn{1}{|l|}{$(n,T)$} &  & \multicolumn{9}{|c|}{$\sigma \left(
w_{p},\varrho _{t}\right) =\sigma \cdot \left[ \exp (0.5w_{p})+1\right] %
\left[ \exp (0.5\varrho _{t})+1\right] $} \\ 
\multicolumn{1}{|l|}{$(100,64)$} & .733 & .128 & .113 & .125 & .076 & .063 & 
.071 & .162 & .140 & .051 \\ 
\multicolumn{1}{|l|}{$(100,128)$} & .760 & .134 & .122 & .137 & .063 & .052
& .064 & .141 & .131 & .055 \\ 
\multicolumn{1}{|l|}{$(100,256)$} & .788 & .141 & .137 & .145 & .057 & .057
& .057 & .105 & .098 & .056 \\ 
\multicolumn{1}{|l|}{$(n,T)$} &  & \multicolumn{9}{|c|}{$\sigma \left(
w_{p},\varrho _{t}\right) =\sigma \cdot \left[ \exp (2w_{p})+1\right] \left[
\exp (0.5\varrho _{t})+1\right] $} \\ 
\multicolumn{1}{|l|}{$(100,64)$} & 2.560 & .131 & .115 & .123 & .077 & .063
& .069 & .162 & .140 & .049 \\ 
\multicolumn{1}{|l|}{$(100,128)$} & 2.626 & .135 & .130 & .133 & .070 & .061
& .062 & .138 & .123 & .049 \\ 
\multicolumn{1}{|l|}{$(100,256)$} & 2.588 & .145 & .139 & .146 & .062 & .058
& .055 & .108 & .102 & .054 \\ 
&  & \multicolumn{9}{|c|}{$\sigma \left( w_{p},\varrho _{t}\right) =\sigma
\cdot \left[ \exp (2w_{p})+1\right] $} \\ 
\multicolumn{1}{|l|}{$(100,64)$} & 2.117 & .074 & .049 & .067 & .086 & .051
& .068 & .129 & .099 & .063 \\ 
\multicolumn{1}{|l|}{$(100,128)$} & 2.135 & .064 & .052 & .060 & .075 & .053
& .059 & .111 & .093 & .077 \\ 
\multicolumn{1}{|l|}{$(100,256)$} & 2.074 & .052 & .049 & .052 & .061 & .051
& .053 & .084 & .076 & .071 \\ \hline
\multicolumn{11}{|l|}{Note: The time dependence assumed is AR(1) with $\rho
=0.7.$} \\ 
\multicolumn{11}{|l|}{$\sigma $ is chosen to ensure that the expected $%
\sigma ^{2}\left( w_{p})\sigma ^{2}(\varrho _{t}\right) $ equals 1.} \\ 
\hline
\end{tabular}%
\end{table}%

The results in Table \ref{MC4} show that under the additive formulation of
the regressor, the performance of the cluster based inference and robust
cluster based inference (which accounts for a non-constant $\sigma
_{2}\left( \rho _{t}\right) )$, are quite similar. The robust cluster based
inference is required for our first three formulations, where $\sigma
(w_{p},\varrho _{t})=\sigma _{1}(w_{p})\sigma _{2}(\varrho _{t}),$ whereas
the final formulation permits the original cluster based inference. Compared
to the results in Table \ref{MC1} (case with weak spatial and temporal
dependence), the rejection rates in the presence of (conditional)
heteroskedasticity are only slightly larger (in part explained by the need
to use estimates for $\sigma _{2}\left( \rho _{t}\right) )$. Since $%
\widetilde{\acute{x}}_{pt}=\widetilde{x}_{pt}$ under the additive
formulation of the regressor, a comparison with the results in Table \ref%
{MC1} is more straightforward here than under the multiplicative
formulation. Rejection rates that rely on the na\"{\i}ve bootstrap compare
favourably with that of the HAC based inference that uses the MBB and there
does not appear a serious deterioration in the performance of the (robust)
cluster based inference when the severity of heteroskedasticity increases,
either via $\delta _{1}$ or $\delta _{2},$ in this setting.

Under the multiplicative formulation of the regressor, the performance of
the robust cluster based inference is clearly superior in our first three
specifications where robust cluster based inference is required. The cluster
based inference that does not account for (conditional) heteroskedasticity
that is not purely cross-sectional in nature (i.e., in the presence of
non-constant $\sigma _{2}\left( \rho _{t}\right) ),$ deteriorates quite
quickly with $\delta _{2}$ (parameter reflecting the severity of temporal
heteroskedasticity). The rejection rates that use the robust estimator of
the long-run variance, \ref{v_estimate_robust}, are much closer to the
nominal 5\% rejection rates, whether we use the asymptotic critical values
or the bootstrap algorithms. In fact, rejection rates that rely on the na%
\"{\i}ve bootstrap compare again quite well with the HAC based inference
that use the MBB, which reveals the robustness of our estimator to this type
of (conditional) heteroskedasticity. This is a welcome result, given that
our estimator is simple to apply and does not require the choice of any
smoothing parameter.

\section{\textbf{CONCLUSIONS}}

In this paper we extend the literature on inference in panel data models in
the presence of both temporal and cross-sectional dependence\ of unknown
form. While a standard methodology, based on the $HAC$ estimator, is often
invoked and used in the context of time series regression models, in the
presence of cross-sectional dependence its implementation has only recently
been considered, see Kim and Sun $\left( 2013\right) $, Driscoll and Kraay $%
\left( 1998\right) $ or Vogelsang $\left( 2012\right) $. To deal with
various potential caveats of the $HAC$\ estimator, we propose a cluster
based estimator which is able to take into account both types of dependence
and allows the temporal dependence to be heterogeneous across individuals,
extending the work of Arellano $\left( 1987\right) $\ and Driscoll and Kraay 
$\left( 1998\right) $\ in a substantial way. We provide a new CLT that
accounts for an unknown and general temporal spatial dependence structure
that permits strong spatial dependence. We thereby provide primitive
conditions that guarantee Kim and Sun's $\left( 2013\right) ,$ Driscoll and
Kraay's $\left( 1998\right) $ and Gon\c{c}alves' $\left( 2011\right) $
assumption of the existence of a suitable CLT.

Our approach is based on the insightful observation that the spectral
representation of the fixed effect panel data model is such that the errors
become approximately temporally uncorrelated and heteroskedastic allowing
the use of a cluster estimator of the long run variance in the frequency
domain. As the cluster estimator may not be reliable in small samples, and
therefore may not provide a good approximation to make accurate inferences,
we present and examine bootstrap schemes in the frequency domain that are
also bandwidth parameter free.

Our simulation results reveal that our cluster estimator performs quite well
even in the presence of strong spatial dependence. For large panels,
inference based on our cluster estimator is properly sized even in the
presence of heterogeneous time dependence unlike Driscoll and Kraay's HAC
based inference of cross sectional averages that ignores such heterogeneity.
Our bootstrap schemes provide small sample improvements, where inference
that use the na\"{\i}ve bootstrap, in particular, is well sized, and reveal
large improvements in size relative to HAC based inference when fixed-b
asymptotic critical values are used. Improvements over MBB based inference
are more limited, except in the presence of heterogeneous time dependence.
We have shown the robustness of our cluster based inference to the presence
of \textquotedblleft groupwise\textquotedblright\ heteroskedasticity. To
enable us to adapt to the presence of \textquotedblleft
groupwise\textquotedblright\ heteroskedasticity that is not purely
cross-sectional in nature, a simple robust cluster based inference procedure
was proposed that also does not require the selection of any smoothing
parameter. \bigskip

\renewcommand{\theequation}{A.\arabic{equation}} \setcounter{equation}{0}%
\renewcommand{\thelemma}{A.\arabic{lemma}} \setcounter{lemma}{0}%
\renewcommand{\theremark}{A.\arabic{remark}} \setcounter{remark}{0}

\begin{center}
{\Large \textbf{{\large {Appendix A: PROOF OF MAIN RESULTS}}}}
\end{center}

We first introduce some notation. For a generic function $h$, we shall
abbreviate $h\left( \lambda _{j}\right) $ by $h\left( j\right) $ and for
generic sequences $\left\{ \psi _{pt}\right\} _{t=1}^{T}$, $p=1,...,n$, 
\begin{equation*}
\mathcal{J}_{\overline{\psi },\cdot }\left( j\right) =\frac{1}{T^{1/2}}%
\sum_{t=1}^{T}\left( \frac{1}{n}\sum_{q=1}^{n}\psi _{qt}\right)
e^{-it\lambda _{j}}\text{.}
\end{equation*}%
Using expression $\left( 10.3.12\right) $ of Brockwell and Davis $\left(
1991\right) $, we also have the useful relation 
\begin{eqnarray}
\mathcal{J}_{u,p}\left( j\right) &=&\mathcal{B}_{u,p}\left( -j\right) 
\mathcal{J}_{\xi ,p}\left( j\right) +\mathrm{Y}_{u,p}\left( j\right)
\label{bartlett} \\
\mathcal{J}_{x,p}\left( j\right) &=&\mathcal{B}_{x,p}\left( -j\right) 
\mathcal{J}_{\chi ,p}\left( j\right) +\mathrm{Y}_{x,p}\left( j\right) \text{%
, \ \ \ \ }p=1,...,n\text{,}  \notag
\end{eqnarray}%
where $\mathcal{B}_{u,p}\left( j\right) =:\mathcal{B}_{u,p}\left(
e^{i\lambda _{j}}\right) $, $\mathcal{B}_{x,p}\left( j\right) =:\mathcal{B}%
_{x,p}\left( e^{i\lambda _{j}}\right) $ and 
\begin{eqnarray}
\mathrm{Y}_{u,p}\left( j\right) &=&\sum_{\ell =0}^{\infty }d_{\ell }\left(
p\right) e^{-i\ell \lambda _{j}}\left( \frac{1}{T^{1/2}}\left\{
\sum_{t=1-\ell }^{T-\ell }-\sum_{t=1}^{T}\right\} \xi _{pt}e^{-it\lambda
_{j}}\right)  \label{y_p} \\
\mathrm{Y}_{x,p}\left( j\right) &=&\sum_{\ell =0}^{\infty }c_{\ell }\left(
p\right) e^{-i\ell \lambda _{j}}\left( \frac{1}{T^{1/2}}\left\{
\sum_{t=1-\ell }^{T-\ell }-\sum_{t=1}^{T}\right\} \chi _{pt}e^{-it\lambda
_{j}}\right) \text{.}  \notag
\end{eqnarray}%
Finally, we shall make use of the well know result 
\begin{eqnarray}
E\mathcal{J}_{\chi ,p}\left( j\right) \mathcal{J}_{\chi ,q}\left( -k\right)
&=&\varphi _{x}\left( p,q\right) \mathbf{1}\left( j=k\right)  \label{dft_cov}
\\
E\mathcal{J}_{\xi ,p}\left( j\right) \mathcal{J}_{\xi ,q}\left( -k\right)
&=&\varphi _{u}\left( p,q\right) \mathbf{1}\left( j=k\right) \text{.}  \notag
\end{eqnarray}

\subsection{\textbf{PROOF OF THEOREM \protect\ref{ThmEst}}}

$\left. {}\right. $

For completeness, we provide the proof using the time domain estimator, $%
\hat{\beta},$\ and the frequency domain estimator,\textbf{\ }$\tilde{\beta}$%
\textbf{. }\newline
We begin with $\hat{\beta}$. Without loss of generality assume that $x_{pt}$
is scalar. Using $\left( \ref{not}\right) $ and standard arguments, we obtain%
\begin{eqnarray}
&&\sum_{t=1}^{T}\sum_{p=1}^{n}\widetilde{x}_{pt}\widetilde{u}%
_{pt}=\sum_{t=1}^{T}\sum_{p=1}^{n}x_{pt}u_{pt}-\sum_{t=1}^{T}\sum_{p=1}^{n}%
\left( \overline{x}_{\cdot t}+\overline{x}_{p\cdot }-\overline{x}_{\cdot
\cdot }\right) u_{pt}  \notag \\
&&-\sum_{t=1}^{T}\sum_{p=1}^{n}\left( \overline{u}_{\cdot t}+\overline{u}%
_{p\cdot }-\overline{u}_{\cdot \cdot }\right) x_{pt}+o_{p}\left( \left(
nT\right) ^{1/2}\right) \text{.}  \label{the_3}
\end{eqnarray}%
Because the second and third terms on the right of $\left( \ref{the_3}%
\right) $ are handled similarly, we shall only look at the second. Now%
\begin{eqnarray*}
E\left( \sum_{t=1}^{T}\sum_{p=1}^{n}\overline{x}_{\cdot t}u_{pt}\right) ^{2}
&=&\sum_{t,s=1}^{T}\sum_{p,q=1}^{n}E\left( \overline{x}_{\cdot t}\overline{x}%
_{\cdot s}\right) \gamma _{u,pq}\left( t-s\right) \varphi _{u}\left(
p,q\right) \\
&=&\frac{1}{n^{2}}\sum_{p_{2},q_{2},p_{1},q_{1}=1}^{n}\varphi _{x}\left(
p_{2},q_{2}\right) \varphi _{u}\left( p_{1},q_{1}\right)
\sum_{t,s=1}^{T}\gamma _{x,p_{2}q_{2}}\left( t-s\right) \gamma
_{u,p_{1}q_{1}}\left( t-s\right) \\
&\leq &C\frac{T}{n^{2}}\left( \sum_{p_{2},q_{2}=1}^{n}\left\vert \varphi
_{x}\left( p_{2},q_{2}\right) \right\vert \right) \left(
\sum_{p_{1},q_{1}=1}^{n}\left\vert \varphi _{u}\left( p_{1},q_{1}\right)
\right\vert \right) \\
&=&o\left( nT\right) .
\end{eqnarray*}%
The latter displayed expression holds true because Conditions $C1$ and $C2$
imply that 
\begin{equation}
\sum_{t,s=1}^{T}\sup_{p,q}\left\vert \gamma _{x,pq}\left( t-s\right)
\right\vert +\sup_{p,q}\left\vert \gamma _{u,pq}\left( t-s\right)
\right\vert <C\text{,}  \label{the_1}
\end{equation}%
whereas Condition $C3$, see also Remark \ref{Remark_1}, implies that%
\footnote{%
For two nonnegative sequences $\left\{ \alpha _{p}\right\} $ and $\left\{
\beta _{p}\right\} $, $\sum \alpha _{p}\beta _{p}<C$ implies that $\sum
\alpha _{p}\sum \beta _{p}=o\left( n\right) $ if $\sum \left( \alpha
_{p}+\beta _{p}\right) =o\left( n\right) $.} 
\begin{equation}
\sum_{q=1}^{n}\varphi _{u}\left( p,q\right) \sum_{q=1}^{n}\varphi _{x}\left(
p,q\right) =o\left( n\right)  \label{pp}
\end{equation}%
so that 
\begin{equation}
\sum_{p_{1},p_{2}=1}^{n}\varphi _{u}\left( p_{1},p_{2}\right)
\sum_{q_{1},q_{2}=1}^{n}\varphi _{x}\left( q_{1},q_{2}\right) =o\left(
n^{3}\right) \text{.}  \label{the_5}
\end{equation}

Proceeding similarly with $\sum_{t=1}^{T}\sum_{p=1}^{n}\overline{x}_{p\cdot
}u_{pt}$ and $\overline{x}_{\cdot \cdot }\sum_{t=1}^{T}\sum_{p=1}^{n}u_{pt}$%
, we can conclude using $\left( \ref{the_3}\right) $ that%
\begin{equation*}
\frac{1}{\left( nT\right) ^{1/2}}\sum_{t=1}^{T}\sum_{p=1}^{n}\widetilde{x}%
_{pt}\widetilde{u}_{pt}=\frac{1}{\left( nT\right) ^{1/2}}\sum_{t=1}^{T}%
\sum_{p=1}^{n}x_{pt}u_{pt}+o_{p}\left( 1\right) \overset{d}{\rightarrow }%
\mathcal{N}\left( 0,\Phi \right)
\end{equation*}%
by Lemma \ref{Lem1App}. From here it is standard to conclude that $\left(
nT\right) ^{1/2}\left( \widehat{\beta }-\beta \right) \rightarrow _{d}%
\mathcal{N}\left( 0,\Sigma ^{-1}\Phi \Sigma ^{-1}\right) $.

We now show that $\left( nT\right) ^{1/2}\left( \widetilde{\beta }-\beta
\right) \rightarrow _{d}\mathcal{N}\left( 0,\Sigma ^{-1}\Phi \Sigma
^{-1}\right) $. Proceeding similarly as we did above, we shall examine%
\begin{eqnarray}
&&\frac{1}{\left( nT\right) ^{1/2}}\sum_{p=1}^{n}\sum_{j=1}^{T-1}\mathcal{J}%
_{x,p}\left( j\right) \mathcal{J}_{u,p}\left( -j\right) -\frac{1}{\left(
nT\right) ^{1/2}}\sum_{p=1}^{n}\sum_{j=1}^{T-1}\mathcal{J}_{x,p}\left(
j\right) \mathcal{J}_{\overline{u},\cdot }\left( -j\right)  \label{the_4} \\
&&-\frac{1}{\left( nT\right) ^{1/2}}\sum_{p=1}^{n}\sum_{j=1}^{T-1}\mathcal{J}%
_{\overline{x},p}\left( j\right) \mathcal{J}_{u,\cdot }\left( -j\right) 
\text{.}  \notag
\end{eqnarray}%
The first term of $\left( \ref{the_4}\right) $ converges in distribution to $%
\mathcal{N}\left( 0,\Phi \right) $ by Lemma \ref{Lem3App}. So, to complete
the proof it suffices to show that the last two terms of $\left( \ref{the_4}%
\right) $ are $o_{p}\left( 1\right) $. We examine the second term only, with
the third term being handled similarly. By standard algebra and $\left( \ref%
{bartlett}\right) $, this term is%
\begin{eqnarray}
&&\frac{1}{n^{3/2}}\sum_{p,q=1}^{n}\frac{1}{T^{1/2}}\sum_{j=1}^{T-1}\mathcal{%
B}_{x,p}\left( j\right) \mathcal{B}_{u,q}\left( j\right) \mathcal{J}_{\chi
,p}\left( j\right) \mathcal{J}_{\xi ,q}\left( -j\right)  \notag \\
&&+\frac{1}{n^{3/2}}\sum_{p,q=1}^{n}\frac{1}{T^{1/2}}\sum_{j=1}^{T-1}%
\mathcal{B}_{x,p}\left( j\right) \mathcal{J}_{\chi ,p}\left( j\right)
\left\{ \mathcal{J}_{u,q}\left( -j\right) -\mathcal{B}_{u,q}\left( j\right) 
\mathcal{J}_{\xi ,q}\left( -j\right) \right\}  \notag \\
&&+\frac{1}{n^{3/2}}\sum_{p,q=1}^{n}\frac{1}{T^{1/2}}\sum_{j=1}^{T-1}%
\mathcal{B}_{u,p}\left( j\right) \mathcal{J}_{\xi ,q}\left( -j\right)
\left\{ \mathcal{J}_{x,q}\left( -j\right) -\mathcal{B}_{x,q}\left( j\right) 
\mathcal{J}_{\chi ,p}\left( j\right) \right\}  \label{a_5} \\
&&+\frac{1}{n^{3/2}}\sum_{p,q=1}^{n}\frac{1}{T^{1/2}}\sum_{j=1}^{T-1}\left( 
\mathcal{J}_{x,q}\left( -j\right) -\mathcal{B}_{x,q}\left( j\right) \mathcal{%
J}_{\chi ,p}\left( j\right) \right) \times \left( \mathcal{J}_{u,q}\left(
-j\right) -\mathcal{B}_{u,q}\left( j\right) \mathcal{J}_{\xi ,q}\left(
-j\right) \right) \text{.}  \notag
\end{eqnarray}%
We examine the second term of $\left( \ref{a_5}\right) $ first. Using $%
\left( \ref{dft_cov}\right) $, we have that its second moment is bounded by 
\begin{eqnarray*}
&&\frac{1}{Tn^{3}}\sum_{p_{1},p_{2},q_{1},q_{2}=1}^{n}\varphi _{u}\left(
q_{1},q_{2}\right) \varphi _{x}\left( p_{1},p_{2}\right) \frac{1}{T}%
\sum_{j=1}^{T-1}\sup_{p_{1},p_{2}}\left\vert f_{x,p_{1}p_{2}}\left( j\right)
\right\vert \\
&=&\frac{1}{Tn^{3}}\sum_{q_{1},q_{2}=1}^{n}\varphi _{u}\left(
q_{1},q_{2}\right) \sum_{p_{1},p_{2}=1}^{n}\varphi _{x}\left(
p_{1},p_{2}\right) \\
&=&o\left( T^{-1}\right) \text{,}
\end{eqnarray*}%
by Lemma \ref{z_12} and $\left( \ref{the_5}\right) $. Likewise the third and
fourth terms of $\left( \ref{a_5}\right) $ are $o_{p}\left( T^{-1/2}\right) $%
. So to complete the proof we need to examine the first term of $\left( \ref%
{a_5}\right) $, whose second moment is, by $\left( \ref{the_5}\right) $ and
using $\left( \sup_{p,q}\left\vert f_{x,pq}\left( j\right) \right\vert
+\sup_{p,q}\left\vert f_{u,pq}\left( j\right) \right\vert \right) \leq C$,
bounded by%
\begin{equation*}
\frac{1}{Tn^{3}}\sum_{j=1}^{T-1}\sup_{p,q}\left\vert f_{x,pq}\left( j\right)
\right\vert \left\vert f_{u,pq}\left( j\right) \right\vert
\sum_{p_{1},p_{2}=1}^{n}\varphi _{x}\left( p_{1},p_{2}\right)
\sum_{q_{1},q_{2}=1}^{n}\varphi _{u}\left( q_{1},q_{2}\right) =o\left(
1\right) \text{.}
\end{equation*}%
This concludes the proof of the theorem. 
\hfill%
$\square $\newpage

\subsection{\textbf{PROOF OF PROPOSITION \protect\ref{PropV1}}}

$\left. {}\right. $

We begin with part $\left( \mathbf{a}\right) $. We need to show that, for
any $k_{1},k_{2}=1,...,k$, 
\begin{eqnarray*}
\breve{\Phi}_{k_{1},k_{2}} &=&\frac{1}{T}\sum_{j=1}^{T-1}\left\{ \left( 
\frac{1}{n^{1/2}}\sum_{p=1}^{n}\mathcal{J}_{\widetilde{x},p,k_{1}}\left(
j\right) \mathcal{J}_{\widehat{u},p}\left( -j\right) \right) \left( \frac{1}{%
n^{1/2}}\sum_{p=1}^{n}\mathcal{J}_{\widetilde{x},p,k_{2}}\left( -j\right) 
\mathcal{J}_{\widehat{u},p}\left( j\right) \right) \right\} \\
&&\overset{P}{\rightarrow }\Phi _{k_{1},k_{2}}\text{.}
\end{eqnarray*}%
To simplify the notation we shall assume that $k=1$. Now, after observing
that 
\begin{equation*}
\mathcal{J}_{\widehat{u},p}\left( j\right) =\mathcal{J}_{\widetilde{u}%
,p}\left( j\right) -\left( \widetilde{\beta }-\beta \right) \mathcal{J}_{%
\widetilde{x},p}\left( j\right) \text{,}
\end{equation*}%
we have that $\breve{\Phi}=:\breve{\Phi}_{1,1}$ is 
\begin{eqnarray}
&&\frac{1}{T}\sum_{j=1}^{T-1}\left\{ \left( \frac{1}{n^{1/2}}\sum_{p=1}^{n}%
\mathcal{J}_{\widetilde{x},p}\left( j\right) \mathcal{J}_{u,p}\left(
-j\right) \right) \left( \frac{1}{n^{1/2}}\sum_{p=1}^{n}\mathcal{J}_{%
\widetilde{x},p}\left( -j\right) \mathcal{J}_{u,p}\left( j\right) \right)
\right\}  \notag \\
&&+2\left( \widetilde{\beta }-\beta \right) \frac{1}{T}\sum_{j=1}^{T-1}\left%
\{ \left( \frac{1}{n^{1/2}}\sum_{p=1}^{n}\mathcal{I}_{\widetilde{x},p}\left(
j\right) \right) \left( \frac{1}{n^{1/2}}\sum_{p=1}^{n}\mathcal{J}_{%
\widetilde{x},p}\left( -j\right) \mathcal{J}_{u,p}\left( j\right) \right)
\right\}  \notag \\
&&+\left( \widetilde{\beta }-\beta \right) ^{2}\frac{1}{T}%
\sum_{j=1}^{T-1}\left( \frac{1}{n^{1/2}}\sum_{p=1}^{n}\mathcal{I}_{%
\widetilde{x},p}\left( j\right) \right) ^{2}\text{.}  \label{Prop_5}
\end{eqnarray}

The third term of $\left( \ref{Prop_5}\right) $ is $O_{p}\left(
T^{-1}\right) $ by Lemma \ref{Lem21} and $\widetilde{\beta }-\beta
=O_{p}\left( \left( nT\right) ^{-1/2}\right) $. The second term of $\left( %
\ref{Prop_5}\right) $ is also $o_{p}\left( 1\right) $ by Cauchy-Schwarz's
inequality if we show that the first term converges in probability to $\Phi $%
. Since%
\begin{equation}
\mathcal{J}_{\widetilde{x},p}\left( j\right) =\mathcal{J}_{x,p}\left(
j\right) -\mathcal{J}_{\overline{x},\cdot }\left( j\right) \text{,}
\label{dif_dft}
\end{equation}%
this result holds true if we show that 
\begin{equation}
\frac{1}{T}\sum_{j=1}^{T-1}\left\{ \left( \frac{1}{n^{1/2}}\sum_{p=1}^{n}%
\mathcal{J}_{x,p}\left( j\right) \mathcal{J}_{u,p}\left( -j\right) \right)
\left( \frac{1}{n^{1/2}}\sum_{p=1}^{n}\mathcal{J}_{x,p}\left( -j\right) 
\mathcal{J}_{u,p}\left( j\right) \right) \right\} \overset{P}{\rightarrow }%
\Phi  \label{prop_5}
\end{equation}%
and%
\begin{eqnarray}
&&\frac{1}{T}\sum_{j=1}^{T-1}\left\{ \left( \frac{1}{n^{1/2}}\sum_{p=1}^{n}%
\mathcal{J}_{\overline{x},\cdot }\left( j\right) \mathcal{J}_{u,p}\left(
-j\right) \right) \left( \frac{1}{n^{1/2}}\sum_{p=1}^{n}\mathcal{J}%
_{x,p}\left( -j\right) \mathcal{J}_{u,p}\left( j\right) \right) \right\} 
\notag \\
&&+\frac{1}{T}\sum_{j=1}^{T-1}\left\{ \left( \frac{1}{n^{1/2}}\sum_{p=1}^{n}%
\mathcal{J}_{\overline{x},\cdot }\left( j\right) \mathcal{J}_{u,p}\left(
-j\right) \right) \left( \frac{1}{n^{1/2}}\sum_{p=1}^{n}\mathcal{J}_{%
\overline{x},\cdot }\left( -j\right) \mathcal{J}_{u,p}\left( j\right)
\right) \right\}  \notag \\
&=&o_{p}\left( 1\right) \text{.}  \label{prop_55}
\end{eqnarray}

First we examine $\left( \ref{prop_55}\right) $. We begin with the first
term on the left of $\left( \ref{prop_55}\right) $, whose first moment is%
\begin{eqnarray*}
&&\frac{1}{T}\sum_{j=1}^{T-1}\sum_{p=1}^{n}E\left( \mathcal{J}_{\overline{x}%
,\cdot }\left( j\right) \mathcal{J}_{x,p}\left( -j\right) \right) E\left( 
\mathcal{J}_{\overline{u},p}\left( -j\right) \mathcal{J}_{u,p}\left(
j\right) \right) \\
&=&\frac{C}{Tn^{2}}\sum_{j=1}^{T-1}\sum_{p=1}^{n}\sum_{r=1}^{n}\varphi
_{x}\left( p,r\right) \sum_{q=1}^{n}\varphi _{u}\left( p,q\right) \left\{ 1+%
\frac{C}{T}\right\} \text{.}
\end{eqnarray*}%
using Lemma \ref{z_12}, after we observe that the factor in brackets is $%
n^{1/2}\mathcal{J}_{\overline{x},\cdot }\left( j\right) \mathcal{J}_{%
\overline{u},\cdot }\left( -j\right) $. Using $\left( \ref{pp}\right) ,$ we
conclude that the last displayed expression is $o\left( 1\right) $. Next, we
observe that Lemma \ref{B_5} implies, for instance, that 
\begin{eqnarray*}
&&E\left( \mathcal{J}_{\overline{u},\cdot }\left( -j\right) \mathcal{J}%
_{u,p}\left( j\right) \mathcal{J}_{\overline{u},\cdot }\left( -k\right) 
\mathcal{J}_{u,q}\left( k\right) \right) -E^{2}\left( \mathcal{J}_{\overline{%
u},\cdot }\left( -j\right) \mathcal{J}_{u,p}\left( j\right) \right) \\
&=&\varphi _{u}\left( p,q\right) \frac{1}{n^{2}}\sum_{p_{1},q_{1}=1}^{n}%
\varphi _{u}\left( p_{1},q_{1}\right) \left\{ \mathbf{1}\left( j=k\right) +%
\frac{C}{T}\right\} \text{.}
\end{eqnarray*}%
The variance of the first term on the left of $\left( \ref{prop_55}\right) ,$
therefore, is bounded by 
\begin{equation*}
\frac{1}{T^{2}}\sum_{j,k=1}^{T-1}\sum_{p,q=1}^{n}\varphi \left( p,q\right) 
\frac{1}{n^{4}}\sum_{p_{1},q_{1}=1}^{n}\varphi _{u}\left( p_{1},q_{1}\right)
\sum_{p_{2},q_{2}=1}^{n}\varphi _{x}\left( p_{2},q_{2}\right) \left\{ 
\mathbf{1}\left( j=k\right) +\frac{C}{T}\right\} =o\left( \frac{1}{T}\right)
\end{equation*}%
using Condition $C3$ and $\left( \ref{the_5}\right) $. Hence the first term
on the left of $\left( \ref{prop_55}\right) $ is $o_{p}\left( 1\right) $.
The same conclusion holds true for the second term of $\left( \ref{prop_55}%
\right) $.

To complete the proof of part $\left( \mathbf{a}\right) $, it remains to
show $\left( \ref{prop_5}\right) $. Using $\left( \ref{bartlett}\right) $,
we have that $\left( \ref{prop_5}\right) $ holds true if the following
expressions $\left( \ref{prop_51}\right) -\left( \ref{prop_53}\right) $ are $%
o_{p}\left( 1\right) $;%
\begin{eqnarray}
&&\frac{1}{nT}\sum_{j=1}^{T-1}\left\{ \left( \sum_{p=1}^{n}\mathcal{B}%
_{x,p}\left( -j\right) \mathcal{B}_{u,p}\left( j\right) \mathcal{J}_{\chi
,p}\left( j\right) \mathcal{J}_{\xi ,p}\left( -j\right) \right) \right. 
\notag \\
&&\text{ \ \ \ \ \ \ \ \ \ }\left. \left( \sum_{p=1}^{n}\mathcal{B}%
_{x,p}\left( -j\right) \mathcal{B}_{u,p}\left( j\right) \mathcal{J}_{\chi
,p}\left( j\right) \mathcal{J}_{\xi ,p}\left( -j\right) \right) \right\}
-\Phi \text{,}  \label{prop_51}
\end{eqnarray}%
\begin{equation}
\frac{1}{nT}\sum_{j=1}^{T-1}\left( \sum_{p=1}^{n}\mathcal{B}_{x,p}\left(
-j\right) \mathcal{J}_{\chi ,p}\left( j\right) \mathrm{Y}_{u,p}\left(
-j\right) \right) \left( \sum_{p=1}^{n}\mathcal{B}_{u,p}\left( j\right) 
\mathcal{J}_{\xi ,p}\left( -j\right) \mathrm{Y}_{x,p}\left( j\right) \right) 
\text{,}  \label{prop_52}
\end{equation}%
\begin{equation}
\frac{1}{nT}\sum_{j=1}^{T-1}\left( \sum_{p=1}^{n}\mathrm{Y}_{x,p}\left(
j\right) \mathrm{Y}_{u,p}\left( -j\right) \right) \left( \sum_{p=1}^{n}%
\mathrm{Y}_{u,p}\left( -j\right) \mathrm{Y}_{x,p}\left( j\right) \right)
\label{prop_53}
\end{equation}%
We begin by showing that $\left( \ref{prop_51}\right) $ is $o_{p}\left(
1\right) $. First, the expectation of $\left( \ref{prop_51}\right) $ is%
\begin{equation*}
\frac{1}{n}\sum_{p,q=1}^{n}\varphi \left( p,q\right) \frac{1}{T}%
\sum_{j=1}^{T-1}\mathcal{B}_{x,p}\left( -j\right) \mathcal{B}_{x,q}\left(
j\right) \mathcal{B}_{u,p}\left( j\right) \mathcal{B}_{u,q}\left( -j\right)
-\Phi =O\left( T^{-1}\right)
\end{equation*}%
because, by continuous differentiability of $f_{x,pq}\left( -\lambda \right)
f_{u,pq}\left( \lambda \right) $, we have that 
\begin{equation*}
\frac{1}{T}\sum_{j=1}^{T-1}\mathcal{B}_{x,p}\left( -j\right) \mathcal{B}%
_{x,q}\left( j\right) \mathcal{B}_{u,p}\left( j\right) \mathcal{B}%
_{u,q}\left( -j\right) -\int_{0}^{2\pi }f_{x,pq}\left( -\lambda \right)
f_{u,pq}\left( \lambda \right) d\lambda =O\left( T^{-1}\right) \text{.}
\end{equation*}%
Next, because $\left( \ref{dft_cov}\right) $ implies that%
\begin{eqnarray*}
&&E\left\{ \left( \mathcal{J}_{\chi ,p_{1}}\left( j\right) \mathcal{J}_{\xi
,p_{1}}\left( -j\right) \mathcal{J}_{\chi ,q_{1}}\left( -j\right) \mathcal{J}%
_{\xi ,q_{1}}\left( j\right) -E\left( \cdot \right) \right) \right. \\
&&\text{ \ \ \ \ \ }\left. \left( \mathcal{J}_{\chi ,p_{2}}\left( -k\right) 
\mathcal{J}_{\xi ,p_{2}}\left( k\right) \mathcal{J}_{\chi ,q_{2}}\left(
k\right) \mathcal{J}_{\xi ,q_{2}}\left( -k\right) -E\left( \cdot \right)
\right) \right\} \\
&=&\varphi _{x}\left( p_{1},p_{2}\right) \varphi _{x}\left(
q_{1},q_{2}\right) \varphi _{u}\left( q_{1},p_{2}\right) \varphi _{u}\left(
p_{1},q_{2}\right) \mathbf{1}\left( j=k\right) \\
&&+\varphi _{x}\left( p_{1},p_{2}\right) \varphi _{x}\left(
q_{1},q_{2}\right) \varphi _{u}\left( p_{1},p_{2}\right) \varphi _{u}\left(
q_{1},q_{2}\right) \mathbf{1}\left( j=k\right) \\
&&+2\varphi _{x}\left( p_{1},p_{2}\right) \varphi _{x}\left(
q_{1},q_{2}\right) \sum_{\ell =1}^{\infty }c_{\ell }\left( p_{1}\right)
c_{\ell }\left( p_{2}\right) c_{\ell }\left( q_{1}\right) c_{\ell }\left(
q_{2}\right) \mathbf{1}\left( j=k\right) \\
&&+\sum_{\ell =1}^{\infty }c_{\ell }\left( p_{1}\right) c_{\ell }\left(
p_{2}\right) c_{\ell }\left( q_{1}\right) c_{\ell }\left( q_{2}\right)
\sum_{\ell =1}^{\infty }d_{\ell }\left( p_{1}\right) d_{\ell }\left(
p_{2}\right) d_{\ell }\left( q_{1}\right) d_{\ell }\left( q_{2}\right)
\left( \mathbf{1}\left( j=k\right) +\frac{\kappa _{4,\xi }\kappa _{4,\chi }}{%
T}\right) \text{,}
\end{eqnarray*}%
standard algebra yields that the second moment of $\left( \ref{prop_51}%
\right) $ is $o\left( 1\right) $, when recognizing%
\begin{eqnarray}
\sum_{\ell =1}^{\infty }d_{\ell }\left( p_{1}\right) d_{\ell }\left(
p_{2}\right) d_{\ell }\left( q_{1}\right) d_{\ell }\left( q_{2}\right) &\leq
&\sum_{\ell =1}^{\infty }d_{\ell }\left( p_{1}\right) d_{\ell }\left(
p_{2}\right) \sum_{\ell =1}^{\infty }d_{\ell }\left( q_{1}\right) d_{\ell
}\left( q_{2}\right)  \notag \\
&=&\varphi _{u}\left( p_{1},p_{2}\right) \varphi _{u}\left(
q_{1},q_{2}\right)  \label{d1}
\end{eqnarray}%
\begin{eqnarray}
\sum_{\ell =1}^{\infty }c_{\ell }\left( p_{1}\right) c_{\ell }\left(
p_{2}\right) c_{\ell }\left( q_{1}\right) c_{\ell }\left( q_{2}\right) &\leq
&\sum_{\ell =1}^{\infty }c_{\ell }\left( p_{1}\right) c_{\ell }\left(
p_{2}\right) \sum_{\ell =1}^{\infty }c_{\ell }\left( q_{1}\right) c_{\ell
}\left( q_{2}\right)  \notag \\
&=&\varphi _{x}\left( p_{1},p_{2}\right) \varphi _{x}\left(
q_{1},q_{2}\right)  \label{c1}
\end{eqnarray}%
and 
\begin{eqnarray}
\sum_{p_{1}=1}^{n}\varphi _{x}\left( p_{1},p_{2}\right) \varphi _{u}\left(
p_{1},q_{2}\right) &\leq &\left( \sum_{p_{1}=1}^{n}\varphi _{x}^{1/\alpha
}\left( p_{1},p_{2}\right) \right) ^{\alpha }\left(
\sum_{p_{1}=1}^{n}\varphi _{u}^{1/1-\alpha }\left( p_{1},q_{2}\right)
\right) ^{1-\alpha }  \notag \\
&=&O\left( 1\right)  \label{A}
\end{eqnarray}%
since $\sum_{p_{1}=1}^{n}\varphi _{x}\left( p_{1},p_{2}\right) \varphi
_{u}\left( p_{1},p_{2}\right) =O\left( 1\right) $ implies $\varphi
_{x}\left( p_{1},p_{2}\right) =O\left( p_{1}^{-\alpha }\right) $ and $%
\varphi _{u}\left( p_{1},p_{2}\right) =O\left( p_{1}^{-\beta }\right) $ with 
$\alpha +\beta >1$.

Next consider $\left( \ref{prop_52}\right) $. Because $\sup_{p}\left\vert 
\mathcal{B}_{x,p}\left( -j\right) \mathcal{B}_{u,p}\left( j\right)
\right\vert <C$, the second moment of $\left( \ref{prop_52}\right) $ is
bounded by%
\begin{eqnarray*}
&&\frac{1}{\left( nT\right) ^{2}}\sum_{j,k=1}^{T-1}%
\sum_{p_{1},q_{1},p_{2},q_{2}=1}^{n}\left\vert E\left\{ \mathcal{J}_{\chi
,p_{1}}\left( j\right) \mathcal{J}_{\chi ,q_{1}}\left( -k\right) \mathrm{Y}%
_{x,p_{2}}\left( j\right) \mathrm{Y}_{x,q_{2}}\left( -k\right) \right\}
\right. \\
&&\left. E\left\{ \mathrm{Y}_{u,p_{1}}\left( -j\right) \mathrm{Y}%
_{u,q_{1}}\left( k\right) \mathcal{J}_{\xi ,p_{2}}\left( -j\right) \mathcal{J%
}_{\xi ,q_{2}}\left( k\right) \right\} \right\vert \text{.}
\end{eqnarray*}%
From here, proceeding as with $\left( \ref{prop_51}\right) $ but using
Lemmas \ref{z_12} and \ref{z_11} as needed, we easily conclude that $\left( %
\ref{prop_52}\right) =o_{p}\left( 1\right) $ by Markov's inequality, since
for instance 
\begin{eqnarray*}
&&E\left\{ \mathcal{J}_{\chi ,p_{1}}\left( j\right) \mathcal{J}_{\chi
,q_{1}}\left( k\right) \mathrm{Y}_{x,p_{2}}\left( -j\right) \mathrm{Y}%
_{x,q_{2}}\left( -k\right) \right\} \\
&=&E\left( \mathcal{J}_{\chi ,p_{1}}\left( j\right) \mathcal{J}_{\chi
,q_{1}}\left( k\right) \right) E\left( \mathrm{Y}_{x,p_{2}}\left( -j\right) 
\mathrm{Y}_{x,q_{2}}\left( -k\right) \right) \\
&&+E\left( \mathcal{J}_{\chi ,p_{1}}\left( j\right) \mathrm{Y}%
_{x,p_{2}}\left( -j\right) \right) E\left( \mathcal{J}_{\chi ,q_{1}}\left(
k\right) \mathrm{Y}_{x,q_{2}}\left( -k\right) \right) \\
&&+E\left( \mathcal{J}_{\chi ,p_{1}}\left( j\right) \mathrm{Y}%
_{x,q_{2}}\left( -k\right) \right) E\left( \mathcal{J}_{\chi ,q_{1}}\left(
k\right) \mathrm{Y}_{x,p_{2}}\left( -j\right) \right) \\
&&+cum\left( \mathcal{J}_{\chi ,p_{1}}\left( j\right) ;\mathcal{J}_{\chi
,q_{1}}\left( k\right) ;\mathrm{Y}_{x,p_{2}}\left( -j\right) ;\mathrm{Y}%
_{x,q_{2}}\left( -k\right) \right) \text{.}
\end{eqnarray*}%
The proof of part $\left( \mathbf{a}\right) $ now concludes since $\left( %
\ref{prop_53}\right) =o_{p}\left( 1\right) $ by standard algebra and Lemmas %
\ref{z_12} and \ref{z_11}.

Part $\left( \mathbf{b}\right) $. Because the continuous differentiability
of $f_{x,p}\left( \lambda \right) $, we have that $T^{-1}%
\sum_{j=1}^{T}f_{x,p}\left( j\right) \rightarrow \int_{0}^{2\pi
}f_{x,p}\left( \lambda \right) d\lambda =:\Sigma _{x,p}$, see Brillinger $%
\left( 1981,\text{p. 15}\right) $, so we can conclude by Lemma \ref{z_13}
and $\left( \ref{dif_dft}\right) $, that to finish the proof, it suffices to
show that 
\begin{equation*}
\frac{1}{nT}\sum_{p=1}^{n}\sum_{j=1}^{T-1}\mathcal{I}_{\overline{x},\cdot
}\left( j\right) \text{ and }\frac{2}{nT}\sum_{p=1}^{n}\sum_{j=1}^{T-1}%
\mathcal{J}_{\overline{x},\cdot }\left( j\right) \mathcal{J}_{x,p}\left(
j\right) =o_{p}\left( 1\right) \text{.}
\end{equation*}%
are both $o_{p}\left( 1\right) $. However this is the case proceeding
similarly as with the proof of $\left( \ref{prop_55}\right) $, so it is
omitted.%
\hfill%
$\square $\bigskip \bigskip

\subsection{\textbf{PROOF OF THEOREM \protect\ref{ThmBoot}}}

$\left. {}\right. $

Because Lemma \ref{Lem21} implies that $\left( nT\right)
^{-1}\sum_{p=1}^{n}\sum_{j=1}^{T-1}\mathcal{I}_{\widetilde{x},p}\left(
j\right) \overset{P}{\rightarrow }\Sigma _{x}$ and abbreviating $\widehat{f}%
_{u}\left( j\right) =\frac{1}{n}\sum_{q=1}^{n}\mathcal{I}_{\widehat{u}%
,q}\left( j\right) $, it suffices to show%
\begin{eqnarray}
&&\left( \mathbf{i}\right) \text{ \ }\frac{1}{T^{1/2}n^{1/2}}%
\sum_{p=1}^{n}\sum_{j=1}^{T-1}\mathcal{J}_{\widetilde{x},p}\left( j\right)
\left( \widehat{f}_{u}^{1/2}\left( j\right) -f_{u}^{1/2}\left( j\right)
\right) \mathcal{J}_{u^{\ast },p}\left( -j\right) =o_{p^{\ast }}\left(
1\right)  \label{The2_1} \\
&&\left( \mathbf{ii}\right) \text{ \ }\frac{1}{T^{1/2}n^{1/2}}%
\sum_{p=1}^{n}\sum_{j=1}^{T-1}\mathcal{J}_{\widetilde{x},p}\left( \lambda
_{j}\right) f_{u}^{1/2}\left( j\right) \mathcal{J}_{u^{\ast },p}\left(
-j\right) \overset{d^{\ast }}{\rightarrow }\mathcal{N}\left( 0,\Phi \right)
\label{The2_2}
\end{eqnarray}

We begin with part $\left( \mathbf{ii}\right) $. The left hand side of $%
\left( \ref{The2_2}\right) $ is%
\begin{eqnarray}
&&\frac{1}{T^{1/2}n^{1/2}}\sum_{p=1}^{n}\sum_{j=1}^{T-1}f_{u}^{1/2}\left(
j\right) \mathcal{B}_{x,p}\left( j\right) \mathcal{J}_{\chi ,p}\left(
j\right) \mathcal{J}_{u^{\ast },p}\left( -j\right)  \label{The2_4} \\
&&+\frac{1}{T^{1/2}n^{1/2}}\sum_{p=1}^{n}\sum_{j=1}^{T-1}f_{u}^{1/2}\left(
j\right) \left( \mathcal{J}_{\widetilde{x},p}\left( j\right) -\mathcal{B}%
_{x,p}\left( j\right) \mathcal{J}_{\chi ,p}\left( j\right) \right) \mathcal{J%
}_{u^{\ast },p}\left( -j\right) \text{.}  \notag
\end{eqnarray}

The second (bootstrap) moment of the second term of $\left( \ref{The2_4}%
\right) $ is%
\begin{equation}
\frac{1}{nT}\sum_{p,q=1}^{n}\sum_{j=1}^{T-1}f_{u}\left( j\right) \widehat{%
\sigma }_{u,pq}\left( \mathcal{J}_{\widetilde{x},p}\left( j\right) -\mathcal{%
B}_{x,p}\left( j\right) \mathcal{J}_{\chi ,p}\left( j\right) \right) \left( 
\mathcal{J}_{\widetilde{x},q}\left( -j\right) -\mathcal{B}_{x,q}\left(
-j\right) \mathcal{J}_{\chi ,q}\left( -j\right) \right)  \label{The2_5}
\end{equation}%
using 
\begin{equation}
E^{\ast }\left( \mathcal{J}_{u^{\ast },p}\left( j\right) \mathcal{J}%
_{u^{\ast },q}\left( -k\right) \right) =\widehat{\sigma }_{u,pq}\mathbf{1}%
\left( j=k\right) ;~\ \ \ \widehat{\sigma }_{u,pq}=\frac{1}{T}\sum_{t=1}^{T}%
\widehat{u}_{pt}\widehat{u}_{qt}\text{,}  \label{e}
\end{equation}%
By Lemma \ref{z_12} and $\left( \ref{bartlett}\right) $, 
\begin{eqnarray*}
&&E\left( \mathcal{J}_{\widetilde{x},p}\left( j\right) -\mathcal{B}%
_{x,p}\left( j\right) \mathcal{J}_{\chi ,p}\left( j\right) \right) \left( 
\mathcal{J}_{\widetilde{x},q}\left( -j\right) -\mathcal{B}_{x,p}\left(
-j\right) \mathcal{J}_{\chi ,p}\left( -j\right) \right) =\frac{C}{T}\varphi
_{x}\left( p,q\right) ; \\
&&\widehat{\sigma }_{u,pq}=\varphi _{u}\left( p,q\right) \left(
1+O_{p}\left( T^{-1/2}\right) \right) .
\end{eqnarray*}%
Hence it easily follows that the expected value of equation $\left( \ref%
{The2_5}\right) $ is $o\left( 1\right) $ and consequently the second term of 
$\left( \ref{The2_4}\right) $ is $o_{p^{\ast }}\left( 1\right) $, after we
observe that $\left( \ref{The2_5}\right) $ is a nonnegative expression.

Turning to the first term of $\left( \ref{The2_4}\right) ,$ let us denote 
\begin{equation}
\Xi _{s,t}^{\ast }\left( n\right) =\frac{1}{n^{1/2}}\sum_{p=1}^{n}\chi
_{ps}u_{pt}^{\ast };\text{ \ \ \ \ }\mathcal{G}\left( j\right) =:\mathcal{B}%
_{x,p}\left( j\right) f_{u,p}^{1/2}\left( j\right) .  \label{Not_1}
\end{equation}%
Standard algebra yields that the first term of $\left( \ref{The2_4}\right) $
is 
\begin{equation}
\frac{1}{\widetilde{T}^{1/2}}\frac{1}{T}\sum_{t,s=1}^{T}\Xi _{s,t}^{\ast
}\left( n\right) \sum_{j=1}^{\widetilde{T}}\mathcal{G}\left( j\right)
e^{i\left( t-s\right) \lambda _{j}}=\frac{1}{T^{1/2}}\sum_{t,s=1}^{T}\phi
\left( \left\vert t-s\right\vert \right) \Xi _{s,t}^{\ast }\left( n\right) +%
\frac{C}{T^{3/2}}\sum_{t,s=1}^{T}\Xi _{s,t}^{\ast }\left( n\right) \text{,}
\label{Lem3_3*}
\end{equation}%
where to simplify the notation we assume that $\varphi _{x}\left( p,p\right)
=\varphi _{u}\left( p,p\right) =1$ for all $p=1,...,n$ and $\phi \left(
r\right) $ is the $r$th Fourier coefficient of $\mathcal{G}\left( j\right) $%
. Hence the right hand side of $\left( \ref{Lem3_3*}\right) $ can now be
written as%
\begin{equation}
\frac{\phi \left( 0\right) }{T^{1/2}}\sum_{t=1}^{T-\ell }\frac{1}{n^{1/2}}%
\sum_{p=1}^{n}\chi _{pt}u_{pt}^{\ast }+\sum_{\ell =1}^{T-1}\frac{\phi \left(
\ell \right) }{T^{1/2}}\sum_{t=1}^{T-\ell }\frac{1}{n^{1/2}}\left\{
\sum_{p=1}^{n}\chi _{pt}u_{p,t+\ell }^{\ast }+\sum_{p=1}^{n}\chi _{p,t+\ell
}u_{pt}^{\ast }\right\} \text{.}  \label{The2_6}
\end{equation}%
Because $\phi \left( r\right) =O\left( r^{-2}\right) $ by Conditions $C1$
and $C2,$ given the independence of the sequences of random variables $%
n^{-1/2}\sum_{p=1}^{n}\chi _{pt}u_{p,t+\ell }^{\ast }$ and \ $%
n^{-1/2}\sum_{p=1}^{n}\chi _{p,t+\ell }u_{pt}^{\ast }$ in $t$, to complete
the proof of part $\left( \mathbf{ii}\right) $, it suffices to to show that%
\begin{equation*}
\Lambda _{t,n}^{\ast }=:\frac{1}{n^{1/2}}\sum_{p=1}^{n}\chi _{pt}u_{p,t+\ell
}^{\ast }\overset{d^{\ast }}{\rightarrow }\mathcal{N}\left( 0,\frac{T-\ell }{%
T}\lim_{n\rightarrow \infty }\frac{1}{n}\sum_{p,q=1}^{n}\varphi \left(
p,q\right) \right) \text{.}
\end{equation*}

The second bootstrap moment of $\Lambda _{t,n}^{\ast }$ is%
\begin{equation*}
\frac{1}{n}\sum_{p,q=1}^{n}\chi _{pt}\chi _{qt}\frac{1}{T}\sum_{r=1}^{T-\ell
}\widehat{u}_{p,r+\ell }\widehat{u}_{q,r+\ell }=\frac{1}{n}%
\sum_{p,q=1}^{n}\chi _{pt}\chi _{qt}\frac{1}{T}\sum_{r=1}^{T-\ell
}u_{p,r+\ell }u_{q,r+\ell }\left( 1+o_{p}\left( 1\right) \right) \text{,}
\end{equation*}%
by standard algebra and Theorem \ref{ThmEst}. Now, Conditions $C1$ and $C2$
imply that%
\begin{equation*}
\frac{1}{n}\sum_{p,q=1}^{n}\left\{ E\left( \chi _{pt}\chi _{qt}\right) \frac{%
1}{T}\sum_{r=1}^{T-\ell }E\left( u_{p,r+\ell }u_{q,r+\ell }\right) \right\} =%
\frac{T-\ell }{T}\frac{1}{n}\sum_{p,q=1}^{n}\varphi \left( p,q\right) \text{.%
}
\end{equation*}%
Moreover, because $E\left( u_{p_{1},t+\ell }u_{q_{1},t+\ell }u_{p_{2},s+\ell
}u_{q_{2},s+\ell }\right) =E\left(
u_{p_{1}t}u_{q_{1}t}u_{p_{2}s}u_{q_{2}s}\right) $%
\begin{eqnarray*}
&&E\left( \frac{1}{n}\sum_{p,q=1}^{n}\chi _{pt}\chi _{qt}\frac{1}{T}%
\sum_{t=1}^{T-\ell }u_{pt}u_{qt}\right) ^{2} \\
&=&\frac{1}{n^{2}}\sum_{p_{1},q_{1},p_{2},q_{2}=1}^{n}E\left( \chi
_{p_{1}t}\chi _{q_{1}t}\chi _{p_{2}t}\chi _{q_{2}t}\right) \frac{1}{T^{2}}%
\sum_{t,s=1}^{T-\ell }E\left( u_{p_{1}t}u_{q_{1}t}u_{p_{2}s}u_{q_{2}s}\right)
\\
&=&\frac{1}{n^{2}}\sum_{p_{1},q_{1},p_{2},q_{2}=1}^{n}\frac{1}{T^{2}}%
\sum_{t,s=1}^{T-\ell }\left\{ E\left( \chi _{p_{1}t}\chi _{q_{1}t}\right)
E\left( \chi _{p_{2}t}\chi _{q_{2}t}\right) +E\left( \chi _{p_{1}t}\chi
_{q_{2}t}\right) E\left( \chi _{p_{2}t}\chi _{q_{1}t}\right) \right. \\
&&~\ \ \ \ \ \ \ \ \ \ \ \ \ \ \ \ \ \ \ \ \ \ \ \ \ \ \ \ \left. +E\left(
\chi _{p_{1}t}\chi _{p_{2}t}\right) E\left( \chi _{q_{1}t}\chi
_{q_{2}t}\right) +cum\left( \chi _{p_{1}t};\chi _{q_{1}t;}\chi
_{p_{2}t;}\chi _{q_{2}t}\right) \right\} \\
&&~\ \ \ \ \ \ \ \ \ \ \ \ \ \ \ \ \ \ \ \ \ \ \ \ \ \ \ \ \times \left\{
E\left( u_{p_{1}t}u_{q_{1}t}\right) E\left( u_{p_{2}s}u_{q_{2}s}\right)
+E\left( u_{p_{1}t}u_{q_{2}s}\right) E\left( u_{p_{2}s}u_{q_{1}t}\right)
\right. \\
&&~\ \ \ \ \ \ \ \ \ \ \ \ \ \ \ \ \ \ \ \ \ \ \ \ \ \ \ \left. +E\left(
u_{p_{1}t}u_{p_{2}s}\right) E\left( u_{q_{1}t}u_{q_{2}s}\right) +cum\left(
u_{p_{1}t};u_{q_{1}t};u_{p_{2}s};u_{q_{2}s}\right) \right\} \\
&=&\frac{1}{n^{2}T^{2}}\sum_{p_{1},q_{1},p_{2},q_{2}=1}^{n}\sum_{t,s=1}^{T-%
\ell }E\left( \chi _{p_{1}t}\chi _{q_{1}t}\right) E\left( \chi _{p_{2}t}\chi
_{q_{2}t}\right) E\left( u_{p_{1}t}u_{q_{1}t}\right) E\left(
u_{p_{2}s}u_{q_{2}s}\right) \left( 1+o\left( 1\right) \right) \\
&=&\left( \frac{T-\ell }{T}\frac{1}{n}\sum_{p,q=1}^{n}\varphi \left(
p,q\right) \right) ^{2}\left( 1+o\left( 1\right) \right)
\end{eqnarray*}%
because $E\left( u_{ps}u_{qr}\right) =\varphi _{u}\left( p,q\right) \gamma
_{u,pq}\left( r-s\right) $, $\sum_{r,s=1}^{T}\left\vert \gamma _{u,pq}\left(
r-s\right) \right\vert =O\left( T\right) $ and $\left( \ref{A}\right) $.
This shows that the second moment converges to the square of the first
moment, and hence $E^{\ast }\left\vert \Lambda _{t,n}^{\ast }\right\vert
^{2}-\frac{T-\ell }{T}\frac{1}{n}\sum_{p,q=1}^{n}\varphi \left( p,q\right)
=o_{p}\left( 1\right) $.

Thus, it remains to show the Lindeberg's condition to complete the proof of
part $\left( \mathbf{ii}\right) $. To that end, it suffices to show that 
\begin{equation*}
\frac{1}{n^{2}}\sum_{p=1}^{n}E^{\ast }\left( \chi _{pt}u_{p,t+\ell }^{\ast
}\right) ^{4}=o_{p}\left( 1\right) \text{.}
\end{equation*}%
The left hand side of the last displayed expression is%
\begin{eqnarray*}
\frac{1}{n^{2}}\sum_{p=1}^{n}\left\Vert \chi _{pt}\right\Vert ^{4}\frac{1}{T}%
\sum_{t=1}^{T-\ell }\widehat{u}_{p,t+\ell }^{4} &=&\frac{1}{n^{2}}%
\sum_{p=1}^{n}\left\Vert \chi _{pt}\right\Vert ^{4}\frac{1}{T}%
\sum_{t=1}^{T-\ell }u_{p,t+\ell }^{4}\left( 1+o_{p}\left( 1\right) \right) \\
&=&O_{p}\left( n^{-1}\right) \text{,}
\end{eqnarray*}%
which completes the proof of part $\left( \mathbf{ii}\right) $.

Next we prove part $\left( \mathbf{i}\right) $. The left side of $\left( \ref%
{The2_1}\right) $ is%
\begin{eqnarray}
&&\frac{1}{T^{1/2}n^{1/2}}\sum_{p=1}^{n}\sum_{j=1}^{T-1}\left( \widehat{f}%
_{u}^{1/2}\left( j\right) -f_{u}^{1/2}\left( j\right) \right) \mathcal{B}%
_{x,p}\left( j\right) \mathcal{J}_{\chi ,p}\left( j\right) \mathcal{J}%
_{u^{\ast },p}\left( -j\right)  \label{a_27} \\
&&+\frac{1}{T^{1/2}n^{1/2}}\sum_{p=1}^{n}\sum_{j=1}^{T-1}\left( \widehat{f}%
_{u}^{1/2}\left( j\right) -f_{u}^{1/2}\left( j\right) \right) \left( 
\mathcal{J}_{\widetilde{x},p}\left( j\right) -\mathcal{B}_{x,p}\left(
j\right) \mathcal{J}_{\chi ,p}\left( j\right) \right) w_{u^{\ast },p}\left(
-j\right) \text{.}  \notag
\end{eqnarray}%
We shall only show explicitly that the first term of $\left( \ref{a_27}%
\right) $ is $o_{p^{\ast }}\left( 1\right) $, the second term following
similarly if not easier proceeding as with the second term of $\left( \ref%
{The2_4}\right) $ and Lemma \ref{z_12}. Now by $\left( \ref{e}\right) $, the
first term of $\left( \ref{a_27}\right) $ has second bootstrap moment given
by 
\begin{equation*}
\frac{1}{T}\sum_{t=1}^{T}\frac{1}{nT}\sum_{j=1}^{T-1}\left\{ \widehat{f}%
_{u}^{1/2}\left( j\right) -f_{u}^{1/2}\left( j\right) \right\}
^{2}f_{x}\left( j\right) \sum_{p,q=1}^{n}\widehat{u}_{pt}\widehat{u}_{qt}%
\mathcal{J}_{\chi ,p}\left( j\right) \mathcal{J}_{\chi ,q}\left( -j\right) 
\text{.}
\end{equation*}%
Because the last displayed expression is a nonnegative expression, to show
that it is $o_{p}\left( 1\right) $, it suffices to show that its first
moment converges to zero. To that end, we first observe that%
\begin{equation}
\left\{ \widehat{f}_{u}^{1/2}\left( j\right) -f_{u}^{1/2}\left( j\right)
\right\} ^{2}\leq \left\vert \frac{1}{n}\sum_{q=1}^{n}\mathcal{I}_{\widehat{u%
},q}\left( j\right) -f_{u}\left( j\right) \right\vert =o_{p}\left( 1\right)
\label{fe}
\end{equation}%
using standard arguments and Theorem 1 under Condition\textbf{\ }$C4$. On
the other hand, proceeding similarly as in Proposition \ref{PropV1}, we
obtain easily that 
\begin{equation*}
\frac{1}{n}\sum_{p,q=1}^{n}\widehat{u}_{pt}\mathcal{J}_{\chi ,p}\left(
j\right) \widehat{u}_{qt}\mathcal{J}_{\chi ,q}\left( -j\right) =\frac{1}{n}%
\sum_{p,q=1}^{n}u_{pt}\mathcal{J}_{\chi ,p}\left( j\right) u_{qt}\mathcal{J}%
_{\chi ,q}\left( -j\right) \left( 1+o_{p}\left( 1\right) \right) \text{,}
\end{equation*}%
and thus the proof of part $\left( \mathbf{i}\right) $, and thereby the
theorem, is completed if 
\begin{equation*}
E\left( \sum_{p,q=1}^{n}u_{pt}u_{qt}\mathcal{J}_{\chi ,p}\left( j\right) 
\mathcal{J}_{\chi ,q}\left( -j\right) \right) =O\left( n\right) \text{.}
\end{equation*}%
But the left hand side of the last displayed expression is%
\begin{equation*}
\sum_{p,q=1}^{n}\varphi _{u}\left( p,q\right) \frac{1}{T}\sum_{t,s=1}^{T}E%
\left( x_{pt}x_{qs}\right) e^{-i\left( t-s\right) \lambda
_{j}}=C\sum_{p,q=1}^{n}\varphi _{u}\left( p,q\right) \varphi _{x}\left(
p,q\right) =O\left( n\right)
\end{equation*}%
by Condition $C3,$ which completes the proof of the theorem. 
\hfill%
$\square $

\subsection{\textbf{PROOF OF PROPOSITION \protect\ref{PropV1Boot}}}

$\left. {}\right. $

As with the proof of Proposition \ref{PropV1}, we shall assume that $k=1$.
Now, after observing that 
\begin{equation*}
\mathcal{J}_{u^{\ast }p}\left( j\right) =\mathcal{J}_{\widetilde{u}^{\ast
},p}\left( j\right) -\left( \widetilde{\beta }^{\ast }-\widetilde{\beta }%
\right) \mathcal{J}_{\widetilde{x},p}\left( j\right) \text{,}
\end{equation*}%
we have that $\breve{\Phi}^{\ast }$\textbf{\ }equals the sum of the
following expressions $\left( \ref{prop2_1}\right) -\left( \ref{prop2_3}%
\right) $; 
\begin{eqnarray}
&&\frac{1}{T}\sum_{j=1}^{T-1}\widehat{f}_{u}\left( j\right) \left( \frac{1}{%
n^{1/2}}\sum_{p=1}^{n}\mathcal{J}_{\widetilde{x},p}\left( j\right) \mathcal{J%
}_{u^{\ast },p}\left( -j\right) \right) \left( \frac{1}{n^{1/2}}%
\sum_{p=1}^{n}\mathcal{J}_{\widetilde{x},p}\left( -j\right) \mathcal{J}%
_{u^{\ast },p}\left( j\right) \right) -\breve{\Phi}  \label{prop2_1} \\
&&2\left( \widetilde{\beta }^{\ast }-\widetilde{\beta }\right) \frac{1}{T}%
\sum_{j=1}^{T-1}\widehat{f}_{u}^{1/2}\left( j\right) \left( \frac{1}{n^{1/2}}%
\sum_{p=1}^{n}\mathcal{I}_{\widetilde{x},p}\left( j\right) \right) \left( 
\frac{1}{n^{1/2}}\sum_{p=1}^{n}\mathcal{J}_{\widetilde{x},p}\left( -j\right) 
\mathcal{J}_{u^{\ast },p}\left( j\right) \right)  \label{prop2_2} \\
&&\left( \widetilde{\beta }^{\ast }-\widetilde{\beta }\right) ^{2}\frac{1}{T}%
\sum_{j=1}^{T-1}\left( \frac{1}{n^{1/2}}\sum_{p=1}^{n}\mathcal{I}_{%
\widetilde{x},p}\left( j\right) \right) ^{2}\text{.}  \label{prop2_3}
\end{eqnarray}

That $\left( \ref{prop2_3}\right) $ is $o_{p^{\ast }}\left( 1\right) $
follows straightforwardly by Theorem \ref{ThmBoot} and Lemma \ref{Lem21} and 
$\left( \ref{prop2_2}\right) $ is $o_{p^{\ast }}\left( 1\right) $ by
Cauchy-Schwarz's inequality if we show that $\left( \ref{prop2_1}\right) $
is $o_{p^{\ast }}\left( 1\right) $. To that end, using $\left( \ref{dif_dft}%
\right) $ and $\left( \ref{e}\right) $, we have%
\begin{eqnarray*}
E^{\ast }\left( \ref{prop2_1}\right) &=&\frac{1}{nT}\sum_{j=1}^{T-1}\widehat{%
f}_{u}\left( j\right) \sum_{p,q=1}^{n}\mathcal{J}_{x,p}\left( j\right) 
\mathcal{J}_{x,q}\left( -j\right) \widehat{\sigma }_{u,pq}-\breve{\Phi} \\
&&+\frac{1}{nT}\sum_{j=1}^{T-1}\widehat{f}_{u}\left( j\right)
\sum_{p,q=1}^{n}\mathcal{J}_{\overline{x},\cdot }\left( j\right) \mathcal{J}%
_{\overline{x},\cdot }\left( -j\right) \widehat{\sigma }_{u,pq}.
\end{eqnarray*}%
Because $\widehat{\sigma }_{u,pq}=\varphi _{u}\left( p,q\right) \left(
1+o_{p}\left( 1\right) \right) $ and $\breve{\Phi}-\Phi =o_{p}\left(
1\right) $ by Proposition \ref{PropV1}, proceeding as in the proof of
Theorem \ref{ThmBoot} part $\left( \mathbf{i}\right) $, it suffices to
examine the behaviour of%
\begin{eqnarray}
&&\frac{1}{nT}\sum_{j=1}^{T-1}f_{u}\left( j\right) \sum_{p,q=1}^{n}\left\{
\varphi _{u}\left( p,q\right) \mathcal{J}_{x,p}\left( j\right) \mathcal{J}%
_{x,q}\left( -j\right) \right\} -\Phi  \label{Prop2_5} \\
&&+\frac{1}{T}\sum_{j=1}^{T-1}f_{u}\left( j\right) \mathcal{J}_{\overline{x}%
,\cdot }\left( j\right) \mathcal{J}_{\overline{x},\cdot }\left( -j\right) 
\frac{1}{n}\sum_{p,q=1}^{n}\varphi _{u}\left( p,q\right) \text{.}
\label{Prop2_4}
\end{eqnarray}%
$\left( \ref{Prop2_4}\right) $ is $o_{p}\left( 1\right) $ as we now show. As
it is a nonnegative sequence, it suffices to show that its first mean
converges to zero. Using $\left( \ref{bartlett}\right) $ and then Lemmas \ref%
{z_12} and \ref{z_11}, we have that its first moment is proportional to%
\begin{equation*}
\frac{1}{n^{2}}\sum_{p,q=1}^{n}\varphi _{x}\left( p,q\right) \frac{1}{n}%
\sum_{p,q=1}^{n}\varphi _{u}\left( p,q\right) =o\left( 1\right)
\end{equation*}%
by $\left( \ref{pp}\right) $. Because the first moment of $\left( \ref%
{Prop2_5}\right) $ is $o\left( 1\right) $, it then remains to show that the
(bootstrap) variance of $\left( \ref{prop2_1}\right) ,$ with $\mathcal{J}_{%
\widetilde{x},p}\left( j\right) $ replaced by $\mathcal{J}_{x,p}\left(
j\right) ,$ converges to zero. Using $\left( \ref{e}\right) $, the
(bootstrap) variance is 
\begin{eqnarray*}
&&\frac{1}{T^{2}}\sum_{j=1}^{T-1}\widehat{f}_{u}^{2}\left( j\right) \left( 
\frac{1}{n^{2}}\sum_{p_{1},q_{1}p_{2},q_{2}=1}^{n}\mathcal{J}%
_{x,p_{1}}\left( j\right) \mathcal{J}_{x,q_{1}}\left( -j\right) \mathcal{J}%
_{x,p_{2}}\left( -j\right) \mathcal{J}_{x,q_{2}}\left( j\right) \widehat{%
\sigma }_{u,p_{1}p_{2}}\widehat{\sigma }_{u,q_{1}q_{2}}\right) \\
&&+\frac{\kappa _{4,\xi }\left( 1+o_{p}\left( 1\right) \right) }{T^{3}n^{2}}%
\sum_{j,k=1}^{T-1}\left\{ \widehat{f}_{u}\left( j\right) \widehat{f}%
_{u}\left( k\right) \right. \\
&&~\ \ \ \ \ \ \ \ \ \ \ \ \left. \times
\sum_{p_{1},q_{1}p_{2},q_{2}=1}^{n}\varphi _{u}\left( p_{1},q_{1}\right)
\varphi _{u}\left( p_{2},q_{2}\right) \mathcal{J}_{x,p_{1}}\left( j\right) 
\mathcal{J}_{x,q_{1}}\left( -j\right) \mathcal{J}_{x,p_{2}}\left( -k\right) 
\mathcal{J}_{x,q_{2}}\left( k\right) \right\} ,
\end{eqnarray*}%
with Lemma \ref{cum_2} guaranteeing%
\begin{equation*}
cum^{\ast }\left( u_{p_{1}t}^{\ast },u_{q_{1}t}^{\ast },u_{p_{2}t}^{\ast
},u_{q_{2}t}^{\ast }\right) =\kappa _{4,\xi }\varphi _{u}\left(
p_{1},q_{1}\right) \varphi _{u}\left( p_{2},q_{2}\right) \left(
1+o_{p}\left( 1\right) \right) \text{.}
\end{equation*}%
From here we proceed as before after noticing that $\widehat{\sigma }%
_{u,p_{1}p_{2}}=\varphi _{u}\left( p_{1},p_{2}\right) \left( 1+o_{p}\left(
1\right) \right) $. This completes the proof of the proposition. 
\hfill%
$\square $

\subsection{\textbf{PROOF OF PROPOSITION \protect\ref{Pro_estboot}}}

$\left. {}\right. $

As with the proof of Theorem \ref{ThmBoot}, it suffices to show that 
\begin{equation}
\frac{1}{T^{1/2}n^{1/2}}\sum_{j=1}^{T-1}\sum_{p=1}^{n}\mathcal{J}_{%
\widetilde{x},p}\left( j\right) \mathcal{J}_{\widehat{u},p}\left( -j\right)
\eta _{j}\overset{d^{\ast }}{\rightarrow }\mathcal{N}\left( 0,\Phi \right) 
\text{.}  \label{a_26}
\end{equation}%
Because $\eta _{j}$ are normally distributed it suffices to show 
\begin{equation*}
E^{\ast }\left( \frac{1}{T^{1/2}n^{1/2}}\sum_{j=1}^{T-1}\sum_{p=1}^{n}%
\mathcal{J}_{\widetilde{x},p}\left( j\right) \mathcal{J}_{\widehat{u}%
,p}\left( -j\right) \eta _{j}\right) ^{2}\overset{P}{\rightarrow }\Phi \text{%
.}
\end{equation*}%
This is the case as we now show. The left hand side of the last displayed
expression is%
\begin{eqnarray*}
&&\frac{1}{nT}\sum_{j=1}^{T-1}\sum_{p,q=1}^{n}\mathcal{J}_{\widetilde{x}%
,p}\left( j\right) \mathcal{J}_{\widetilde{x},q}\left( -j\right) \mathcal{J}%
_{\widehat{u},p}\left( -j\right) \mathcal{J}_{\widehat{u},q}\left( j\right)
\\
&=&\frac{1}{nT}\sum_{j=1}^{T-1}\sum_{p,q=1}^{n}\mathcal{J}_{\widetilde{x}%
,p}\left( j\right) \mathcal{J}_{\widetilde{x},q}\left( -j\right) \mathcal{J}%
_{u,p}\left( -j\right) \mathcal{J}_{u,q}\left( j\right) +o_{p}\left( 1\right)
\end{eqnarray*}%
as $\widehat{u}_{pt}-u_{pt}=\left( \widetilde{\beta }-\beta \right) x_{pt}$
and $\widetilde{\beta }-\beta =O_{p}\left( T^{-1/2}n^{-1/2}\right) $. Using $%
\left( \ref{dif_dft}\right) $ and proceeding as in the proof of part $\left( 
\mathbf{a}\right) $ of Proposition \ref{PropV1}, we now have that the right
hand side is%
\begin{eqnarray*}
&&\frac{1}{nT}\sum_{j=1}^{T-1}\sum_{p,q=1}^{n}\mathcal{J}_{x,p}\left(
j\right) \mathcal{J}_{x,q}\left( -j\right) \mathcal{J}_{u,p}\left( -j\right) 
\mathcal{J}_{u,q}\left( -j\right) \\
&&+\frac{2}{nT}\sum_{j=1}^{T-1}\sum_{p,q=1}^{n}\mathcal{J}_{x,p}\left(
j\right) \mathcal{J}_{\overline{x},\cdot }\left( -j\right) \mathcal{J}%
_{u,p}\left( -j\right) \mathcal{J}_{u,q}\left( -j\right) \\
&&+\frac{1}{nT}\sum_{j=1}^{T-1}\sum_{p,q=1}^{n}\mathcal{J}_{\overline{x}%
,\cdot }\left( j\right) \mathcal{J}_{\overline{x},\cdot }\left( -j\right) 
\mathcal{J}_{u,p}\left( -j\right) \mathcal{J}_{u,q}\left( -j\right)
+o_{p}\left( 1\right) \text{.}
\end{eqnarray*}%
The first term converges in probability to $\Phi $, whereas the second term
follows by Cauchy-Schwarz's inequality if the third term is also $%
o_{p}\left( 1\right) $. But that term is $o_{p}\left( 1\right) $ proceeding
as in the proof of part $\left( \mathbf{a}\right) $ of Proposition \ref%
{PropV1} using Lemma \ref{B_5}. Again observe that the expression is
nonnegative. This concludes the proof. 
\hfill%
$\square \bigskip $

\begin{center}
\textbf{{\large {Appendix B: LEMMAS}}}
\end{center}

\renewcommand{\theequation}{B.\arabic{equation}} \setcounter{equation}{0}%
\renewcommand{\thelemma}{B.\arabic{lemma}} \setcounter{lemma}{0}%
\renewcommand{\theremark}{B.\arabic{remark}} \setcounter{remark}{0}

First denoting $\Upsilon _{\ell ,p}\left( j\right) =\left\{ \sum_{t=1-\ell
}^{T-\ell }-\sum_{t=1}^{T}\right\} \xi _{pt}e^{-it\lambda _{j}}$ and $\Psi
_{\ell ,p}\left( j\right) =\left\{ \sum_{t=1-\ell }^{T-\ell
}-\sum_{t=1}^{T}\right\} \chi _{pt}e^{-it\lambda _{j}}$, we have that $%
\mathrm{Y}_{u,p}\left( j\right) $ and $\mathrm{Y}_{x,p}\left( j\right) $
given in $\left( \ref{y_p}\right) $ can be decomposed as%
\begin{eqnarray}
\mathrm{Y}_{u,p}\left( j\right) &=&\mathrm{Y}_{u,p}^{\left( 1\right) }\left(
j\right) +\mathrm{Y}_{u,p}^{\left( 2\right) }\left( j\right)  \label{y_n12}
\\
\mathrm{Y}_{x,p}\left( j\right) &=&\mathrm{Y}_{x,p}^{\left( 1\right) }\left(
j\right) +\mathrm{Y}_{x,p}^{\left( 2\right) }\left( j\right) \text{,}  \notag
\end{eqnarray}%
where 
\begin{eqnarray*}
\mathrm{Y}_{u,p}^{\left( 1\right) }\left( j\right) &=&\frac{1}{T^{1/2}}%
\sum_{\ell =0}^{T}d_{\ell }\left( p\right) e^{-i\ell \lambda _{j}}\Upsilon
_{\ell ,p}\left( j\right) ;~\ \mathrm{Y}_{u,p}^{\left( 2\right) }\left(
j\right) =\frac{1}{T^{1/2}}\sum_{\ell =T+1}^{\infty }d_{\ell }\left(
p\right) e^{-i\ell \lambda _{j}}\Upsilon _{\ell ,p}\left( j\right) \\
\mathrm{Y}_{x,p}^{\left( 1\right) }\left( j\right) &=&\frac{1}{T^{1/2}}%
\sum_{\ell =0}^{T}c_{\ell }\left( p\right) e^{-i\ell \lambda _{j}}\Psi
_{\ell ,p}\left( j\right) ;~\ \mathrm{Y}_{x,p}^{\left( 2\right) }\left(
j\right) =\frac{1}{T^{1/2}}\sum_{\ell =T+1}^{\infty }c_{\ell }\left(
p\right) e^{-i\ell \lambda _{j}}\Psi _{\ell ,p}\left( j\right) \text{.}
\end{eqnarray*}

\begin{lemma}
\label{z_12}Assuming $C1$ and $C2$, we have that for $p,q=1,..,n$ and some $%
\upsilon _{u},\upsilon _{x}>0$ finite, 
\begin{eqnarray}
E\left( \mathrm{Y}_{w,p}^{\left( 1\right) }\left( j\right) \mathrm{Y}%
_{w,q}^{\left( 1\right) }\left( -k\right) \right) &=&\frac{\upsilon
_{w}\varphi _{w}\left( p,q\right) }{T};~\ \ \ w=:u\text{ or }x  \label{b_4}
\\
E\left( \mathrm{Y}_{w,p}^{\left( 2\right) }\left( j\right) \mathrm{Y}%
_{w,q}^{\left( 2\right) }\left( -k\right) \right) &=&o\left( T^{-2}\right)
\varphi _{w}\left( p,q\right) \mathbf{1}\left( j=k\right) ;~\ \ \ w=:u\text{
or }x\text{.}  \label{b_5}
\end{eqnarray}
\end{lemma}

\begin{proof}
We examine only the case when $w=:u$, with the proof for $w=:x$ similarly
handled. We begin with $\left( \ref{b_5}\right) $. Because for $\ell \geq T$%
, $E\left( \Upsilon _{\ell ,p}\left( j\right) \Upsilon _{\ell ,q}\left(
-k\right) \right) =2T\varphi _{u}\left( p,q\right) \mathbf{1}\left(
j=k\right) $, we obtain that the left hand side of $\left( \ref{b_5}\right) $
is 
\begin{equation*}
2\sum_{\ell _{1},\ell _{2}=T+1}^{\infty }d_{\ell _{1}}\left( p\right)
d_{\ell _{2}}\left( q\right) \varphi _{u}\left( p,q\right) \mathbf{1}\left(
j=k\right) \text{.}
\end{equation*}%
The conclusion then follows because Condition $C1$ implies that $\sum_{\ell
=T+1}^{\infty }\sup_{p}\left\vert d_{\ell }\left( p\right) \right\vert
=o\left( T^{-1}\right) $.

Next we consider $\left( \ref{b_4}\right) $. By definition, the left side is 
\begin{equation*}
\frac{1}{T}\sum_{\ell _{1},\ell _{2}=0}^{T}d_{\ell _{1}}\left( p\right)
d_{\ell _{2}}\left( q\right) E\left( \Upsilon _{\ell ,p}\left( j\right)
\Upsilon _{\ell ,q}\left( -k\right) \right) =\varphi _{u}\left( p,q\right) 
\frac{\upsilon _{u}}{T}
\end{equation*}%
since $\Upsilon _{\ell ,p}\left( j\right) =\left\{ \sum_{t=1-\ell
}^{0}-\sum_{t=T-\ell +1}^{T}\right\} \xi _{pt}e^{it\lambda _{j}}$ when $\ell
\leq T$, so that 
\begin{equation*}
E\left( \Upsilon _{\ell ,p}\left( j\right) \Upsilon _{\ell ,q}\left(
-k\right) \right) =2\varphi _{u}\left( p,q\right) \sum_{t=1}^{\ell
}e^{it\left( \lambda _{j}-\lambda _{k}\right) }\text{.}
\end{equation*}%
We now conclude because $\sum_{\ell =0}^{\infty }\ell \sup_{p}\left\vert
d_{\ell }\left( p\right) \right\vert <\infty $ by Condition $C1$.\hfill
\end{proof}

\begin{lemma}
\label{z_11}Assuming $C1$ and $C2$, we have that for $p,q=1,..,n$, 
\begin{eqnarray*}
\left( \mathbf{a}\right) \text{ \ }E\left( \mathrm{Y}_{u,p}^{\left( 1\right)
}\left( j\right) \mathcal{J}_{\xi ,q}\left( -k\right) \right) &=&\varphi
_{u}\left( p,q\right) \frac{1}{T}\sum_{\ell =0}^{T}d_{\ell }\left( p\right)
e^{-i\ell \lambda _{j}}\sum_{t=1}^{\ell }e^{it\lambda _{j-k}} \\
E\left( \mathrm{Y}_{u,p}^{\left( 2\right) }\left( j\right) \mathcal{J}_{\xi
,q}\left( -k\right) \right) &=&\varphi _{u}\left( p,q\right) \mathbf{1}%
\left( j=k\right) o\left( T^{-2}\right) \\
\left( \mathbf{b}\right) \text{ \ }E\left( \mathrm{Y}_{x,p}^{\left( 1\right)
}\left( j\right) \mathcal{J}_{\chi ,q}\left( -k\right) \right) &=&\varphi
_{x}\left( p,q\right) \frac{1}{T}\sum_{\ell =0}^{T}c_{\ell }\left( p\right)
e^{-i\ell \lambda _{j}}\sum_{t=1}^{\ell }e^{it\lambda _{j-k}} \\
E\left( \mathrm{Y}_{x,p}^{\left( 2\right) }\left( j\right) \mathcal{J}_{\chi
,q}\left( -k\right) \right) &=&\varphi _{x}\left( p,q\right) \mathbf{1}%
\left( j=k\right) o\left( T^{-2}\right) \text{.}
\end{eqnarray*}
\end{lemma}

\begin{proof}
As in the proof of Lemma \ref{z_12} we shall only show part $\left( \mathbf{a%
}\right) $. To that end, we first notice that Condition $C1$ implies that%
\begin{equation*}
E\left( \Upsilon _{\ell ,p}\left( j\right) \mathcal{J}_{\xi ,q}\left(
-k\right) \right) =\frac{\varphi _{u}\left( p,q\right) }{T^{1/2}}\left( 
\mathbf{1}\left( j=k\right) \mathbf{1}\left( \ell \geq T\right)
+\sum_{t=T-\ell +1}^{T}e^{it\lambda _{j-k}}\mathbf{1}\left( \ell <T\right)
\right) \text{.}
\end{equation*}%
From here the proof concludes by standard algebra.\hfill
\end{proof}

\begin{lemma}
\label{cum}Assuming $C1$ and $C2$, we have that 
\begin{eqnarray}
\left\vert cum\left( \xi _{p_{1}t};\xi _{p_{2}t};\xi _{p_{3}t};\xi
_{p_{4}t}\right) \right\vert &\leq &\left\vert \kappa _{4,\xi }\right\vert
\varphi _{u}\left( p_{1},p_{2}\right) \varphi _{u}\left( p_{3},p_{4}\right) 
\notag \\
\left\vert cum\left( \chi _{p_{1}t};\chi _{p_{2}t};\chi _{p_{3}t};\chi
_{p_{4}t}\right) \right\vert &\leq &\left\vert \kappa _{4,\chi }\right\vert
\varphi _{x}\left( p_{1},p_{2}\right) \varphi _{x}\left( p_{3},p_{4}\right)
\label{a_13}
\end{eqnarray}
\end{lemma}

\begin{proof}
Using inequality $\left( \ref{d1}\right) ,$ the proof follows easily since
by definition 
\begin{equation*}
cum\left( \xi _{p_{1}t};\xi _{p_{2}t};\xi _{p_{3}t};\xi _{p_{4}t}\right)
=\kappa _{4,\xi }\sum_{\ell =1}^{\infty }a_{\ell }\left( p_{1}\right)
a_{\ell }\left( p_{2}\right) a_{\ell }\left( p_{3}\right) a_{\ell }\left(
p_{4}\right) .
\end{equation*}%
The proof is similar for the second expression in $\left( \ref{a_13}\right) $%
, where inequality $\left( \ref{c1}\right) $ is used instead of $\left( \ref%
{d1}\right) $.
\end{proof}

\begin{lemma}
\label{cum_2}Assuming $C1$ and $C2$, for some $\tau >2$, 
\begin{eqnarray*}
\left\vert cum\left(
u_{p_{1}t_{1}};u_{p_{2}t_{2}};u_{p_{3}t_{3}};u_{p_{4}t_{4}}\right)
\right\vert &\leq &C\frac{\left\vert \kappa _{4,\xi }\right\vert \varphi
_{u}\left( p_{1},p_{2}\right) \varphi _{u}\left( p_{3},p_{4}\right) }{\left(
t_{2}-t_{1}\right) ^{\tau }\left( t_{3}-t_{1}\right) ^{\tau }\left(
t_{4}-t_{1}\right) ^{\tau }} \\
\left\vert cum\left(
x_{p_{1}t_{1}};x_{p_{2}t_{2}};x_{p_{3}t_{3}};x_{p_{4}t_{4}}\right)
\right\vert &\leq &C\frac{\left\vert \kappa _{4,\chi }\right\vert \varphi
_{x}\left( p_{1},p_{2}\right) \varphi _{x}\left( p_{3},p_{4}\right) }{\left(
t_{2}-t_{1}\right) ^{\tau }\left( t_{3}-t_{1}\right) ^{\tau }\left(
t_{4}-t_{1}\right) ^{\tau }}\text{.}
\end{eqnarray*}
\end{lemma}

\begin{proof}
As in the proof of Lemma \ref{cum}, we handle the first displayed inequality
only. Without loss of generality we take $t_{1}\leq t_{2}\leq t_{3}\leq
t_{4} $. Condition $C1$ and the definition of the fourth cumulant then yield
that 
\begin{eqnarray*}
cum\left( u_{p_{1}t_{1}};u_{p_{2}t_{2}};u_{p_{3}t_{3}};u_{p_{4}t_{4}}\right)
&=&\sum_{k=1}^{\infty }d_{k}\left( p_{1}\right) d_{k+t_{2}-t_{1}}\left(
p_{2}\right) d_{k+t_{3}-t_{1}}\left( p_{3}\right) d_{k+t_{4}-t_{1}}\left(
p_{4}\right) \\
&&\times cum\left( \xi _{p_{1}t};\xi _{p_{2}t};\xi _{p_{3}t};\xi
_{p_{4}t}\right) \text{.}
\end{eqnarray*}%
From here we conclude using Lemma \ref{cum} and the fact that Condition $C1$
implies that $\sup_{p}\left\vert d_{k}\left( p\right) \right\vert =O\left(
k^{-\tau }\right) $ for some $\tau >2$.
\end{proof}

\begin{lemma}
\label{B_5}Assuming $C1$ and $C2$, we have that for $w=:u$ or $x$, 
\begin{equation}
E\left( \mathcal{J}_{w,p_{1}}\left( j\right) \mathcal{J}_{w,p_{2}}\left(
-k\right) \right) =f_{w,p_{1}p_{2}}\left( j\right) \varphi _{w}\left(
p_{1},p_{2}\right) \left\{ \mathbf{1}\left( j=k\right) +\frac{C}{T}\right\}
\label{b_51}
\end{equation}%
and%
\begin{eqnarray}
&&E\left( \mathcal{J}_{w,p_{1}}\left( j\right) \mathcal{J}_{w,p_{2}}\left(
-j\right) \mathcal{J}_{w,p_{3}}\left( k\right) \mathcal{J}_{w,p_{4}}\left(
-k\right) \right)  \label{b_52} \\
&=&\varphi _{w}\left( p_{1},p_{2}\right) \varphi _{w}\left(
p_{3},p_{4}\right) \left\{ 1\mathbf{+1}\left( j=k\right) +\frac{C}{T}%
\right\} \text{.}  \notag
\end{eqnarray}
\end{lemma}

\begin{proof}
Consider $w=:u$, say. By $\left( \ref{bartlett}\right) $, we have that the
left hand side of $\left( \ref{b_51}\right) $ is 
\begin{equation*}
E\left( \left( \mathcal{B}_{u,p_{1}}\left( -j\right) \mathcal{J}_{\xi
,p_{1}}\left( j\right) +\mathrm{Y}_{u,p_{1}}\left( j\right) \right) \left( 
\mathcal{B}_{u,p_{2}}\left( k\right) \mathcal{J}_{\xi ,p_{2}}\left(
-k\right) +\mathrm{Y}_{u,p_{2}}\left( -k\right) \right) \right) ,
\end{equation*}%
which using $\left( \ref{dft_cov}\right) $ equals the right hand side of $%
\left( \ref{b_51}\right) $ by Lemmas \ref{z_12} and \ref{z_11}.

Next, the left hand side of $\left( \ref{b_52}\right) $ is%
\begin{eqnarray*}
&&E\left( \mathcal{J}_{u,p_{1}}\left( j\right) \mathcal{J}_{u,p_{2}}\left(
-j\right) \right) E\left( \mathcal{J}_{u,p_{3}}\left( k\right) \mathcal{J}%
_{u,p_{4}}\left( -k\right) \right) +E\left( \mathcal{J}_{u,p_{1}}\left(
j\right) \mathcal{J}_{u,p_{3}}\left( k\right) \right) E\left( \mathcal{J}%
_{u,p_{2}}\left( -j\right) \mathcal{J}_{u,p_{4}}\left( -k\right) \right) \\
&&+E\left( \mathcal{J}_{u,p_{1}}\left( j\right) \mathcal{J}_{u,p_{4}}\left(
-k\right) \right) E\left( \mathcal{J}_{u,p_{3}}\left( k\right) \mathcal{J}%
_{u,p_{2}}\left( -j\right) \right) +cum\left( \mathcal{J}_{u,p_{1}}\left(
j\right) ;\mathcal{J}_{u,p_{2}}\left( -j\right) ;\mathcal{J}_{u,p_{3}}\left(
k\right) ;\mathcal{J}_{u,p_{4}}\left( -k\right) \right) \text{.}
\end{eqnarray*}%
Using $\left( \ref{b_51}\right) ,$ the first three terms of the last
displayed expression are proportional to 
\begin{equation*}
f_{u,p_{1}p_{2}}\left( j\right) f_{u,p_{3}p_{4}}\left( j\right) \varphi
_{u}\left( p_{1},p_{2}\right) \varphi _{u}\left( p_{3},p_{4}\right) \mathbf{1%
}\left( j=k\right) ,
\end{equation*}%
while the absolute value of the last term is bounded by 
\begin{eqnarray*}
\frac{1}{T^{2}}\sum_{t_{1},t_{2},t_{3},t_{4}=1}^{T}\left\vert cum\left(
u_{p_{1}t_{1}};u_{p_{2}t_{2}};u_{p_{3}t_{3}};u_{p_{4}t_{4}}\right)
\right\vert &\leq &C\frac{\left\vert \kappa _{4,\xi }\right\vert }{T^{2}}%
\sum_{t_{1},t_{2},t_{3},t_{4}=1}^{T}\frac{\varphi _{u}\left(
p_{1},p_{2}\right) \varphi _{u}\left( p_{3},p_{4}\right) }{\left(
t_{2}-t_{1}\right) ^{\tau }\left( t_{3}-t_{1}\right) ^{\tau }\left(
t_{4}-t_{1}\right) ^{\tau }} \\
&\leq &\frac{C}{T}\varphi _{u}\left( p_{1},p_{2}\right) \varphi _{u}\left(
p_{3},p_{4}\right)
\end{eqnarray*}%
because $\tau >2$ using Lemma \ref{cum_2}. From here the conclusion follows
easily.
\end{proof}

\begin{lemma}
\label{z_13}Assuming $C2-C3$, we have that as $n,T$\textbf{\ }$\rightarrow
\infty $,%
\begin{equation}
E\left( \frac{1}{n}\sum_{p=1}^{n}\mathcal{I}_{x,p}\left( j\right)
-f_{x,p}\left( j\right) \right) ^{2}=o\left( 1\right) \text{.}
\label{lem3_1}
\end{equation}
\end{lemma}

\begin{proof}
Standard algebra yields that the left hand side of $\left( \ref{lem3_1}%
\right) $ is bounded by 
\begin{equation*}
E\left( \frac{1}{n}\sum_{p=1}^{n}\left\{ \mathcal{J}_{x,p}\left( j\right) 
\mathcal{J}_{x,p}^{\prime }\left( -j\right) -E\left( \mathcal{J}_{x,p}\left(
j\right) \mathcal{J}_{x,p}^{\prime }\left( -j\right) \right) \right\}
\right) ^{2}+\left( \frac{1}{n}\sum_{p=1}^{n}E\mathcal{I}_{x,p}\left(
j\right) -f_{x,p}\left( j\right) \right) ^{2}\text{.}
\end{equation*}%
Now $n^{-1}\sum_{p=1}^{n}E\mathcal{I}_{x,p}\left( j\right) -f_{x,p}\left(
j\right) =O\left( T^{-1}\right) $ is standard as $f_{x,p}\left( \lambda
\right) $ is twice continuously differentiable, whereas Lemma \ref{B_5}
implies that the first term of the last displayed expression is 
\begin{equation*}
\frac{C}{n^{2}}\sum_{p,q=1}^{n}\varphi _{x}^{2}\left( p,q\right) \left( 1+%
\frac{C}{T}\right) =o\left( 1\right)
\end{equation*}%
by Condition $C3$, see also Remark 1.
\end{proof}

\begin{lemma}
\label{Lem21}Under $C1-C3$, we have that as $n,T$\textbf{\ }$\rightarrow
\infty $,%
\begin{eqnarray}
\frac{1}{T}\sum_{j=1}^{T-1}\left( \frac{1}{n}\sum_{p=1}^{n}\mathcal{I}_{%
\widetilde{x},p}\left( j\right) \right) ^{2}-\left( \frac{1}{n}\sum_{p=1}^{n}%
\mathcal{I}_{x,p}\left( j\right) \right) ^{2} &=&o_{p}\left( 1\right)
\label{lem21_1} \\
\frac{1}{T}\sum_{j=1}^{T-1}\left( \frac{1}{n}\sum_{p=1}^{n}\mathcal{I}%
_{x,p}\left( j\right) \right) ^{2}-\int_{-\pi }^{\pi }\left(
\lim_{n\rightarrow \infty }\frac{1}{n}\sum_{p=1}^{n}f_{x,p}\left( \lambda
\right) \right) ^{2}d\lambda &=&o_{p}\left( 1\right) \text{.}
\label{lem21_2}
\end{eqnarray}
\end{lemma}

\begin{proof}
Noticing that 
\begin{equation*}
\frac{1}{n}\sum_{p=1}^{n}\mathcal{I}_{\widetilde{x},p}\left( j\right) -%
\mathcal{I}_{x,p}\left( j\right) =-\mathcal{I}_{\overline{x},\cdot }\left(
j\right) \text{,}
\end{equation*}%
we obtain that the left hand side of $\left( \ref{lem21_1}\right) $ equals%
\begin{equation*}
\frac{1}{T}\sum_{j=1}^{T-1}\mathcal{I}_{\overline{x},\cdot }^{2}\left(
j\right) -\frac{2}{T}\sum_{j=1}^{T-1}\mathcal{I}_{\overline{x},\cdot }\left(
j\right) \frac{1}{n}\sum_{p=1}^{n}\mathcal{I}_{x,p}\left( j\right) \text{.}
\end{equation*}%
We shall examine the first term of the last displayed expression, with the
second one being handled similarly, if not easier. Now, by definition 
\begin{equation*}
\mathcal{I}_{\overline{x},\cdot }\left( j\right) =\frac{1}{n^{2}}%
\sum_{p,q=1}^{n}\mathcal{J}_{x,p}\left( j\right) \mathcal{J}_{x,q}\left(
-j\right) \text{,}
\end{equation*}%
so that Lemma \ref{B_5}, in particular $\left( \ref{b_52}\right) $, implies
that 
\begin{equation*}
E\mathcal{I}_{\overline{x},\cdot }^{2}\left( j\right) =\frac{1}{n^{4}}%
\sum_{p_{1},...,p_{4}=1}^{n}\varphi _{x}\left( p_{1},p_{2}\right) \varphi
_{x}\left( p_{3},p_{4}\right) \left\{ 1\mathbf{+1}\left( j=k\right) +\frac{C%
}{T}\right\} =o\left( 1\right)
\end{equation*}%
because\textbf{\ }$n^{-2}\sum_{p_{1},p_{2}=1}^{n}\varphi _{x}\left(
p_{1},p_{2}\right) =o\left( 1\right) $\textbf{\ }by ergodicity. This
completes the proof of $\left( \ref{lem21_1}\right) $.

Regarding $\left( \ref{lem21_2}\right) $, it suffices to show that 
\begin{eqnarray}
\frac{1}{T}\sum_{j=1}^{T-1}\left( \frac{1}{n}\sum_{p=1}^{n}\mathcal{I}%
_{x,p}\left( j\right) -E\left( \mathcal{I}_{x,p}\left( j\right) \right)
\right) ^{2} &=&o_{p}\left( 1\right)  \label{lem21_3} \\
\frac{1}{T}\sum_{j=1}^{T-1}\left( \frac{1}{n}\sum_{p=1}^{n}\mathcal{I}%
_{x,p}\left( j\right) -E\left( \mathcal{I}_{x,p}\left( j\right) \right)
\right) \frac{1}{n}\sum_{p=1}^{n}E\left( \mathcal{I}_{x,p}\left( j\right)
\right) &=&o_{p}\left( 1\right) ,  \label{lem21_4}
\end{eqnarray}%
because the continuous differentiability of $f_{x,p}\left( \lambda \right) $
implies%
\begin{equation*}
\frac{1}{T}\sum_{j=1}^{T-1}\frac{1}{n}\sum_{p=1}^{n}E\left( \mathcal{I}%
_{x,p}\left( j\right) \right) -\int_{-\pi }^{\pi }\lim_{n\rightarrow \infty }%
\frac{1}{n}\sum_{p=1}^{n}f_{x,p}\left( \lambda \right) =o\left( 1\right)
\end{equation*}%
by standard arguments\textbf{.} Now $\left( \ref{lem21_3}\right) $ holds
true by Lemma \ref{z_13} and $\left( \ref{lem21_4}\right) $ follows by
Cauchy-Schwarz's inequality.
\end{proof}

The next lemma extends a Central Limit Theorem in Phillips and Moon $\left(
1999\right) $ when their independence condition fails.

\begin{lemma}
\label{Lem1App}Let $\left\{ u_{pt}\right\} _{t\in \mathbb{Z}}$ and $\left\{
x_{pt}\right\} _{t\in \mathbb{Z}}$, $p\in \mathbb{N}^{+}$, satisfy
Conditions $C1-C3$. Then as $n,T\rightarrow \infty $, 
\begin{equation}
\frac{1}{T^{1/2}}\sum_{t=1}^{T}\frac{1}{n^{1/2}}\sum_{p=1}^{n}x_{pt}u_{pt}%
\overset{d}{\rightarrow }\mathcal{N}\left( 0,\Phi \right) \text{.}
\label{Theo_1}
\end{equation}
\end{lemma}

\begin{proof}
First, Hidalgo and Schafgans' $\left( 2017\right) $ Theorem 1 implies that 
\begin{equation}
z_{n,t}=\frac{1}{n^{1/2}}\sum_{p=1}^{n}x_{pt}u_{pt}\overset{d}{\rightarrow }%
\mathcal{N}\left( 0,\Omega _{t}\right) \text{, }t=1,...,T\text{,}
\label{Theo_2}
\end{equation}%
and also for any $r,s\geq 0$, 
\begin{equation*}
\frac{1}{n^{1/2}}\sum_{p=1}^{n}\chi _{p,t+r}\xi _{p,t+s}\overset{d}{%
\rightarrow }\mathcal{N}\left( 0,\Omega _{t,r,s}\right) \text{.}
\end{equation*}%
Now, Phillips and Moon's $\left( 1999\right) $ Theorem 2 cannot be employed
as the latter result requires that the left hand side of $\left( \ref{Theo_2}%
\right) $, that is $\left\{ z_{n,t}\right\} _{t\geq 1}$, is a sequence of
independent random variables.

Dropping the subscript \textquotedblleft $p$\textquotedblright\ for
notational convenience, we have that 
\begin{equation}
u_{t}x_{t}=\left( D_{u}\left( L\right) \xi _{t}\right) \left( C_{x}\left(
L\right) \chi _{t}\right) \text{,}  \label{Theo_4}
\end{equation}%
where 
\begin{equation*}
D_{u}\left( L\right) =\sum_{\ell =0}^{\infty }d_{\ell }L^{\ell };\text{ \ \
\ \ }C_{x}\left( L\right) =\sum_{\ell =0}^{\infty }c_{\ell }L^{\ell }
\end{equation*}%
by Conditions $C1$ and $C2$. We now employ a \textquotedblleft
second-order\textquotedblright\ BN decomposition similar to that in Phillips
and Solo $\left( 1992\text{, p. 978-979}\right) $. First, we notice that
standard algebra yields that the right hand side of $\left( \ref{Theo_4}%
\right) $ is%
\begin{eqnarray*}
&&\sum_{\ell =0}^{\infty }d_{\ell }c_{\ell }\xi _{t-\ell }\chi _{t-\ell
}+\left( \sum_{\ell =0}^{\infty }\sum_{k=\ell +1}^{\infty
}+\sum_{k=0}^{\infty }\sum_{\ell =k+1}^{\infty }\right) d_{\ell }c_{k}\xi
_{t-\ell }\chi _{t-k} \\
&=&\sum_{\ell =0}^{\infty }d_{\ell }c_{\ell }\xi _{t-\ell }\chi _{t-\ell
}+\sum_{k=1}^{\infty }\left( \sum_{\ell =0}^{\infty }d_{\ell }c_{\ell +k}\xi
_{t-\ell }\chi _{t-k-\ell }\right) +\sum_{\ell =1}^{\infty }\left(
\sum_{k=0}^{\infty }c_{k}d_{k+\ell }\chi _{t-k}\xi _{t-k-\ell }\right) \\
&=&\sum_{\ell =0}^{\infty }d_{\ell }c_{\ell }\xi _{t-\ell }\chi _{t-\ell
}+\sum_{k=1}^{\infty }\left( \sum_{\ell =0}^{\infty }d_{\ell }c_{\ell
+k}L^{\ell }\right) \xi _{t}\chi _{t-k}+\sum_{\ell =1}^{\infty }\left(
\sum_{k=0}^{\infty }c_{k}d_{k+\ell }L^{k}\right) \chi _{t}\xi _{t-\ell } \\
&=&\varrho _{0}\left( L\right) \xi _{t}\chi _{t}+\sum_{k=1}^{\infty }\varrho
_{k}\left( L\right) \xi _{t}\chi _{t-k}+\sum_{\ell =1}^{\infty }g_{\ell
}\left( L\right) \chi _{t}\xi _{t-\ell }\text{,}
\end{eqnarray*}%
where $\varrho _{k}\left( L\right) =\sum_{\ell =0}^{\infty }d_{\ell }c_{\ell
+k}L^{\ell }$ and $g_{\ell }\left( L\right) =\sum_{k=0}^{\infty
}c_{k}d_{k+\ell }L^{k}$. Observe that $\varrho _{0}\left( L\right)
=g_{0}\left( L\right) $.

Next, because for a generic polynomial $h\left( L\right) =\sum_{\ell
=0}^{\infty }h_{\ell }L^{\ell }$, we have the identity $h\left( L\right)
=h\left( 1\right) -\left( 1-L\right) \widetilde{h}\left( L\right) $, where $%
\widetilde{h}\left( L\right) =\sum_{\ell =0}^{\infty }\widetilde{h}_{\ell
}L^{\ell }$ with $\widetilde{h}_{\ell }=\sum_{p=\ell +1}^{\infty }h_{p}$, we
can write the right hand side of the last displayed equality as%
\begin{eqnarray}
&&\varrho _{0}\left( 1\right) \xi _{t}\chi _{t}+\xi _{t}\sum_{k=1}^{\infty
}\varrho _{k}\left( 1\right) \chi _{t-k}+\chi _{t}\sum_{\ell =1}^{\infty
}g_{\ell }\left( 1\right) \xi _{t-\ell }  \label{Theo_16} \\
&&-\left( 1-L\right) \sum_{k=1}^{\infty }\widetilde{dc}_{k}\xi _{t-k}\chi
_{t-k}-\left( 1-L\right) \sum_{k=1}^{\infty }\widetilde{\varrho }_{k}\left(
L\right) \xi _{t}\chi _{t-k}-\left( 1-L\right) \sum_{\ell =1}^{\infty }%
\widetilde{g}_{\ell }\left( L\right) \chi _{t}\xi _{t-\ell }\text{.}  \notag
\end{eqnarray}%
Observe that 
\begin{eqnarray*}
\widetilde{dc}_{k} &=&\widetilde{\varrho }_{0}\left( L\right) \text{, \ \ }%
\widetilde{\varrho }_{k}\left( L\right) =\sum_{\ell =0}^{\infty }\widetilde{%
\upsilon }_{\ell ,k}L^{\ell }\text{ \ with }\widetilde{\upsilon }_{\ell
,k}=\sum_{p=\ell +1}^{\infty }d_{p}c_{p+k}, \\
\widetilde{g}_{\ell }\left( L\right) &=&\sum_{k=0}^{\infty }\widetilde{%
\omega }_{k,\ell }L^{\ell }\text{ \ with }\widetilde{\omega }_{k,\ell
}=\sum_{p=k+1}^{\infty }c_{p}d_{p+\ell },
\end{eqnarray*}%
and $\xi _{t}\sum_{k=1}^{\infty }\varrho _{k}\left( 1\right) \chi _{t-k}$
and $\chi _{t}\sum_{\ell =1}^{\infty }g_{\ell }\left( 1\right) \xi _{t-\ell
} $ are mutually independent martingale differences.

Given $\left( \ref{Theo_16}\right) ,$ we can write the left hand side of $%
\left( \ref{Theo_1}\right) $ as the sum of six terms. The contribution due
to the fourth term of $\left( \ref{Theo_16}\right) $ is 
\begin{equation*}
\sum_{k=1}^{\infty }\widetilde{dc}_{k}\frac{1}{T^{1/2}}\frac{1}{n^{1/2}}%
\sum_{p=1}^{n}\xi _{p,t-k}\chi _{p,t-k}=O_{p}\left( T^{-1/2}\right)
\end{equation*}%
because $E\left( n^{-1/2}\sum_{p=1}^{n}\xi _{p,t-k}\chi _{p,t-k}\right)
^{2}<C$ and by summability of the sequence $\left\{ \widetilde{dc}%
_{k}\right\} _{k\in \mathbb{N}^{+}}$. Next, the contribution due to the
fifth and sixth terms of $\left( \ref{Theo_16}\right) $ follow similarly and
hence they are $o_{p}\left( 1\right) $.

So, we need to examine the contribution due to the first three terms of $%
\left( \ref{Theo_16}\right) $ on the left side of $\left( \ref{Theo_1}%
\right) $, that is 
\begin{equation}
\frac{\varrho _{0}\left( 1\right) }{\left( Tn\right) ^{1/2}}%
\sum_{t=1}^{T}\sum_{p=1}^{n}\xi _{pt}\chi _{pt}+\frac{1}{\left( Tn\right)
^{1/2}}\sum_{t=1}^{T}\sum_{p=1}^{n}\xi _{pt}\widetilde{\chi }_{pt}+\frac{1}{%
\left( Tn\right) ^{1/2}}\sum_{t=1}^{T}\sum_{p=1}^{n}\widetilde{\xi }%
_{pt}\chi _{pt}\text{,}  \label{bn_1}
\end{equation}%
where%
\begin{equation*}
\widetilde{\chi }_{pt}=:\sum_{k=1}^{\infty }\varrho _{k}\left( 1\right) \chi
_{p,t-k};\text{ \ \ }\widetilde{\xi }_{pt}=:\sum_{\ell =1}^{\infty }g_{\ell
}\left( 1\right) \xi _{p,t-\ell }\text{.}
\end{equation*}

The result that the first term of $\left( \ref{bn_1}\right) $ converges to a
normal random variable follows by (the proof of) Hidalgo and Schafgans' $%
\left( 2017\right) $ Theorem 1 and Phillips and Moon's $\left( 2002\right) $
Theorem 2 as $n^{-1/2}\sum_{p=1}^{n}\xi _{pt}\chi _{pt}$ are independent
sequences in $t$. Because the second and third terms of $\left( \ref{bn_1}%
\right) $ are similar, we only handle the second one explicitly. Now, that
term is 
\begin{equation}
\sum_{k=1}^{K}\varrho _{k}\left( 1\right) \frac{1}{\left( Tn\right) ^{1/2}}%
\sum_{t=1}^{T}\sum_{p=1}^{n}\xi _{pt}\chi _{p,t-k}+\sum_{k=K+1}^{\infty
}\varrho _{k}\left( 1\right) \frac{1}{\left( Tn\right) ^{1/2}}%
\sum_{t=1}^{T}\sum_{p=1}^{n}\xi _{pt}\chi _{p,t-k}\text{.}  \label{bn_2}
\end{equation}%
By summability of $\varrho _{k}\left( 1\right) $ and given that%
\begin{equation*}
E\left( \frac{1}{\left( Tn\right) ^{1/2}}\sum_{t=1}^{T}\sum_{p=1}^{n}\xi
_{pt}\chi _{p,t-k}\right) ^{2}=\frac{1}{Tn}\sum_{t=1}^{T}\sum_{p,q}\varphi
\left( p,q\right) \leq C
\end{equation*}%
by Condition $C3$, we obtain that by choosing $K$ large enough the second
term of $\left( \ref{bn_2}\right) $ is $o_{p}\left( 1\right) $. The first
term of $\left( \ref{bn_2}\right) $ on the other hand converges to a normal
random variable proceeding as with the first term of $\left( \ref{bn_1}%
\right) $. The proof is then completed using Bernstein's lemma.
\end{proof}

\begin{lemma}
\label{Lem3App}Under the same conditions of Lemma \ref{Lem1App}, we have
that 
\begin{equation}
\frac{1}{\widetilde{T}^{1/2}}\sum_{j=1}^{\widetilde{T}}\frac{1}{n^{1/2}}%
\sum_{p=1}^{n}\mathcal{J}_{x,p}\left( j\right) \mathcal{J}_{u,p}\left(
-j\right) \overset{d}{\rightarrow }\mathcal{N}\left( 0,\Phi \right) \text{.}
\label{Lem3_1}
\end{equation}
\end{lemma}

\begin{proof}
Using $\left( \ref{bartlett}\right) $ and $\left( \ref{b_51}\right) $ of
Lemma \ref{B_5}, we have that the left side of $\left( \ref{Lem3_1}\right) $
is governed by%
\begin{eqnarray}
&&\frac{1}{\widetilde{T}^{1/2}}\sum_{j=1}^{\widetilde{T}}\frac{1}{n^{1/2}}%
\sum_{p=1}^{n}\mathcal{B}_{x,p}\left( j\right) \mathcal{B}_{u,p}\left(
-j\right) \mathcal{J}_{\chi ,p}\left( j\right) \mathcal{J}_{\xi ,p}\left(
-j\right)  \notag \\
&=&\frac{1}{\widetilde{T}^{1/2}}\sum_{j=1}^{\widetilde{T}}\frac{1}{T}%
\sum_{t,s=1}^{T}\Xi _{s,t}\left( n;j\right) e^{i\left( t-s\right) \lambda
_{j}}\text{,}  \label{Lem3_2}
\end{eqnarray}%
where 
\begin{equation}
\Xi _{s,t}\left( n;j\right) =\frac{1}{n^{1/2}}\sum_{p=1}^{n}\mathcal{G}%
_{p}\left( j\right) \chi _{ps}\xi _{pt};\text{ \ \ \ \ }\mathcal{G}%
_{p}\left( j\right) =:\mathcal{B}_{x,p}\left( j\right) \mathcal{B}%
_{u,p}\left( -j\right) \text{.}  \label{Lem4_1}
\end{equation}%
Because $\left\{ \chi _{pt}\right\} _{t\in \mathbb{Z}}$ and $\left\{ \xi
_{pt}\right\} _{t\in \mathbb{Z}}$, $p\in \mathbb{N}^{+}$,$\ $are mutually
independent $iid$ zero mean sequences, we have that $\Xi _{s,t}\left(
n\right) $ is independent of $\Xi _{r,m}\left( n\right) $ if $s\neq r$ and $%
t\neq m$ and uncorrelated if $s\neq r$ and $t=m$ or $s=r$ and $t\neq m$. By
Lemma \ref{Lem1App}, it follows that $\Xi _{s,t}\left( n;j\right)
\rightarrow _{d}\mathcal{N}\left( 0,\widetilde{\text{\textsl{V}}}\left(
j\right) \right) $, where 
\begin{equation*}
\widetilde{\text{\textsl{V}}}\left( j\right) =\lim_{n\rightarrow \infty }%
\frac{1}{n}\sum_{p,q=1}^{n}f_{x,pq}\left( j\right) f_{u,pq}\left( j\right)
\varphi \left( p,q\right)
\end{equation*}%
and $E\left\Vert \Xi _{s,t}\left( n\right) \right\Vert ^{4}<C$.

Next, the right hand side of $\left( \ref{Lem3_2}\right) $ is%
\begin{eqnarray}
&&\frac{2^{1/2}}{T^{3/2}}\sum_{t,s=1}^{T}\frac{1}{n^{1/2}}\sum_{p=1}^{n}\chi
_{ps}\xi _{pt}\left\{ \sum_{j=1}^{\widetilde{T}}g_{p}\left( j\right)
e^{i\left( t-s\right) \lambda _{j}}\right\}  \notag \\
&=&\frac{1}{T^{1/2}}\sum_{t,s=1}^{T}\frac{1}{n^{1/2}}\sum_{p=1}^{n}\phi
_{p}\left( t-s\right) \chi _{ps}\xi _{pt}\left( 1+\frac{C}{T}\right)
\label{Lem3_3}
\end{eqnarray}%
using Brillinger's $\left( 1981\right) $ Exercise 1.7.14(\emph{b}), where $%
\phi _{p}\left( s\right) $ denotes the $s-th$ Fourier coefficient of $%
g_{p}\left( \lambda _{j}\right) $ defined in $\left( \ref{Lem4_1}\right) $.
Note also that Parseval's equality, see Fuller's $\left( 1996\right) $
Theorem 3.1.6, implies that 
\begin{equation*}
\sum_{\ell =-\infty }^{\infty }\phi _{p}^{2}\left( \ell \right) =\frac{1}{2n}%
\int_{-\pi }^{\pi }g_{p}^{2}\left( \lambda \right) d\lambda =\frac{1}{2\pi }%
\int_{-\pi }^{\pi }f_{x,p}\left( \lambda \right) f_{u,p}\left( \lambda
\right) d\lambda \text{.}
\end{equation*}

Now, the right hand side of $\left( \ref{Lem3_3}\right) $ can be written as%
\begin{equation*}
\frac{1}{T^{1/2}}\sum_{t=1}^{T-\ell }\frac{1}{n^{1/2}}\sum_{p=1}^{n}\phi
_{p}\left( 0\right) \chi _{pt}\xi _{pt}+\frac{1}{T^{1/2}}\sum_{\ell
=1}^{T-1}\sum_{t=1}^{T-\ell }\frac{1}{n^{1/2}}\left\{ \sum_{p=1}^{n}\phi
_{p}\left( \ell \right) \left( \chi _{pt}\xi _{p,t+\ell }+\chi _{p,t+\ell
}\xi _{pt}\right) \right\} \text{.}
\end{equation*}%
From here, we conclude the proof proceeding as we did in Lemma \ref{Lem1App}
since, say,%
\begin{equation*}
\frac{1}{n^{1/2}}\sum_{p=1}^{n}\phi _{p}\left( \ell \right) \chi _{pt}\xi
_{p,t+\ell }
\end{equation*}%
is a sequence of independent random variables in the $t$ dimension which
converges to a Gaussian random variable by arguments similar to those in the
proof of Hidalgo and Schafgans' $\left( 2017\right) $ Theorem 1 and 
\begin{equation*}
\frac{1}{T^{1/2}}\sum_{\ell =b}^{T-1}\sum_{t=1}^{T-\ell }\frac{1}{n^{1/2}}%
\sum_{p=1}^{n}\phi _{p}\left( \ell \right) \chi _{pt}\xi _{p,t+\ell
}=o_{p}\left( 1\right)
\end{equation*}%
by choosing $b$ large enough since $\phi _{p}\left( \ell \right) =O\left(
\ell ^{-2}\right) $.\bigskip
\end{proof}

\end{document}